\documentclass[article]{elsarticle}

\makeatletter
\def\ps@pprintTitle{%
 \let\@oddhead\@empty
 \let\@evenhead\@empty
 \def\@oddfoot{}%
 \let\@evenfoot\@oddfoot}
\makeatother

\newcommand{\St}{\mathsf{St}}
\newcommand{\Eff}{\mathsf{Eff}}
\newcommand{\Transf}{\mathsf{Transf}}
\newcommand{\st}[1]{\mathbf{#1}}
\newcommand{\set}[1]{\mathsf{#1}}

\newcommand{\Tr}{\operatorname{Tr}}

\newcommand{\map}[1]{\mathcal{#1}}

\usepackage{amsmath,amssymb,amsfonts}
\usepackage{ifsym}
\usepackage{amsthm}
\usepackage{xypic}
\usepackage{graphicx}
\usepackage{dcolumn}
\usepackage{bm}
\usepackage{lineno,hyperref}
\modulolinenumbers[100]
\usepackage[matrix,frame,arrow]{xy}
\usepackage{amsmath}

\newcommand{\qw}[1][-1]{\ar @{-} [0,#1]}

\newcommand{\gate}[1]{*{\xy *+<.6em>{#1};p\save+LU;+RU **\dir{-}\restore\save+RU;+RD **\dir{-}\restore\save+RD;+LD **\dir{-}\restore\POS+LD;+LU **\dir{-}\endxy} \qw}

\newcommand{\measureD}[1]{*{\xy*+=+<.5em>{\vphantom{\rule{0em}{.1em}#1}}*\cir{r_l};p\save*!R{#1} \restore\save+UC;+UC-<.5em,0em>*!R{\hphantom{#1}}+L **\dir{-} \restore\save+DC;+DC-<.5em,0em>*!R{\hphantom{#1}}+L **\dir{-} \restore\POS+UC-<.5em,0em>*!R{\hphantom{#1}}+L;+DC-<.5em,0em>*!R{\hphantom{#1}}+L **\dir{-} \endxy} \qw}

\newcommand{\multimeasureD}[2]{*+<1em,.9em>{\hphantom{#2}}\save[0,0].[#1,0];p\save !C *{#2},p+LU+<0em,0em>;+RU+<-.8em,0em> **\dir{-}\restore\save +LD;+LU **\dir{-}\restore\save +LD;+RD-<.8em,0em> **\dir{-} \restore\save +RD+<0em,.8em>;+RU-<0em,.8em> **\dir{-} \restore \POS !UR*!UR{\cir<.9em>{r_d}};!DR*!DR{\cir<.9em>{d_l}}\restore \qw}

\newcommand{\multigate}[2]{*+<1em,.9em>{\hphantom{#2}} \qw \POS[0,0].[#1,0];p !C *{#2},p \save+LU;+RU **\dir{-}\restore\save+RU;+RD **\dir{-}\restore\save+RD;+LD **\dir{-}\restore\save+LD;+LU **\dir{-}\restore}
\newcommand{\ghost}[1]{*+<1em,.9em>{\hphantom{#1}} \qw}
\newcommand{\Qcircuit}[1][0em]{\xymatrix @*=<#1>} 

\newcommand{\pureghost}[1]{*+<1em,.9em>{\hphantom{#1}}}
\newcommand{\multiprepareC}[2]{*+<1em,.9em>{\hphantom{#2}}\save[0,0].[#1,0];p\save !C
  *{#2},p+RU+<0em,0em>;+LU+<+.8em,0em> **\dir{-}\restore\save +RD;+RU **\dir{-}\restore\save
  +RD;+LD+<.8em,0em> **\dir{-} \restore\save +LD+<0em,.8em>;+LU-<0em,.8em> **\dir{-} \restore \POS
  !UL*!UL{\cir<.9em>{u_r}};!DL*!DL{\cir<.9em>{l_u}}\restore}
\newcommand{\prepareC}[1]{*{\xy*+=+<.5em>{\vphantom{#1\rule{0em}{.1em}}}*\cir{l^r};p\save*!L{#1} \restore\save+UC;+UC+<.5em,0em>*!L{\hphantom{#1}}+R **\dir{-} \restore\save+DC;+DC+<.5em,0em>*!L{\hphantom{#1}}+R **\dir{-} \restore\POS+UC+<.5em,0em>*!L{\hphantom{#1}}+R;+DC+<.5em,0em>*!L{\hphantom{#1}}+R **\dir{-} \endxy}}
\newcommand{\poloFantasmaCn}[1]{{{}^{#1}_{\phantom{#1}}}}

\newtheorem{prop}{Proposition}
\newtheorem{example}{Example}
\newtheorem{cor}{Corollary}
\newtheorem{axiom}{Axiom}
\newtheorem{lemma}{Lemma}
\newtheorem{df}{Definition}
\newtheorem{theorem}{Theorem}
\newtheorem{assumption}{Assumption}

\begin{document}

\begin{frontmatter}

\title{Bridging the gap between general probabilistic theories and the device-independent framework for nonlocality and contextuality}
\author{Giulio Chiribella\footnote{Corresponding author:  giulio@cs.hku.hk}}
\address{Department of Computer Science, The University of Hong Kong, Pokfulam Road, Hong Kong}
\author{Xiao Yuan}
\address{Center for Quantum Information, Institute for Interdisciplinary Information Sciences, Tsinghua University, Beijing, 100084, China}




\begin{abstract}
Characterizing  quantum correlations in terms of information-theoretic principles  is a popular chapter of quantum foundations.    Traditionally, the principles  used for this scope   have been expressed in terms of conditional probability distributions, specifying the probability that a black box produces a certain output upon receiving a certain input.  
This approach is known as \emph{device-independent}. 
  Another major chapter of quantum foundations is the information-theoretic characterization of quantum theory as a whole, with its sets of states and measurements, and with its allowed dynamics.  The  different frameworks  adopted for this scope are known under the umbrella term of \emph{general probabilistic theories}.  With only a few exceptions, the two  research programmes on characterizing quantum correlations and characterizing quantum theory   
have so far proceeded  on separate tracks, each one developing its own methods and  its own agenda.  
Still,  both programmes share the same basic goal: a new and better understanding of  quantum mechanics in  information-theoretic terms.  This considered,  it is quite  striking that the connections between the two programmes are still largely undeveloped.  
This paper aims at  bridging the gap, by presenting a  ``Rosetta stone" for the two frameworks and by illustrating how the two programmes can benefit each other. 

As a case study, we focus on two device-independent features  known as Local Orthogonality (LO) and Consistent Exclusivity (CE).  In a recent work \cite{Chiribella14},  we showed that CE and LO  can be derived from the basic idea that, at the fundamental level,   measurements are repeatable and minimally disturbing.   In this paper we provide a new, alternative derivation  based on a different set of principles, revolving around the notion of \emph{pure orthogonal measurement}---a measurement that  cannot be further refined and that identifies states without error.  
The first principle, Measurement Purification, states that every measurement can be reduced to a pure orthogonal measurement by adding an auxiliary system and  by coarse-graining over  some outcomes.  
 The second principle, Locality of Pure Orthogonal Measurements, states that two pure orthogonal measurements performed independently on two systems yield a pure orthogonal measurement on the composite system.    The third  principle, Strong No Disturbance Without Information, states that every measurement that does not extract information about a  source can be realized without disturbing  the states in that source and without disturbing the pure orthogonal measurements that identify states in that source. These three principles together imply LO. CE is then derived by adding a fourth principle, called Pure State Identification, stating that every  outcome of a pure orthogonal measurement identifies a pure state. 
 \end{abstract}


\end{frontmatter}

\linenumbers

\section{Introduction}

One of the most profound mysteries of quantum theory is  nonlocality \cite{EPR,bell}, namely the fact that experiments performed at spacelike separated locations can exhibit stronger correlations than those allowed by any local realistic model.
 Still,  quantum correlations are not the strongest correlations one can imagine:  the assumption  that correlations cannot be used to communicate at unbounded speed, known as No-Signalling, is compatible with a larger set of exotic, non-quantum correlations  \cite{rastall1985locality,pr95}.  This observation stimulated the search for other principles, of similar information-theoretic flavour, aimed at achieving a complete characterization.    Up to now, several principles for quantum correlations have been proposed, such as Non-Trivial Communication Complexity \cite{Dam05, brassard2006}, No-Advantage in Nonlocal Computation \cite{linden2007}, Information Causality \cite{Paw09}, Macroscopic Locality \cite{Navascues09}, and Local Orthogonality (LO) \cite{Fritz12}.
These principles have been spectacularly successful in constraining the allowed correlations, narrowing them to a set that is close to the quantum set.   However,    no combination of the presently known principles  is  sufficient to characterize the quantum set completely \cite{almost}.   Similar considerations apply to the study of quantum contextuality \cite{kochen,spekkens2005contextuality}, where   the principles of Consistent Exclusivity (CE) \cite{cabello2013,henson2012,acin2012,cabello2012specker} and Macroscopic Non-Contextuality \cite{henson2015macroscopic} have been proposed in order to characterize the degree of contextuality exhibited by projective measurements in quantum theory. Also in this case, a complete information-theoretic characterization of the contextuality bounds satisfied by projective quantum measurements is still missing.

On the other hand, several reconstructions of quantum theory from information-theoretic principles  have been proposed in recent years
 \cite{deri, hardy11, masanes11,dakic11, hardy01, dariano, goyal10, masanes12,wilce2012conjugates}.
 With different background assumptions and slightly different goals, these works single out  the Hilbert space formalism of quantum theory: in particular, they imply that physical systems are associated to  Hilbert spaces, that states are described  by density matrices, and that the probabilities of measurement outcomes are computed with the Born rule.  As a byproduct,  they also characterize the particular sets of correlations arising in quantum theory.  Is this a solution to the  long sought-after  characterization of quantum nonlocality and contextuality? Yes and no. \emph{Yes}, because every set of information-theoretic principles that singles out quantum theory provides also an information-theoretic justification  of quantum nonlocality and contextuality.  \emph{And no}, because such a  justification may not be as satisfactory as one may desire:  ideally,   one would like to  have principles that \emph{directly} imply bounds on  correlations, without the detour of  a full derivation of the Hilbert space framework.   
 
The importance of a direct characterization is reflected  in the nature of the principles used for  quantum correlations. 
Principles like those in Refs.     \cite{khalfin1985quantum,rastall1985locality,pr95,Dam05, brassard2006, linden2007, Paw09,Navascues09,Fritz12,cabello2013,henson2012,acin2012,cabello2012specker}   refer only  to the conditional probabilities of obtaining  output data from  input data, without making any assumption on the process  generating the output from the input.  
The framework in which these principles are formulated  has been aptly named \emph{device-independent} (see  \cite{scarani} and \cite{bancal2013device} for an introduction). 
In stark contrast, the framework used for reconstructing quantum theory  is not device-independent. And for good reasons: a full-fledged physical theory is not only about input-output probability distributions, but also about physical systems and how they can interact through physical processes \cite{coecke2010universe}.    Naturally, the principles used for reconstructing quantum theory  presuppose that experimental data have already been organized  in the basic structure of a physical theory, which has systems, states, transformations, and measurements at its backbone. For example, principles like Local Tomography \cite{hardy01,mauro2, barrett07} or Ideal Compression \cite{deri}  explicitly refer to the ``states of a given physical system", to the ``measurements performed on a composite system", and to the ``processes that transform a system into another". The formulation of these principles is based on the framework of \emph{general probabilistic theories}  \cite{hardy01,barrett07,mauro2,nobroad,teleportation,puri,deri,barnum11,hardy11,hardy2013,chiribella14dilation}, which describes on the same footing classical and quantum theory, as well as many  hypothetical, post-quantum theories.



Up to now, the research programmes on reconstructing quantum theory and that on characterizing quantum correlations have proceeded along separate tracks.   However, it is clear that   the interaction between these two programmes  has a potential for understanding the picture of reality (if any) at which quantum theory is hinting.
  In a recent work \cite{Chiribella14},  we started exploring the relations between   principles for correlations  and principles  for general probabilistic theories.  We first defined a class of ideal measurements, called \emph{sharp measurements}, which represent an ideal standard of measurements that are repeatable and cause minimal disturbance on future observations. 
   Then we postulated that all measurements are fundamentally sharp, i.e. they can be obtained by performing a joint sharp measurement on the system together  and on an environment.    Combining  this requirement with two additional requests about the compositional properties of sharp measurements we have been able to derive the validity  of  Local Orthogonality   and Consistent Exclusivity.

In this paper we present an alternative derivation of LO and CE, based on a different notion of ``ideal measurement" and on a different set of physical principles.  We take LO and CE as the subject of this further study because they provide the simplest testbed for investigating  the interplay between  device-independent and device-dependent notions. 
We will first discuss the relations between  the device-independent framework and the framework of general probabilistic theories.
 Then we will set up the scene for the derivation of LO and CE, defining a privileged class of measurements, here called \emph{spiky measurements}.  Spiky measurements  are obtained by coarse-graining pure orthogonal measurements, i.e. measurements that cannot be further refined and that identify  some states without error.   In quantum theory, spiky measurements coincide with projective measurements, which in turn coincide with the sharp measurements defined in Ref. \cite{Chiribella14}.      In a general theory, however, spiky and sharp measurements can potentially differ.
Using the notion of pure orthogonal measurement, we then formulate three requirements which allow one to derive LO and, under an additional assumption, CE.  Interestingly, our requirements do not include Causality \cite{puri}, which is instead \emph{derived}  as a byproduct.    Nevertheless, one of the requirements, called Sufficient Orthogonality, has no immediate operational interpretation.  To address this issue, we show how Sufficient Orthogonality can be reduced to a strong version of the  No Information Without Disturbance  property discussed in Refs. \cite{deri,Pfister}. Such a reduction, alas, assumes Causality.  

The paper is structured as follows:  in Sections \ref{sec:devind} and \ref{sec:opprob} we present the device-independent framework and the framework of general probabilistic theories, respectively.  The bridge between the two frameworks is provided in Section \ref{sec:physimod}, where we specify  the physical model in which the input-output probabilities are generated.  The relation between No-Signalling at the level of probability distributions and Causality at the level of physical processes is discussed in Section \ref{sec:causnosig}.     Sections \ref{sec:spiky} and \ref{sec:ax} provide the definition of spiky measurements and state three axioms about their structure. The three axioms imply the validity of Local Orthogonality (briefly recalled in Section \ref{sec:lo}) and Causality, as shown by the derivation in Section \ref{sec:derivinglo}.   Since one of the three axioms has no compelling operational interpretation, in Section \ref{sec:derSO} we show one way to reduce it to more fundamental  physical statements---in this case, Causality and a strong version of the No Information Without Disturbance principle. 
The analysis carried out for nonlocality is then applied to the study of contextuality:  
 we  first review the device-independent framework and illustrate CE as an example of device-independent principle (Section \ref{sec:CE}) and then bridge it with the framework of general probabilistic theories  (Section \ref{sec:bridge}).  In Section \ref{sec:derivingCE} we show different formulations of  CE as a physical principle regarding a privileged class of measurements, which generalize projective measurements in quantum theory. Choosing spiky measurements as our privileged class of measurements, we then provide a derivation of CE (Section  \ref{sec:spikyCE}).     
Finally,   in Section \ref{sec:relations} we compare the notion of spiky measurement, used in this paper, with other potential generalizations of the notion of projective measurement in quantum theory.   The conclusions are drawn in   Section \ref{sec:conclu}.  
The Appendices present some technical proofs that are not of immediate interest for the comprehension of the main points of the paper.


\section{The device-independent framework for nonlocality}\label{sec:devind}
In this section we briefly review the device-independent framework for nonlocality \cite{pr95,barrett2005nosignaling,acin06}, pointing the reader to \cite{scarani} for a more extended  discussion.  
The  framework describes   games where a group of players respond to a set of possible questions  posed by a referee.   The strategy of the players is described by the conditional probability distribution of the answers  given the questions.  Regarding the questions as inputs and the answers as outputs,  the limitations on the physical theory that describes the players' strategies  are encoded into limitations on the  allowed  input-output probability distributions.

\subsection{Non-local  games}\label{subsect:nonlocal}
A non-local game is a game  involving  $N $ players and a referee, where  the referee gives to the $i$-th player an input $x_i$ in some input alphabet $\mathsf X_i$ and the  player returns an output $y_i$  in some output alphabet $\mathsf Y_i$.
For brevity, let us denote by $ \st x  =  (x_1,x_2,\dots, x_N)   \in  \set X_1  \times  \set X_2\times \cdots \times \set X_N   =  : \prod_{i=1}^N  \set X_i$ the string of all inputs given by the referee and by $ \st y  =  (y_1,y_2,\dots, y_N)   \in    \set Y_1  \times \set Y_2\times \cdots \times \set Y_N    =:  \prod_{i=1}^N  \set Y_i   $ the string of all outputs returned by the players. In each run of the game, the referee chooses the input string $\st x$ at random according to a probability distribution  $q(\mathbf{x})$ and  assigns a payoff  $\omega (\st x,\st y)$ to the output string $\st y$.  The goal of the players is to maximize their expected payoff, given by
\begin{align}\label{payoff}
\omega  =   \sum_{\st x}   q(\st x)   \left[ \sum_{\st y}  \omega (\st x, \st y)   ~ p(\st y| \st x)       \right]  \, ,
\end{align}
where $p(\st y| \st x)$ is the conditional probability that the players produce the output  $\st y $ upon receiving the input $\st x$.


\subsection{Principles about input-output distributions}

The input-output probability distribution $p(\st y|\st x)$ describes the strategy of the players in a black box fashion, disregarding the specific details of the devices used to generate the outputs. Such a description, called \emph{device-independent},  is particularly suited for  cryptographic applications \cite{ekert1991quantum,mayers1998quantum,barrett2005nosignaling,acin06,acin2007device,masanes2011secure,vazirani2014fully}.   In this context, the constraints on the allowed strategies are expressed as constraints on the allowed probability distributions
For example, the most common constraint  in the literature   is the No-Signalling principle \cite{khalfin1985quantum,rastall1985locality,pr95}, which imposes that the correlations  in the probability distribution $p(\st y|\st x)$ cannot be used to simulate classical communication among the players.    For $N=2$, No-Signalling amounts to the set of  linear constraints
\begin{align}
\nonumber
&  \sum_{y_2 \in\set Y_2}   p(y_1,y_2| x_1,x_2) =  \sum_{y_2 \in\set Y_2}   p(y_1,y_2| x_1,x'_2) \quad \forall x_2,x_2'  \in  \set X_2 \\
&  \sum_{y_1 \in\set Y_1}   p(y_1,y_2| x_1,x_2) =  \sum_{y_1 \in\set Y_1}   p(y_1,y_2| x'_1,x_2) \quad \forall x_1,x_1'  \in  \set X_1 \, .
 \label{nosi}
\end{align}
For $N  \ge 2$, no-signalling is imposed by partitioning the $N$ players into two disjoint groups and by imposing the above equations for all possible bipartitions.

Principles like  Non-Trivial Communication Complexity \cite{Dam05, brassard2006}, No-Advantage in Nonlocal Computation \cite{linden2007}, Information Causality \cite{Paw09},  Macroscopic Locality \cite{Navascues09}  and Local Orthogonality \cite{Fritz12} are also  examples of restrictions about input-output probability distributions.
For example, Non-Trivial Communication Complexity  is the requirement that the  probability distribution $p(\st y|\st x)$ should not allow two players to compute arbitrary Boolean functions with a single bit of classical communication.

Treating $p(\st y|\st x)$ as a black box also allows  for  an interesting  connection with
the framework of interactive proof systems  \cite{goldwasser1989knowledge}, as  highlighted in     Ref. \cite{brassard2008classical,brassard2013classical}.   In short, one regards the $N$ players as $N$ untrusted provers and the referee as a verifier, with the communication  between provers and verifier restricted to be classical.  In this context, different  physical principles   represent different constraints on the power of the provers.
Starting from  the seminal work by Raz \cite{raz1998parallel},  the no-signalling constraint  has been studied extensively
\cite{holenstein2007parallel,ito2009oracularization,ito2010polynomial,kalai2014delegate}. 
 It is natural to expect that  also other information-theoretic  principles, such as Non-Trivial Communication Complexity or Information Causality, may have interesting consequences for  interactive proof systems.

\subsection{Characterizing quantum correlations}

The original scope of nonlocal games is the study of quantum correlations.   Here one imagines a scenario where   $N$ parties  prepare $N$ quantum systems in a joint quantum state and the $i$-th party generates the output $y_i$ by performing a measurement on the $i$-th system, choosing  the measurement settings according to the input $x_i$.  In this scenario, the probability distribution $p(\st y|\st x)$
 has the form
\begin{align}\label{quantum}
p(\st y|\st x)  =  \Tr \left[ \left(  P^{(1,x_1)}_{y_1}  \otimes   P^{(2,x_2)}_{y_2}  \otimes \cdots \otimes   P^{(N,x_N)}_{y_N}  \right) \rho \right]  \, ,
\end{align}
where $\rho$ is the quantum state of the $N$ systems and $\left\{P^{(i,x_i)}_{y_i}\right\}_{y_i\in\set Y_i}$ is the Positive Operator-Valued Measure \footnote{We recall that a POVM with outcomes in $\set Y$ is defined as a collection of non-negative operators $  \{  P_y\}_{y\in\set Y}$ satisfying the normalization condition $\sum_{y\in\set Y}  P_y  =  I$,  $I$ being the identity on the system's Hilbert space.}  (POVM) describing the measurement performed by the $i$-th party upon receiving the input $x_i$.

Input-output distributions that are generated as in Eq. (\ref{quantum}) are called \emph{quantum}.
For given $N$ and given input/output alphabets, the  set of quantum input-output  distributions is convex. Hence,  characterizing it is equivalent to characterizing the maximum payoffs achieved by strategies of the form (\ref{quantum}) in all possible games.  Since  the  payoff $\omega$ in Eq. (\ref{payoff}) can be viewed as a \emph{correlation},  the problem of characterizing the maximum payoffs is often referred to as the problem of  \emph{characterizing the set of quantum correlations}.

Clearly, the definition of ``quantum input-output  distribution"  does refer to the way the probability distribution is  generated.   It does  in two ways:
\begin{enumerate}
 \item it prescribes that the $p(\st y|\st x)$ is generated by a specific operational procedure (preparing a multipartite state and performing local measurements)
 \item it specifies the physical theory (quantum theory) in which the procedure is implemented.
 \end{enumerate}
Now the   question  is:  Can we characterize the  set of quantum correlations though device-independent  principles?
The principles proposed so far  \cite{pr95,Dam05,brassard2006, linden2007,Paw09,Navascues09,Fritz12} are important milestones towards the achievement of this goal.    However, a complete characterization of the quantum set solely in terms of conditional probability  distributions appears to be challenging \cite{almost}.

\section{The framework of operational-probabilistic theories}\label{sec:opprob}

Ultimately, the maximum payoff that the players can win in a  non-local game depends on the physical theory that underlies its implementation.  Constraints on the physical theory imply  constraints on the conditional probability distributions $p(  \st y|  \st x)$ that the players can generate.   For example, the no-signalling conditions of Eq. (\ref{nosi})  are often  motivated by a space-time scenario where physical systems travel at a bounded speed and  the   players are  far enough from one another that no signal can be exchanged among them during a run of the game.

Among all  possible theories, classical and quantum theories are the two prominent examples, due to their central role in physics.  However, in order to understand what is specific about these two theories and to explore future generalizations, it is convenient to step back from their specific details and to place them in  the wider context of general probabilistic  theories \cite{hardy01,barrett07,mauro2,nobroad,teleportation,puri,deri,barnum11,hardy11,hardy2013,chiribella14dilation}---see also the contributed volume \cite{chiribella2016quantum} for an introduction to the different frameworks. 
Among the available frameworks, here we adopt  the framework of \emph{operational-probabilistic theories (OPTs)}  \cite{puri,deri,chiribella14dilation,quantumfromprinciples}, which extends the language of quantum circuits \cite{nielsen2010quantum,mermin93} to arbitrary physical theories, combining the categorical framework initiated by Abramsky and Coecke  \cite{abramsky2004,abramsky2008,coecke2010universe}  with the toolbox of elementary probability theory.    An informal summary of the OPT framework  is provided in the following subsections.   For a more formal exposition we direct the reader to Ref. \cite{chiribella14dilation}.   For more discussion on the physical assumptions at the basis of the OPT framework we recommend Hardy's recent works \cite{hardy11,hardy2013,hardybook}, which adopt a closely related framework and provide a number of enlightening comments on the relation between the operational  and the theoretical level.

\subsection{Operational structure}

An \emph{OPT} \cite{puri} describes the operations that an agent can perform on physical systems.  The theory specifies a catalog of (generally non-deterministic)  devices  that the agent can compose with each other: each device  transforms an input system into an output system, generally in a stochastic way, producing a random outcome  $x$.    We denote  by   $\boldsymbol {\map T}  =\{\mathcal{T}_x\}_{x\in \mathsf X}$ the set of alternative transformations that can occur when a given device is used, and we  represent  each transformation $\mathcal T_x$ as
\begin{equation*}
\Qcircuit @C=1em @R=.7em @! R {
&\qw \poloFantasmaCn{A}  &   \gate{\map T_x}  & \qw \poloFantasmaCn{B}   &  \qw  } \, ,
 \end{equation*}
 where $A$ and $B$ are the input and output system, respectively.  A  collection $\boldsymbol {\map T}$ describing the action of a device is called a \emph{test} \cite{puri}.  Whether or not a given collection $\boldsymbol {\map T}$ is a ``test" is determined by the theory \footnote{Essentially, the only constraints on the set of tests are those arising by coarse-graining and by composition   (see discussion later in this section) For example,   if two tests are composed in series or in parallel, then the resulting collection is also a test.  }.

A test with input $A$ and output $B$ is said to be \emph{of type $A\to B$}.       As a special case, the device can have no  input, 
 in which case its action  consists in preparing a system in a particular \emph{ensemble of  states} \footnote{In quantum theory, the ensemble  $\{\rho_x\}_{x\in\set X}$ would consist of unnormalized density matrices, with the trace of each matrix giving the probability of the corresponding preparation.   }   $\boldsymbol{ \rho}  =  \{\rho_x\}_{x\in\mathsf X}$. Each state of the ensemble is represented as
\begin{equation*}
\Qcircuit @C=1em @R=.7em @! R {
 &   \prepareC{\rho_x}  & \qw \poloFantasmaCn{B}   &  \qw  } \, ,
 \end{equation*}
where $B$ is the  system prepared by the device. In equations, we will often use the Dirac-like notation $  |\rho_x)$.
Likewise,  a device can have trivial output, in which case its action results in  a \emph{demolition measurement}  $\st m = \{m_x\}_{x\in\set X}$, that absorbs the system and produces an outcome with some probability. 
We represent each transformation in the measurement as
\begin{equation*}
\Qcircuit @C=1em @R=.7em @! R {
&\qw \poloFantasmaCn{A}  &   \measureD{m_x} } \, ,
 \end{equation*}
where $A$ is the input state undergoing the measurement. In equations we will often use the Dirac-like notation $(m_x|$.      Traditionally, the transformation $m_x$ is called \emph{effect}  \cite{ludwig1968attempt}.
A  test  $\boldsymbol {\map T}$ of type  $A \to A$ can be thought as a \emph{non-demolition measurement} of system $A$.
We will use the notation $\St (A)$, $\Transf (A\to B)$, 
and $\Eff (B)$ to denote the sets of all states of system $A$, all transformations of $A$ into $B$, and all effects on system $B$, respectively.

The simplest device that can act on a system $A$ is the identity device, which has only one possible outcome, corresponding to the identity transformation, $\map I_A$. Like in  quantum circuits, we  represent the identity on system $A$ with just a wire.
In general, we call a device with a single outcome \emph{deterministic}, because in that case we know for sure which transformation is going to take place.   The subsets consisting  of deterministic states, deterministic transformations, and deterministic effects will be denoted as $\St_1 (A)$,  $\Transf_1 (A\to B)$, and $\Eff_1(B)$, respectively.

The notation $A\otimes B$ represents the composite system consisting of the subsystems $A$ and $B$. 
Composite systems are represented  by multiple wires: for example,
\begin{equation*}
\begin{aligned}
\Qcircuit @C=1em @R=.7em @! R {
 &   \multiprepareC{1}{\rho_x}  & \qw \poloFantasmaCn{A}   &  \qw  \\
 &\pureghost{\rho_x}  &  \qw \poloFantasmaCn{B}  &\qw }
 \end{aligned} \, ,
 \end{equation*}
represents a state of the composite system $A \otimes B$. Devices can be connected in parallel and in series, giving rise to circuits, such as
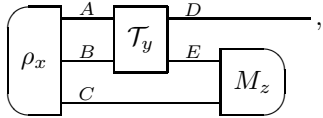
\begin{equation*}
\Qcircuit @C=1em @R=.7em @! R {
 &   \multiprepareC{2}{\rho_x}  & \qw \poloFantasmaCn{A}   &  \multigate{1}{\map T_y}  &  \qw \poloFantasmaCn{D}  & \qw & \qw  \\
& \pureghost{\rho_x}   &  \qw \poloFantasmaCn{B}  &  \ghost{\map T_y}  &  \qw\poloFantasmaCn{E} & \multimeasureD{1}{M_z}   &  \\
 &\pureghost{\rho_x}   &  \qw \poloFantasmaCn{C}   &  \qw    & \qw & \ghost{ M_z }  &  } \, ,
 \end{equation*}
 or, in equation  $       (  \map I_{\map D}  \otimes  M_z ) (    \map T_y  \otimes  \map I_C    )\rho_x$.
Circuits in an operational-probabilistic theory obey the same rules as   circuits in quantum information.
In fact, these  rules  are already encapsulated in the graphical language used to represent them, whose foundation lies in the theory of strict symmetric monoidal categories \cite{coecke2010, selinger2011survey}.  The idea that the definition of a physical theory should be based on (strict) symmetric monoidal categories was introduced by Abramsky and Coecke   \cite{abramsky2004,abramsky2008}.   A discussion of this idea, along with a comprehensive exposition  of the categorical framework can be found in   Coecke's review  \cite{coecke2010universe}.

\subsection{Probabilistic structure}
When a preparation device with  ensemble $\{\rho_x\}_{x\in\mathsf X}$ is connected to a measurement device $\{m_y\}_{y\in\mathsf Y}$, the joint probability distribution of the outcomes is written as
\begin{equation}\label{pjoint}
p(x,y)   =   (  m_y|  \rho_x) \, ,
 \end{equation}
 and is identified with the diagram \begin{equation*}
  \begin{aligned}\Qcircuit @C=1em @R=.7em @! R {
 &  \prepareC {\rho_x}   & \qw \poloFantasmaCn{A}  &   \measureD{m_y} }
\end{aligned} \, .
 \end{equation*}
Note that this is the \emph{joint} probability distribution that the preparation device gives the random outcome $x$ \emph{and} the measurement device gives the outcome $y$.   Accordingly, it is   
is normalized as
 $$\sum_{x\in\set X} \sum_{y\in \set Y}   p(x,y) = 1  \, .$$
 Once probabilities are introduced, the sets of states, effects, and transformations inherit a linear structure\footnote{The linear structure is obtained through an operation of quotient, which consists in identifying transformations that give the same probabilities in all possible circuits. Note that the assumption of convexity   is not made here: the OPT framework can also be used to describe theories with non-convex state spaces, such as Spekkens' toy theory \cite{SpekkensToy}.}, so that we can think of each state, effect, or transformation as an element of a suitable vector space \cite{puri,chiribella14dilation}.  By construction, the action of a transformation on states and effects is linear and, in particular, states (effects) are linear functionals on effects (states) (see  paragraph  II.F of Ref.  \cite{puri} and paragraph 2.3.3 of Ref. \cite{chiribella14dilation} for  details).

Quantum theory can be cast in the framework of OPTs as a special example.  Here systems are described by Hilbert spaces.  A preparation device is described by an ensemble $\{\rho_x\}_{x\in\mathsf X}$ of unnormalized density matrices, acting on the system's Hilbert space and satisfying the condition $\sum_x \Tr  [\rho_x] =1$.   A measurement device is described by a positive operator-valued measure (POVM),  namely a collection $\{  P_y  \}_{y\in  \mathsf Y}$ of non-negative operators satisfying the condition
\begin{align}\label{norm:povm} \sum_{y}  P_y  =  I \, ,
\end{align}
where $I$ is the identity operator on the system's Hilbert space.  The pairing between states and effects is given by the Born rule  
\begin{align}
(  P_y|   \rho_x )   :=  \Tr  [  P_y  \,  \rho _x] \, .    
\end{align}
A test with non-trivial input and output is a \emph{quantum instrument} \cite{davies1970operational},~i.~e.~a collection of completely positive, trace non-increasing linear maps $\{\map T_y\}_{y\in\set Y}$, transforming operators on the input system's Hilbert space into operators on the output system's Hilbert space and satisfying the condition that the map $\sum_{y\in\set Y}  \map T_y$ is trace-preserving.   Classical theory can also be represented in this way, by choosing density matrices and POVM operators that are diagonal in a fixed   basis, and quantum instruments that transform diagonal operators into diagonal operators.  

\subsection{Coarse-graining}
A key notion that comes with the probabilistic structure is the notion of \emph{coarse-graining}:  for a test $\boldsymbol{\map T}  =\{\map T_y\}_{y\in \set Y}$ one can decide to identify some outcomes, thus obtaining another, coarse-grained test.  Mathematically, a coarse-graining is defined by a partition of the outcome set $\set Y$ into mutually disjoint subsets $\{  \set Y_z\}_{z\in\set Z}$.      The coarse-grained test is the test $\boldsymbol{\map T}'  = \{\map T'_z\}_{z\in\set Z}$ defined by
\begin{align}
\map T'_z  :  =  \sum_{y\in\set Y_z}   \map T_y  \, .
\end{align}
Note that the summation is well-defined because  transformations  are elements of a vector space (cf. paragraph II.F of Ref.  \cite{puri} and   paragraph 2.3.3 of Ref. \cite{chiribella14dilation}).

\section{Physical modelling of non-local games}\label{sec:physimod}

The OPT framework   can be naturally applied to the study of nonlocal games. A strategy in a nonlocal game can be modelled as follows:
\begin{enumerate}
\item The correlations shared by the   $N$ players are modelled by a joint state $\rho$ of $N$ systems  $S_i$, $i  =1,\dots, N$, with system $S_i$ in possession of the $i$-th player.  Here we restrict the attention to states $\rho$ that can be prepared deterministically, that is, to states  generated by a preparation device with only one possible outcome.
\item Upon receiving the input $x_i$ from the referee, the $i$-th player will produce an output by performing a measurement on system $S_i$.   Note that in this broad context, ``measurement" can be any process that produces a classical output $y_i$ given the input $x_i$ and the state of the system.
Even evaluating a function of $x_i$ on a computer and reading the result on the screen would count as a ``measurement".
\end{enumerate}
Let us denote by ${\st m}^{i,x_i}  :  =  \left\{  m_{y_i}^{i,x_i}\right\}_{y_i \in\mathsf Y_i}$ the measurement performed by the $i$-th player upon receiving input $x_i$.
The conditional probability distribution $p(\st y|\st x)$, generated by the measurements of all players is then given by
\begin{equation}\label{physprob}
p(\st y|\st x)  =  \left. \left (   m_{y_1}^{ 1,x_1}  \otimes  m_{y_2}^{2,x_2}   \otimes \dots \otimes m_{y_N}^{N,x_N}   \right|  \rho  \right) \, ,
\end{equation}
and corresponds to the diagram
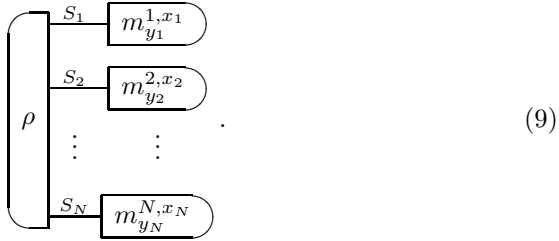
\begin{equation}\label{pdiag}
 \begin{aligned}
 \Qcircuit @C=1em @R=.7em @! R {
 &  \multiprepareC{3}{\rho}   & \qw \poloFantasmaCn{S_1}  &   \measureD{m_{y_1}^{1,x_1}} \\
  &  \pureghost{\rho}   & \qw \poloFantasmaCn{S_2}  &   \measureD{m_{y_2}^{2,x_2}} \\
   &  \pureghost{\rho}   & \vdots  &     \vdots    \\ &  \pureghost{\rho}   & \qw \poloFantasmaCn{S_N}  &   \measureD{m_{y_N}^{N,x_N}}   }
\end{aligned} \, .
\end{equation}
For brevity, we will often use the notation $m_{\st y}^{\st x}$ to denote the product effect
\begin{align}
m_{\st y}^{\st x}  : = m_{y_1}^{ 1,x_1}  \otimes  m_{y_2}^{2,x_2}   \otimes \dots \otimes m_{y_N}^{N,x_N}   \, .
\end{align}
Accordingly,  we will write Eq. (\ref{physprob}) in the compact form 
\begin{equation}\label{physprob1}
p(\st y|\st x)  =  \left. \left (   m_{\st y}^{ \st x}  \right|  \rho  \right) \, ,
\end{equation}

Once a physical theory  has been specified, the goal of the players is to find the best  state $\rho$ and the best measurements that maximize the expected payoff $\omega$, given by Eq.~\eqref{payoff}. For a given theory $\mathbb T$, we denote by $\omega_{\mathbb T}$ the maximum payoff that can be obtained by optimizing over all possible states and measurements allowed  in $\mathbb T$.

\section{Causality, no-signalling, and conditional tests}\label{sec:causnosig}
In general, the probability distribution $p(\st y|\st x)$ in Eq. (\ref{physprob}) can allow for signalling.
In the framework of operational-probabilistic theories, No-Signalling  is imposed by the Causality principle, stating that the probability of an outcome at a given step in a circuit is independent of the choice of tests performed at later steps.  Precisely, the principle can be stated as follows:
\begin{df}[Causality  \cite{puri,deri}] A theory satisfies \emph{causality} iff for every system $S$, every preparation-test $\boldsymbol{\rho}  =\{\rho_x\}_{x\in \set X} $ for system $S$,  and   every two measurements $\st m^0  =\left \{ m^{0}_{y_0}\right\}_{y_0\in \mathsf Y_0}$ and $\st m^1  = \left\{ m^1_{y_1}\right\}_{{y_1}\in \mathsf Y_1}$ on system $S$ the conditional probability distributions  
$p(x,y_z|z)   :  =   \left(  \left.   m^{z}_{y_z}      \right|   \rho_x  \right) $   
satisfy the condition 
\begin{align}\label{NSF}
\sum_{y_0\in\set Y_0}    p(x,y_0| 0 )    =  \sum_{y_1\in\set Y_1}    p(x,y_1| 1 )     \qquad \forall x\in\set X \, .
\end{align}
\end{df}
Informally, Eq. (\ref{NSF}) expresses a condition of No-Signalling from the future: the (marginal) probability of a preparation does not depend on the choice of measurement.     

Causality is   equivalent to the requirement that for every system $S$  there exists a unique effect   $u_S$, called the \emph{unit effect}
\footnote{We  adopt this terminology to facilitate the comparison of our framework with the convex set framework \cite{hardy01, barrett07, barnum11}, where the existence and uniqueness of the unit effect---and therefore the validity of Causality---is built in. },
such that
\begin{equation}\label{causal} \sum_{y\in\set Y}   m_y  =  u_S
\end{equation}
for every measurement   $\{m_y\}_{y\in\set Y}$   on $S$.  When there is no ambiguity, we will drop the subscript from $u_S$.
 In quantum theory,   $u_S$ is the identity operator on the Hilbert space of the system and Eq. (\ref{causal})  expresses the fact that quantum measurements are resolutions of the identity [cf. Eq. (\ref{norm:povm})].

 Causality is equivalent to the statement that for every system  $S$ there exists a unique \emph{deterministic} effect $u_S  \in  \Eff_1  (S)$ \cite{puri}.  In categorical terms, this condition is the terminality of the tensor unit (the trivial system, in our language) and defines a special class of categories called \emph{causal categories}   \cite{coecke2013causal, coecke2014terminality}.  

\subsection{Causality and No-Signalling}
Causality  implies that the probability distributions  $p(\st y|\st x)$ generated by local measurements as in Eq. (\ref{physprob}) satisfy the no-signalling condition   (cf. theorem 1 of Ref. \cite{puri} and theorem 5.1 of  Ref.  \cite{coecke2014terminality}).  In fact, under a minimalistic assumption, Causality is \emph{equivalent} to the request that all the   probability distributions of the form of Eq. (\ref{pdiag}) are no-signalling.
The assumption is that every ensemble of states can be generated by performing a measurement on one side of a bipartite state:
\begin{assumption}[cf. Axiom 2 of \cite{chiribella14dilation}]\label{ass:display}
For every system $A$ and for every ensemble $\{\rho_x\}_{x\in \set X}$, describing a random preparation of $A$, there exists a system $B$, a deterministic state $\sigma \in \St_1  (A\otimes B)$ and a measurement $\{b_x\}_{x\in \set X}$ such that
\begin{equation}\label{display}
\begin{aligned}
\Qcircuit @C=1em @R=.7em @! R {
 &   \prepareC{\rho_x}  & \qw \poloFantasmaCn{A}   &  \qw  }
 \end{aligned}   =
\begin{aligned}
\Qcircuit @C=1em @R=.7em @! R {
 &   \multiprepareC{1}{\sigma}  & \qw \poloFantasmaCn{A}   &  \qw  \\
 &\pureghost{\sigma}  &  \qw \poloFantasmaCn{B}  &\measureD{b_x} }
 \end{aligned}   \qquad \forall x\in\set X \,  .
 \end{equation}
\end{assumption}
This assumption is so natural that could even be included in the definition of OPT:     indeed,  one can think of system  $B$  in Eq. (\ref{display}) as the physical support that carries the classical information about the outcome $x$---information which is read-out  by performing the measurement $\{b_x\}_{x\in\set X}$.   If   Eq. (\ref{display})  were not to hold, we could not represent the outcome   $x$  as information carried by an actual physical system.   Note that this observation   applies not only to ensembles of states, but also to generic tests with non-trivial input and non-trivial output.

Under the validity of Assumption \ref{ass:display}, Causality and No-Signalling are equivalent: 
\begin{prop}\label{prop:causnosig}
For every theory satisfying Assumption \ref{ass:display}, the following conditions are equivalent
\begin{enumerate}
\item the theory is causal
\item every input-out probability distribution $p(\st y|\st x)$ generated as in Eq. (\ref{physprob}) is no-signalling.
\end{enumerate}
\end{prop}
The proof is rather elementary and is provided in  \ref{app:causnosig}.    As a consequence,  maximizing the payoff of a nonlocal game over all possible theories that satisfy  Causality is equivalent to maximizing the payoff $\omega$ in Eq.~\eqref{payoff} over all possible conditional distributions that satisfy No-Signalling.
The relation between Causality and No-Signalling has recently played an important role in the study of  network scenarios inspired by Pearl's notion of causal networks \cite{fritz2014beyond,henson2014theory} and of the entropic relations implied by causal networks in operational-probabilistic theories \cite{chaves2015information}.

\subsection{Causality and conditional tests}
Thanks to Causality, one can  use the information gained in the past  to decide which tests are performed in the future, thus implementing \emph{conditional tests}.  Conditional tests are defined as follows:   If $\boldsymbol{\map T}  =\{\map T_x\}_{x\in \set X}$ is a test  with input $A$ and output  $B$ and,  for every $x$, $\boldsymbol{\map S}^x=\left\{\map S^{x}_y\right\}_{y\in\set Y}$ is a test with input $B$ and output $C$  for every $x$,  then the conditional test    $\left\{  \map S^{x}_y   \map T_x\right\}_{x\in\set X,  y\in\set Y}$ is the test that results from   performing  the test   $\boldsymbol{\map T} $ and, conditionally on outcome $x$, the test $\boldsymbol{\map S}^{x} $, as in the the diagram
\begin{equation*}
\Qcircuit @C=1em @R=.7em @! R {
&\qw \poloFantasmaCn{A}  &   \gate{\map T_x}  & \qw \poloFantasmaCn{B}   &  \gate{\map S_y^{x}}   & \qw\poloFantasmaCn{C}  &\qw} \, .
 \end{equation*}
Causality guarantees that the collection of transformations $\{  \map S^{x}_y   \map T_x\}_{x\in\set X,  y\in\set Y}$ can be included among the tests allowed by the theory  without generating contradictions \cite{puri}.
  Since they \emph{can} be included, one may as well assume that they \emph{are} included, which amounts to the following
\begin{assumption}\label{ass:display2}
 For every test  $\boldsymbol{\map T}  = \{\map T_x\}_{x\in \set X}$ of type $A\to B$   and  for every set  of tests $\boldsymbol{\map S}^x  = \left\{\map S^{x}_y\right\}_{y\in\set Y}$, $x\in\set X$, of type $B \to C$, the collection of transformations   $\left\{  \map S^{x}_y   \map T_x\right\}_{x\in\set X,  y\in\set Y}$   is a test of type  $A\to C$.  
\end{assumption}
Quite importantly, Assumption \ref{ass:display2} \emph{implies} Causality (cf. lemma 7 of Ref. \cite{puri}): only in a causal world the agent can freely choose future tests depending on the outcomes of previous ones.    
From now on, we will assume Assumption \ref{ass:display2} and convexity as part of the Causality package coming: by Causal theory, we will mean a theory satisfying Assumption \ref{ass:display2}.


\section{Spiky measurements}\label{sec:spiky}

In this section we define a privileged class of measurements, which  we call \emph{spiky measurements}.  In Quantum Theory,  spiky measurements coincide with projective measurements, i.e. measurements consisting of projectors on a complete set of orthogonal subspaces. 


\subsection{Purity}
 A \emph{pure transformation} $\map P$ is a transformation that cannot be obtained from the coarse-graining of two different transformations  $\map P_1$ and $\map P_2$: precisely,    the transformation $\map P$ is pure iff one has
\begin{align*}
\map P  =  \map P_1  +  \map P_2  \quad \Longrightarrow   \quad \map P_1 =  p  \map P  \,, \quad    \map P_2  =  (1-p) \map P    \,  ,  \quad   p \in  [0,1] ~.
\end{align*}
Intuitively, the pure transformations  are those for which the evolution of the system is known with the maximal accuracy allowed by the theory.  In quantum theory, the pure transformations are those with a single Kraus operator, i.e. those of the form $\map P (\rho)  =  M \rho M^\dag$, for some operator $M$ satisfying $M^\dag M \le I_{S}$, $I_{S}$ being the identity on the  system's Hilbert space.

As a particular case of pure transformations, one can consider pure states and pure effects.  A pure state is just a pure transformation with trivial input.   A pure effect is a pure transformation with trivial output.   In quantum theory, pure states  and pure effects are proportional to rank-one projectors.   Using the notion of pure effect, it is natural to define pure measurements:
\begin{df} We say that a measurement   $\st m$ is \emph{pure} iff it consists of pure effects.
\end{df}
Intuitively, a pure measurement extracts information in a way that cannot be further refined. For example, for a three-level quantum system,  the computational basis measurement $\{  |0\rangle\langle 0|  , |1\rangle\langle1| ,  |2\rangle\langle2| \}$ is pure, while the two-outcome projective  measurement $\{  |0\rangle\langle  0|  ,  |1\rangle\langle  1|  +  |2\rangle\langle 2|\}$ is not pure, since it can be obtained from the former by coarse-graining.

\subsection{Orthogonality}
In addition to purity, another desirable feature of measurements is orthogonality.
We say that a measurement is orthogonal if it can perfectly distinguish among the states in a given set:
\begin{df}
A measurement on system $S$, say $\st m=\{  m_y\}_{y\in\set Y}$, is \emph{orthogonal} iff there exists a set of states, say $\{  \rho_y\}_{y\in\set Y} $, such that
\begin{align}\label{orto}
(  m_y|  \rho_{y'})  = \delta_{y,y'}  \qquad \forall y,y'\in\set Y \, .
\end{align}
\end{df}
This notion of orthogonality can be easily extended to sets of effects  that do not necessarily form a measurement:
\begin{df}[Orthogonality of states and effects]\label{def:biort}
A    set of  effects $\{m_y\}_{y\in\set Y}\subset \Eff (S)$   and a set of   states $ \{\rho_y\}_{y  \in\set Y}  \subset \St (S)$ are \emph{biorthogonal} iff
\begin{equation*}
( m_y|  \rho_{y'}  ) =  \delta_{y, y'}
\end{equation*}
for every $y,y'\in\set Y$.   A set of   effects  $\{m_y\}$  is \emph{orthogonal} iff there exists a set of   states $\{  \rho_y\}$ such that the two sets are biorthogonal.  A set of states $\{ \rho_y\}$ is \emph{orthogonal}, iff there exists a set of effects  $\{m_y\}$ such that  the two sets  are biorthogonal.
\end{df}

The familiar example of Quantum Theory should not mislead the reader.  In this paper we do \emph{not} define  orthogonal states as states that can be perfectly distinguished by a measurement.  Distinguishability implies orthogonality, but in general the converse does  not hold:       if the states $\{\rho_y\}_{y\in\set Y}$ are orthogonal, this only means that there exist  effects  $\{m_y\}_{y\in\set Y}$ such that $(m_y|\rho_{y'})  =  \delta_{y,y'}$, but in general the effects $\{m_y\}_{y\in\set Y}$ may not form  a measurement  \footnote{Recall that the set of measurements is part of the specification of the theory.}.      Nevertheless,      orthogonality and distinguishability are equivalent notions for \emph{pairs} of states:
\begin{prop}\label{prop:orthotwo}
Two states $\rho_0$ and $\rho_1$ are orthogonal if and only if they are perfectly distinguishable.  
\end{prop}
The proof is  elementary and is provided in  \ref{app:orthotwo}.  The above proposition shows that orthogonality for pairs of states is a very special notion.   

Note that pairwise orthogonality does not imply orthogonality: The condition that two states $\rho_y$ and $\rho_{y'}$ are orthogonal for every $y,y'\in\set Y$, $y\not =  y'$ is not enough to guarantee that the states $\{\rho_y\}_{y\in\set Y}$ are orthogonal.  
The canonical counterexample is the \emph{square bit}  \cite{barnum2010entropy, janotta11},  discussed in the following: 

\begin{example}[The square bit]\label{squit}
Consider a physical system whose deterministic states form a  square.  Suppose that  the measurements are represented as positive affine functionals summing up to the functional that gives 1 on every point of the square  \footnote{For the purpose of this example,   we only need to declare the states and the effects of the system. We omit the specification of the full OPT in which the square bit lives.  As a matter of fact, there are many different OPTs that contain ``square bits" among their systems.  For example,  consider an OPT where the systems are composite systems of square bits, the states are convex combinations of product states,  the measurements are those that can be implemented by  (coarse-grainings of) measurements on individual square bits,  and the general tests are those of the ``measure-and-prepare" form, i.e. those consisting on measuring the input system and preparing an output state depending on the outcome of the measurement.}.

The square bit has  four pure states and four pure effects, given by the vectors
 \begin{align*}
|\varphi_y) =& \left(
\begin{array}{c}  r\cos  (2\pi y/4) \\
   r\sin (  2\pi y/4)  \\
    1
    \end{array}
    \right)  \\
 (a_y|   =   &\frac{1}{2} \left(  \begin{array}{c}
 r\cos [  (2y-1)\pi/ 4]  \\
  r\sin  [ (2y-1)\pi/4]
  \\ 1
  \end{array}
  \right)  \\
  &x= 1,2,3,4 \qquad r = 2^{1/4} \, ,
 \end{align*}
 respectively.
 The probabilities are given by the  scalar product of vectors, yielding
\begin{align}
\nonumber
(a_y|\varphi_y) &= (a_{y\oplus 1}|\varphi_y) =1  \\
\label{squitort}(a_{y\oplus 2}|\varphi_y)  & = (a_{y\oplus 3}|\varphi_y) =0\qquad \forall y = 1,2,3,4
\end{align}
where $\oplus $  denotes the addition modulo 4.   Here there are two pure measurements, namely $\{a_1,a_3\}$ and $\{a_2,a_4\}$.  Indeed, it is easy to check that $a_1  +  a_3  = a_2  +  a_4  =  u $, where $u$ is the deterministic effect
\begin{align*}
(u|    \equiv    \left (
\begin{array}{l}
0\\0\\1
\end{array}
\right)
\end{align*}
giving probability 1 on every pure state.  It is not hard to see that the four pure states $\{\varphi_y\}_{y=1}^4$ are pairwise orthogonal, but not orthogonal (ad, therefore, not perfectly distinguishable).  Similarly, the four effects $\{a_y\}_{y=1}^4$ are pairwise orthogonal, but not orthogonal.
\end{example}

Finally, note that  two orthogonal effects, as defined in Definition \ref{def:biort}, may not coexist in a measurement.    
An easy counterexample can be found in quantum theory.   Consider the two projectors $  P_1  =  |0\rangle\langle 0|  +  |1 \rangle\langle  1|  $  and $P_2  =  |0 \rangle\langle  0|  +  |2 \rangle\langle 2|$.  The two projectors correspond to orthogonal effects in the sense of Definition \ref{def:biort}: indeed, there exist two states $\rho_1  =  |1\rangle\langle   1|$  and $\rho_2  =  |2\rangle\langle  2|$ such that $  \Tr [   P_i  \, \rho_j]   =  \delta_{i,j}$, for every $i$ and $j$ in $\{1,2\}$.  However, $P_1$ and $P_2$ cannot coexist in the same measurement, because for $\rho_0  =  |0 \rangle\langle   0|$ one has $\Tr [P_1 \rho_0  ]   + \Tr [ P_2  \rho_0]  = 2$, in contradiction with the normalization of probabilities.

\subsection{Purity plus orthogonality}
We are now ready to define the notion of \emph{pure orthogonal measurement}:
\begin{df}\label{def:purorthmeas}
A {\em pure and orthogonal measurement} is an orthogonal measurement  consisting of pure effects. 
 \end{df}
In Quantum Theory, the pure orthogonal measurements are the measurements consisting on rank-one projectors on the vectors of an orthonormal basis.  Pure orthogonal measurements  featured in a recent work \cite{brunner2014dimension}, where the authors explored different inequivalent notions of dimension of a physical system.   In this work, the (maximum) number of outcomes in a pure orthogonal measurement was called the \emph{measurement dimension} of the system.

To appreciate the meaning of Definition \ref{def:purorthmeas} outside the quantum context, it is worth having a look at the square bit of Example \ref{squit}.  Here, each of  two pure measurements $\{a_1,a_3\}$ and $\{a_2,a_4\}$ is  orthogonal: for example,    $\{a_1,a_3\}$ allows one to distinguish perfectly between the two states $\rho_1$ and $\rho_3$ defined as
$$\rho_1   =   p \, \varphi_1  + (1-p)  \,   \varphi_4  \qquad  \rho_3   =   q \, \varphi_3  + (1-q)  \,   \varphi_2  \, ,$$
where $p$ and $q$ are arbitrary probabilities.   Note that here a pure effect  can give probability 1 on a mixed state. In this respect, the square bit differs radically from the quantum bit, where a pure effect can give probability 1 on one and only one  pure state. The one-to-one correspondence between pure states and effects is a non-trivial property, which played  an important role in several reconstructions of Quantum Theory \cite{PironBook,hardy11,deri,dakic11,masanes11,wilce2012conjugates,masanes12} and will also play a role in the present paper. 


\subsection{Spiky measurements}

We are now ready to define the set of spiky measurements:
\begin{df}[Spiky measurement]\label{def:spiky}   A measurement $\st m= \{m_y\}_{y\in\set Y}$ is \emph{spiky}
\footnote{Here the ``spikes''  are the  pure orthogonal effects $a_z $, $z\in\set Z$.}
iff it is the coarse-graining of a pure orthogonal measurement, i.e. iff there exists a pure orthogonal measurement $\st a= \{a_{z}\}_{z\in \set Z}$ and a partition of the outcome set $\set Z$ into disjoint subsets $  \left\{\set Z_y\right\}_{y\in\set Y}$ such that
\begin{align*}
(m_y | =  \sum_{z\in\set Z_y}  (a_z |.
\end{align*}
\end{df}

The above definition of ``spiky" measurements is equivalent to the definition of ``sharp" measurements   by Barnum, M\"uller, and Ududec \cite{barnum2014higher}.  In this paper  we prefer to avoid the  term ``sharp", because we would like to reserve it for measurements that are repeatable and minimally disturbing  \cite{Chiribella14},  this being a property   traditionally associated to sharp measurements  in  quantum theory   \cite{Luders50,Busch96}.  Admittedly, the choice of terminology is mostly  a matter of taste here, since  in quantum theory the two definitions coincide and single out set of projective quantum measurements.



\section{Axioms}\label{sec:ax}
Here we present three requirements about spiky measurements.    These three requirements are satisfied by both classical and quantum theory and imply the validity of Causality and   Local Orthogonality. 

\subsection{Measurement Purification}
 Measurement Purification is the statement  that every measurement can be reduced  to a  spiky measurement performed jointly on the system and on an environment:

\begin{axiom}[Measurement Purification]\label{mpuri}
For every system $S$ and for every  measurement on system $S$---say $\st m = \{m_y\}_{y\in\mathsf Y}$---there is another system $E$, a state $\sigma  \in  \St (E)$, and a spiky measurement $\st M  = \{  M_y\}_{y\in\set Y}$ such that
\begin{equation*}
\begin{aligned}
\Qcircuit @C=1em @R=.7em @! R {
   & \qw \poloFantasmaCn{S}   &  \measureD{m_y}  }
 \end{aligned}   =
\begin{aligned}
\Qcircuit @C=1em @R=.7em @! R {
 &  & \qw \poloFantasmaCn{S}   &  \multimeasureD{1}{M_y}  \\
& \prepareC{\sigma}&  \qw \poloFantasmaCn{E}  &\ghost{M_y} }
 \end{aligned}   \qquad \forall y\in\set Y \,  .
 \end{equation*}
\end{axiom}
Roughly speaking, one can think of the above axiom as an operational version of Naimark's theorem  for finite dimensional quantum systems  \cite{Holevo11,heinosaari2011mathematical}, which states  that every quantum measurement can be dilated to a projective measurement performed jointly on the system and on an environment  \footnote{Note that Naimark's theorem also includes the fact that the state of the environment  is \emph{pure} and that the dilation is \emph{unique}, up to partial isometries.  These two additional facts are also important, but not for the purpose of the present paper.}.

The idea that arbitrary measurements can be reduced to ideal measurements by introducing an environment immediately reminds of the Purification Principle \cite{puri,deri,chiribella2012quantum,ChiribellaXiao13,chiribella14dilation}, which states that arbitrary states can be reduced to pure states by adding an environment.  In this sense, the spirit of this paper is akin  to the ``purification philosophy" of Refs.  \cite{puri,deri,chiribella2012quantum,ChiribellaXiao13,chiribella14dilation}  namely the idea that all physical processes  can be reduced to ideal processes by including additional systems into  the description. The interaction with an environment is a powerful structure also in the abstract framework of categorical quantum mechanics \cite{coecke2008axiomatic,coecke2010environment}, where it leads to an axiomatization of  Selinger's CPM construction \cite{selinger2007dagger}.

We now give an elementary consequence of measurement purification that will be useful later:
\begin{lemma}\label{piccolo} Let $\left\{\st m^{x}\right\}_{x\in\set X}$ be a finite set of measurements labelled by a parameter $x\in\set X$.    If measurement purification holds, then  there exists a system $E$, a state $\sigma \in\St(E)$,  and  set of spiky measurements  $\left\{\st M^{x}\right\}_{x\in\set X}$ such that
\begin{equation*}
\begin{aligned}
\Qcircuit @C=1em @R=.7em @! R {
   & \qw \poloFantasmaCn{S}   &  \measureD{m^x_y}  }
 \end{aligned}   =
\begin{aligned}
\Qcircuit @C=1em @R=.7em @! R {
 &  & \qw \poloFantasmaCn{S}   &  \multimeasureD{1}{M^x_y}  \\
& \prepareC{\sigma}&  \qw \poloFantasmaCn{E}  &\ghost{M^x_y} }
 \end{aligned}   \qquad \forall y\in\set Y \,  \forall x\in\set X  \,  .
 \end{equation*}
\end{lemma}
The proof is provided in  \ref{app:mpuri}.  Compared to  the Measurement Purification axiom, the above lemma only adds the fact that the system $E$ and the  state $\sigma \in\St (E)$ can be chosen to be independent of the setting $x\in\set X$,  the dependence on the setting being only in the orthogonal measurement $\st M^{x}$.

\subsection{Locality of Pure Orthogonal Measurements}

Measurement Purification can be interpreted as the statement  that, at the most fundamental level, measurements are generated by the coarse-graining of pure orthogonal measurements.  If one pushes this requirement further, it is natural to ask that the product of two pure orthogonal measurements is a pure orthogonal measurement on the composite system:
\begin{axiom}[Locality of Pure Orthogonal Measurements]\label{loca}
If  $\st m  =\{m_x\}$ and $\st n =  \{n_y\}$ are pure orthogonal measurements on two systems $A$ and $B$, respectively, then their product $\st m\otimes \st n  =  \{  m_x\otimes n_y\}$ is a pure orthogonal measurement on the composite system $A\otimes B$.
\end{axiom}


Two comments are in order:
\begin{enumerate}
\item Locality of Pure Orthogonal Measurement may superficially look as a  consequence of the definition. And in part it is:  clearly, the product of two orthogonal measurements is orthogonal.  The part that does not follow from the definition is the purity of the  product effects. 
This condition would be guaranteed by Local Tomography  \cite{hardy01,mauro2, barrett07}, which is \emph{not} assumed here.  The Locality of Pure Orthogonal Measurements is much weaker condition than Local Tomography: for example, it is satisfied by Quantum Theory on real Hilbert spaces \cite{stueckelberg1960quantum,hardy2012limited,wootters2013optimal}, a well-known example of theory wherein  Local Tomography fails to hold.    
\item If one postulates Measurement Purification, then it is  natural to assume the Locality of Pure Orthogonal Measurements  as well.  Indeed,  suppose that two distant parties, Alice and Bob, perform two pure orthogonal measurements   $\st m$  and $\st n$ on their systems.  By Measurement  Purification, we know that the product measurement  $\st m\otimes \st n$  can be reduced to a pure orthogonal measurement---call it  $\st M$---performed jointly on $A$, $B$, and an environment $E$.  If the measurement $\st M$ could not be chosen of the product form, it would mean that the measurements that are performed independently by Alice and Bob would require some nonlocal interaction at the fundamental level.  
\end{enumerate}

\subsection{Sufficient Orthogonality}\label{sec:sufforth}

Here we introduce \emph{Sufficient Orthogonality (SO)}, a structural property of the measurements allowed by the theory.   We do not attach a particular operational meaning to this property, e.g.  we do not argue that this should be a fundamental principle of physics.  Nevertheless, we show that SO plays a  key role,  allowing  one to derive  LO and CE.   We view  SO as an intermediate step, which   can  be used to reduce LO and CE to other, more fundamental features of physical processes,   such as the \emph{Strong No Disturbance Without Information} principle discussed in Section \ref{sec:derSO}.


\begin{axiom}[Sufficient Orthogonality]\label{suffo}
Every set of pure orthogonal effects can coexist in a measurement, i.e.   for every set of pure orthogonal effects  $\{a_y\}_{y\in\set Y}$    there exists a measurement   $\st m$ such that $\{a_y\}_{y\in\set Y}  \subseteq \st m$. 
\end{axiom}


In the statement of SO it is essential that the effects $\{p_y\}$ are \emph{pure}, otherwise one can find counterexamples even in classical and quantum theory.  For example, the non-pure effects $P_1  =   |0\rangle\langle 0|  +      |1\rangle\langle 1| $ and     $P_2  =   |0\rangle\langle 0|  +      |2\rangle\langle 2| $ are orthogonal, because $\Tr[P_i \,\rho_j ]  =  \delta_{i,j}$ for $i,j=  1,2$ and $\rho_j  = |j\rangle\langle j|$.    However,  they cannot coexist in a quantum measurement, since $P_1+P_2  >  I$.

SO is  satisfied by both  classical and quantum probability theory, where a set of pure orthogonal effects is a set of rank-one projectors $\{P_y\}_{y\in\mathsf Y}$ satisfying $\Tr[P_y \, P_{y'}]  =  \delta_{y,y'}$.    Interestingly, SO is violated by the square bit of Example \ref{squit}.
More generally,  SO is violated by all systems whose state space is a regular polygon of $n>3$ vertices: 

\begin{example}[Regular polygons  \cite{janotta11}]
Consider an hypothetical physical system  $S_n$ whose deterministic states   $\St_1 (S_n)$ form a regular polygon of $n $ vertices.     The vertices of the polygon are the pure states and can be  represented by the real vectors
 \begin{align*}
|\varphi_y) = \left(
\begin{array}{c}  r_n\cos  \frac{2\pi y}  n \\
\\
   r_n\sin \frac{ 2\pi y}n   \\\\
    1
    \end{array}
    \right)    \qquad     r_n := \sqrt{\frac{1}{\cos{\left(\pi/n\right)}}}   \, ,
 \end{align*}
 for $y  =  0,1,\dots, n-1$.   Effects are also represented as real vectors, and the probability of an effect on a state is given by the scalar product.   The unit effect, which has probability $1$ on every pure state, is represented as
  \begin{align*}
(u|   =   &\left(  \begin{array}{c}
 0 \\
  0  \\
  1
  \end{array}
  \right).\\
 \end{align*}
For the measurements, one typically assumes the \emph{no-restriction hypothesis} (cf. Ref. \cite{puri}, definition 16,  and   Ref. \cite{Janotta13}, section III), which  states  that  all collections of positive affine functionals summing up to the unit  represent allowed measurements \footnote{We do not specify here the full OPT in which the regular polygon is included. In general,  it is easy to include a given system $S$, its  state space, and its  set of allowed measurements  into a full-blown OPT.    For example, one can consider the OPT where all systems consists of multiple copies of  $S$, the states are  product states (or convex combination thereof), the measurements are product measurements (or convex combinations thereof), and the general tests are of the measure-and-prepare form. Unless one imposes  additional physical constraints, the specification of the OPT in which a given system can be embedded is highly non unique.  Still, such a specification is irrelevant for the scopes of the present example.   }.   Under this hypothesis, one has the pure effects
 \begin{align*}
(a_y|   =   &\frac{1}{2} \left(  \begin{array}{c}
 r_n\cos  \frac{  (2y-1)\pi} n  \\  \\
  r_n\sin  \frac{(2y-1)\pi}n   \\ \\
  1
  \end{array}
  \right)
  \end{align*}
 for even $n$ and
 \begin{align*}
   (a_y|   =   &\frac{1}{r_n^2+1} \left(  \begin{array}{c}
 r_n\cos \frac{2y\pi} n  \\  \\
  r_n\sin  \frac{ 2y\pi}n   \\ \\
  1
  \end{array}
  \right)
  \end{align*}
    for odd $n$.    Note that for every $y$ one has   $(a_y|  \varphi_y)  =1$.

\medskip

For $n=3$, the effects $\{a_y\}_{y=1}^3$ can coexist in a measurement, and, therefore, all the three pure states are perfectly distinguishable. This is not surprising, because the triangle is a simplex, simplices represent the states of classical systems, and classical systems satisfy SO.

We now show that the triangle is the \emph{only} regular polygon satisfying SO.
Let us start from the case of even $n$. Here, the inner product between a  pure effect and a pure state is  given by
  \begin{align*}
(a_y|\varphi_{y'}) = \frac{1}{2}\left\{r_n^2\cos\left[\frac{(2y-2y'-1)\pi}{n}\right]+1\right\}
 \end{align*}
 and it is immediate to check that one has
 \begin{align*}
 &(a_y|\varphi_{y})=  (   a_{ y  \oplus 1 }|\varphi_y)   =   1  \\
 &(a_{y\oplus  \frac n2  }|\varphi_y)  =   (a_{y\oplus   \frac n2   \oplus 1 }|\varphi_y)    =  0   \qquad \forall y\in\{0,1,\dots , n-1\}
 \end{align*}
 where $\oplus$ denotes the addition modulo $n$.    Clearly, the pure effects $\left\{a_j,a_{j\oplus\frac n 2\oplus 1}\right\}$ are orthogonal, since they are biorthogonal  to the  pure states $\left\{\varphi_j,\varphi_{j\oplus \frac n2}\right\}$.     However, they cannot coexist in a measurement: by absurd, if they coexisted in a measurement, the total probability of the measurement outcomes on the state $|\varphi_{j\oplus \frac n2+1})$ would exceed one:
 \begin{align*}
 (a_j|\varphi_{j\oplus \frac n2\oplus 1})+(a_{j\oplus \frac n2\oplus 1}|\varphi_{j\oplus \frac n2\oplus 1})     =    (a_j|\varphi_{j\oplus \frac n2\oplus 1})+ 1
   >  1  \, .
 \end{align*}
 Hence, every polygon with even number of vertices violates SO.

For odd  number of vertices,  the inner products of a pure effect with a pure state is
  \begin{align}\label{forodd}
(a_y|\varphi_{y'}) = \frac{1}{r_n^2+1}\left\{r_n^2\cos\left[\frac{(2y-2y')\pi}{n}\right]+1\right\}
 \end{align}
and   one has
\begin{align*}
 &(a_y|\varphi_{y})  =   1  \\
 &(a_{y\oplus  \frac {n-1}2  }|\varphi_y)  =   (a_{y\oplus   \frac {n+1}2    }|\varphi_y)    =  0   \qquad \forall y\in\{0,1,\dots , n-1\} \, .
 \end{align*}
Clearly, the two effects $\left\{a_y,a_{y\oplus \frac{n+1}2}\right\}$ are orthogonal, as they are biorthogonal to the states $\left\{\varphi_y,\varphi_{y\oplus \frac{n+1} 2}\right\}$.    However, they cannot coexist in a measurement, because the sum of their probabilities on the state $|\varphi_{y\ominus 1})$ exceeds one, as shown in \ref{app:trigo}.   In summary, the only regular polygon  compatible with SO is the triangle, representing a three-level classical system.
\end{example}

\section{Local Orthogonality}\label{sec:lo}

Here we briefly discuss  LO,  a requirement on  the conditional probability distributions $p(\st y |\st x)$ generated by $N$ players of a non-local game.
To state the requirement, it is handy to introduce a notation for the output/input pairs $\mathbf{e} = (\st x,\st  y)$, which will be called \emph{events}.

\begin{df}[Locally orthogonal events]
Two events $\mathbf e =(\mathbf x,\mathbf y)$ and $\mathbf e'  =(\mathbf x',\mathbf y')$ are  \emph{locally orthogonal}, denoted as $\st e  \perp \st e'$, if there exists at least one party $i$ such that $x_i =  x_i'$ and $y_i\not=y_i'$.
A set  of events $\set O$ is \emph{locally orthogonal} if every two elements in $\set O$ are locally orthogonal.
\end{df}


\medskip

For an event $\st e  =  (\st x,\st y)$, we use the notation $p(\st e):  =  p(\st y|\st x)$. With this notation, LO is defined as follows
\begin{df}[Local Orthogonality  \cite{Fritz12,acin2012}]
A conditional probability distribution $p(\mathbf y|\mathbf x)$ satisfies {\em Local Orthogonality} iff one has
\begin{align}\label{lo}
\sum_{\st e  \in  \set O}  p  ( \st e)   \le 1 
\end{align}
for every locally orthogonal set  $\set O$.
   A theory  satisfies LO iff every probability distribution generated as in Eq. (\ref{physprob1}) satisfies LO.
\end{df}
In a bipartite setting, LO is equivalent to No-Signalling \cite{Fritz12,acin2012}. LO comes to its own in the multipartite setting, where Eq. (\ref{lo}) is more restrictive  than the No-Signalling condition.   In  a theory satisfying  LO, 
 the maximum payoff achievable by the players of a generic game is upper bounded as
\begin{equation}\label{LObound}
\omega_{\mathbb T}  \le \omega_{LO} \, ,
\end{equation}
where $\omega_{LO}$  denotes the maximum of the payoff  $\omega$ in Eq.~\eqref{payoff} over \emph{all} probability distributions $p(\st y|\st x)$ satisfying LO. 

Note that LO  has a slightly different flavour  from  other device-independent principles.  Indeed, principles like Nontrivial Communication Complexity, No Advantage in Nonlocal Computation and Information Causality are expressed as limitations about some distinguished information-theoretic task. Such  limitations   are subsequently  used to derive upper bounds on the payoffs of nonlocal games. Instead, LO is  \emph{defined} as  an upper bound on a payoff, as one can see by comparing the l.h.s. of Eq. (\ref{lo}) with the r.h.s. of  Eq. (\ref{payoff}).
The particular games that define the LO constraint have been characterized  in Ref. \cite{acin2012} and have been therein named \emph{maximally difficult guessing games}.       In this sense,  Eq. (\ref{LObound}) represents the upper bound on the payoff of a generic game under the condition that the payoff in  some privileged class of games is upper bounded as in Eq. (\ref{lo}).



LO  can be generalized to an infinite hierarchy of constraints \cite{Fritz12,acin2012}.  This is done as follows:  Suppose that the $N$ parties are given  $k$ copies of the black box generating  outputs according to the conditional probability distribution $p(\mathbf y|\mathbf x)$.  As a result, the overall input-output distribution will be given by
\begin{align*}
p^{\otimes k}(\mathbf y_1\mathbf y_2\dots\mathbf y_k|\mathbf x_1\mathbf x_2\dots\mathbf x_k) = p(\mathbf y_1|\mathbf x_1) \, p(\mathbf y_2|\mathbf x_2)\cdots p(\mathbf y_k|\mathbf x_k) \, .
\end{align*}
Defining the event $\st e_k = (\mathbf y_1\mathbf y_2\dots\mathbf y_k|\mathbf x_1\mathbf x_2\dots\mathbf x_k)$ and its probability $p^{\otimes k}(\st e_k): = p^{\otimes k}(\mathbf y_1\mathbf y_2\dots\mathbf y_k|\mathbf x_1\mathbf x_2\dots\mathbf x_k)$, one can formulate the $k$-th level of the LO hierarchy as 
\begin{df}[Local Orthogonality at the $k$-th level  \cite{Fritz12,acin2012}.]
A conditional probability distribution $p(\mathbf y|\mathbf x)$ \emph{satisfies LO} at the $k$-th level iff
\begin{align}\label{lok}
\sum_{\st e_k  \in  \set S_k}  p^{\otimes k}  ( \st e_k)   \le 1
\end{align}
for every locally orthogonal set  $\set S_k$.  A theory  satisfies LO at the $k$-the level iff every probability distribution generated as in Eq. (\ref{physprob1}) satisfies LO at the $k$-th level.
\end{df}
By increasing $k$, one gets more and more restrictive conditions on the probability distribution $p(\st y|\st x)$. For example, PR box correlations satisfy LO for $k=1$, but violate it  for $k \ge 2$ \cite{Fritz12}.

\section{Deriving Local Orthogonality and Causality}\label{sec:derivinglo}

We now  provide a derivation of LO  from  Measurement Purification, Locality of Pure Orthogonal Measurements, and Sufficient Orthogonality.  Since LO implies No-Signalling \cite{Fritz12,acin2012} and in our framework No-Signalling is equivalent to Causality (proposition \ref{prop:causnosig} of this paper),  our derivation of LO also amounts to a derivation of Causality.
The derivation consists of a few steps, discussed in the following paragraphs.


\subsection{Local Orthogonality for pure orthogonal measurements}

We start by showing the validity of LO for probability distributions generated by pure orthogonal measurements:
\begin{lemma}\label{lem:babyLO}
Let   $p(\st y|\st x)  =    \left (    m^{\st x}_{\st y} | \rho \right)$ be a set of probability distributions defined as in Eq.~\eqref{physprob1} with the product effect $m^{\st x}_{\st y}$ arising from a set of pure orthogonal measurements.    If the theory satisfies Locality of Pure Orthogonal Measurements and Sufficient Orthogonality,  then it $p(\st y|\st x)$ satisfies Local Orthogonality at all levels of the hierarchy.
\end{lemma}

 \begin{proof}
Firstly, we consider the proof for the first level.   For every locally orthogonal set of events $  \set O  $, we have to prove the relation    $\sum_{(\st x,\st y)  \in\set O}  p(\st y |\st x)\le 1$. The proof runs as follows:
 First, consider an arbitrary  party $i$ and a fixed (but otherwise arbitrary)  input value $x_i$.
   By hypothesis,  the effects   $    \left\{   m^{x_i}_{y_i}\right\}_{y_i\in\mathsf Y_i}   $ are orthogonal, which means that there exists a set of  states $    \left\{  \rho^{x_i}_{y_i}\right\}_{y_i\in\mathsf Y_i}   \subset  \St (S_i)$  such that
 \begin{align}\label{aaa}
  \left(   m^{x_i}_{y_i}  |    \rho^{x_i}_{y'_i}\right)  =  \delta_{y_i,y_i'} \qquad \forall y_i,y_i' \in\set Y_i.
  \end{align}
 Now, define the product states
 \begin{align*}
 \left| \rho^{\st x}_{\st y}  \right)    &:  =  \left | \rho^{ x_1}_{y_1}  \right ) \dots   \left | \rho^{ x_N}_{y_N}  \right )
 \end{align*}
  and the product effects
 \begin{align*}
 \left ( m^{\st x}_{\st y}  \right|   & :  =  \left ( m^{x_1}_{y_1}  \right | \dots   \left ( m^{ x_N}_{y_N}  \right |  \, .
 \end{align*}
By the Locality of Pure Orthogonal Measurements, the effects $m^{\st x}_{\st y}$ are pure.
With this definition, if two events $(\st x,\st y)$ and $(\st x',\st y')$ are locally orthogonal, then one has
\begin{align*}
\left ( m^{\st x}_{\st y}  | \rho^{\st x'}_{\st y'}  \right)      =  0.
\end{align*}
By definition, this means that the pure  effects  $ \left \{ m^{\st x}_{\st y}  \right\}_{(\st x,\st y)\in  \set O} $  are orthogonal.  Invoking Sufficient Orthogonality,  we have  that there exists a measurement $ \st m $ such that    $ \left \{ m^{\st x}_{\st y}  \right\}_{(\st x,\st y)\in  \set O} \subseteq \st m $.  Using this fact we obtain
\begin{align*}
\sum_{(\st x,\st y)  \in\set O}    p(\st y|\st x)  &  =   \sum_{(\st x,\st y)  \in\set O}     \left (   m_{\st y}^{\st x}  |  \rho\right)  \\
  &  \le   \sum_{\st e  \in\widetilde{\set O}}     \left (   m_{  \st e }  |  \rho\right)  \\
  &  =1 \, ,
 \end{align*}
 where $\widetilde {\set O} $ denotes the set of all outcomes of the measurement $\st m$.  
 The above inequality concludes the proof of LO in the case when each party performs a pure orthogonal measurement on one subsystem of a composite system. 
  The  argument can be easily extended to prove LO at every level: in this case, one has  simply to replace $\st x$ and $\st y$ with  the strings $\st x_1\st x_2\dots\st x_k$ and $\st y_1\st y_2\dots\st y_k$, respectively.
 \end{proof}

\subsection{Local Orthogonality for generic measurements}
Having derived LO for pure orthogonal measurements, it is easy to extend the derivation to arbitrary measurements. The strategy is to extend the proof first to spiky measurements (by coarse-graining) and then to arbitrary measurements (by measurement purification). The first step is achieved by the following

\begin{lemma}\label{lem:multicoarse}
Let $  p(\st z|\st x)$  be a conditional probability distribution of the variable  $\st z \in\prod_{i=1}^N \set Z_i$  conditional to the variable $\st x\in \prod_{i=1}^N  \set X_i$.    
Let  $p(\st y|\st x)$ be the probability distribution resulting from local coarse-grainings of $p(\st z|\st x)$, that is,
\begin{align*}
p(\st y|\st x)   =     \sum_{\st z\in\prod_{i=1}^N \set Z_{y_i}}   p( \st z|\st x)  \qquad \forall \st y  \in  \prod_{i=1}^N  \set Y_i,
\end{align*}
where, for each $i$, $\{\set Z_{y_i}\}_{y_i\in\set Y_i}$ is a partition of $\set Z_i$ into disjoint subsets.    If the distributions $p(\st  z| \st x)$ satisfy LO, then    also  coarse-grained distributions $p(  \st y| \st x)$ satisfy LO.
\end{lemma}
We omit the proof of the lemma, which can be found in Section V of Ref. \cite{Sainz14exploring}.  

  An immediate corollary is the following:
\begin{cor}\label{corlo}
Let   $p(\st y|\st x)  =    \left (    m^{\st x}_{\st y} | \rho \right)$ be a set of probability distributions defined as in Eq.~\eqref{physprob1} with a set of spiky measurements.    If the theory satisfies  Locality of Pure Orthogonal Measurements and Sufficient Orthogonality,  then $p(\st y|\st x)$ satisfies LO at all levels of the hierarchy.
\end{cor}

Combining this observation with Measurement Purification, one can prove the desired result:
\begin{theorem}\label{theo:lo}
Every theory that satisfies Axioms \ref{mpuri}, \ref{loca}, and \ref{suffo} must  satisfy LO   at every level of the hierarchy.
\end{theorem}

\begin{proof}  Let us start from the first level of the  hierarchy.
Let $p(\st y|\st x)$ be an arbitrary probability distribution arising from local measurements $\st m^{x_i}$ as in Eq. (\ref{physprob1}).
For  every  party $i$, use lemma \ref{piccolo} to  represent the measurement $\st m^{x_i}$  as
\begin{align*}
\left  (m^{x_i}_{y_i}  \right|       =    \left (  M^{x_i}_{y_i} \right|     [   \map I_{S}  \otimes |\sigma_{i})  ]
\end{align*}
for some spiky  measurement $    \st M^{x_i}  $ on $S_i \otimes E_{i}$ and for some  state $\sigma_{i}  \in\St (E_i)$.
Now, by construction the conditional probability distribution is equal to
  \begin{align*}
p(\st y| \st x)    =   \left (   M^{\st x}_{\st y}  |   \sigma \right)
  \end{align*}
  where
  \begin{align*}
  \left (M^{\st x}_{\st y}  \right|   & :=        \left (M^{x_1}_{ y_1} \right|   \dots     \left (M^{ x_N}_{ y_N} \right | \\
    |\sigma  ) &   :  =  |\rho )  | \sigma_1  )  \dots |\sigma_N)
  \end{align*}
  (with a little abuse of notation, consisting in the fact that the systems are ordered as  $S_1E_1S_2E_2\cdots S_N E_N$ in the expression  of the effect $M^{\st x}_{\st y}$ and as \newline $S_1S_2\dots S_N  E_1E_2,\dots  E_N$ in the expression of the state $\sigma$).
By Corollary \ref{corlo} we conclude that $p(\st y|\st x)$ satisfies LO at every level of the hierarchy.

\end{proof}

\subsection{Deriving Causality}

Since LO implies No-Signalling \cite{Fritz12,acin2012}, we have just shown that
every   probability distribution generated by measurements in a theory satisfying Axioms 1-3 satisfies No-Signalling.
 Under the minimalistic Assumption \ref{ass:display}, Proposition \ref{prop:causnosig}  tells us that the theory must satisfy Causality:
\begin{cor}
If a theory satisfies Axioms 1-3 and  Assumption \ref{ass:display}, then the theory satisfies Causality.
\end{cor}

The fact that Causality follows from the axioms, rather than being assumed from the outset is a pretty remarkable fact.  Up to now, the only axiomatization of quantum theory that does not assume Causality from the outset  is Hardy's  2011 axiomatization  \cite{hardy11}. There, Causality is derived from an axiom called Sharpness, stating  that for every pure state there exists a unique effect that gives probability one on that state   and only on that state  \footnote{ Note that the Sharpness axiom by Hardy is slightly different from the Sharpness axiom used by Wilce in Ref.~\cite{wilce2012conjugates}}. 

It is worth stressing that, as per today, 
 only a few works  acknowledge Causality explicitly as an axiom \cite{puri,deri,hardy11,coecke2013causal,Chiribella14,coecke2014terminality,henson2014theory}, while
   most  works   assume Causality implicitly  as part of  the framework\footnote{For example,   Causality enters in the convex set framework \cite{barnum11}   in the moment when measurements are defined as decompositions of the order unit. }, see~e.~g.~\cite{barrett07, barnum11, dakic11, masanes11, masanes12, Pfister}. Recognizing Causality as an axiom is a good starting point to explore deviations from it,  thus developing an operational approach to quantum gravity and indefinite causal structure \cite{hardy2005probability,hardy2007towards,puri,chiribella2013quantum,oreshkov2012quantum,chiribella2011perfect}.

\section{Deriving Sufficient Orthogonality}\label{sec:derSO}

In the previous section we showed that Local Orthogonality and Causality can be obtained from three requirements on the structure of measurements. While the first two requirements (Measurement Purification and Locality of Pure Orthogonal Measurements) are physically well motivated, the third (Sufficient Orthogonality) sounds rather \emph{ad hoc}.
 Can one reduce it  to some other, better motivated axiom?   In this section we give a possible answer, which however, requires us to \emph{assume} Causality.

\subsection{No Disturbance Without Information}

Informally, the \emph{No Disturbance Without Information (NDWI)}  principle  states that if a measurement extracts no information about  a source, then the measurement can be implemented without disturbing the states in that source.    NDWI  appeared originally in the axiomatization work of Ref. \cite{deri}, where it was obtained as a consequence of the axioms  (cf. Corollary 10 of \cite{deri}).   Recently,   NDWI  has been  promoted to the rank of an axiom  by Pfister and Wehner \cite{Pfister}, who showed that  every discrete theory satisfying this requirement  must be classical.

In order to give the precise statement of NDWI, it is useful  to give some definitions.
Here by \emph{source} we mean a deterministic state $\rho$, considered as the average state  of an ensemble of signal states.  A   \emph{state in the source $\rho$} is a state that can be contained in a convex decomposition of $\rho$:
\begin{df}
Let $\rho$ and $\tau$ be two deterministic states of system $S$. We say that $\tau$ is \emph{in the source $\rho$} iff there exists a nonzero probability $p>  0$ and a state $\tau'  \in\St_1 (S)$ such that  $\rho  =   p  \tau  +  (1-p)\tau'$.
\end{df}
With this definition, a non-informative measurement is one  that gives the same statistics for all possible states in the source:
\begin{df}[Non-informative measurements]
Let $\st m$ be a measurement on system $S$, with outcomes in the set $\set Y$.   We say that the measurement $\st m$ \emph{does not extract information}   about the source $\rho\in\St_1 (S)$  iff  there exists a set of probabilities $\{p_y\}_{y\in\set Y}$ such that
\begin{align*}
\left(  m_y  |  \tau\right)  =    p_y  \qquad \quad   \forall y\in\set Y
\end{align*}
for every state $\tau$ in the source.
\end{df}
In other words, a measurement extracts no information about the state $\rho $  iff the probability of the outcome $y$ is the same for every  state in a convex  decomposition of $\rho$.   In Ref. \cite{Pfister} Pfister and Wehner   consider  the special case of non-informative measurements where the measurement gives an outcome with certainty, i.e. $p_{y_0}  = 1$ for a particular outcome $y_0$.

Let us specify what it means to realize a measurement without disturbing the states in a source:
 \begin{df}[Realization of a measurement]
Let $  \boldsymbol{\map T}  = \{  \map T_y\}_{\in \set Y}$ be a test of type $A\to B$ and let $u\in\Eff_1  (B)$ be a deterministic effect on system $B$.  The pair $(\boldsymbol{\map T},  u)$ is a \emph{realization of the measurement $\st m= \{m_y\}_{y\in\set Y}$} iff one has
\begin{align}\label{ndwi}
(  m_y  |   =    (u|  \map T_y  \qquad \forall y\in\set Y    \, .
\end{align}
\end{df}
\begin{df}[Non-disturbing test]
  A test  $\boldsymbol{\map T} $ of type $S\to S$ is \emph{non-disturbing} for the source $\rho$ iff  one has
\begin{align*}
 \sum_{y \in\set Y}    \map T_y   ~ |\tau  )  =      |\tau) , 
  \end{align*}
 for every state $\tau$ in the source.
 \end{df}

\begin{df}[Non-disturbing realization]
A measurement  $\st m$ on system $S$ admits a \emph{non-disturbing realization} for the source $\rho$ iff there exists a realization  of $\st m$, call it $(\boldsymbol {\map T},  u)$, such that $\boldsymbol {\map T}$ is non-disturbing for the source $\rho$. 
\end{df}

Using the above definitions, we can  give the precise statement of  NDWI:
  \begin{df}
 A theory satisfies \emph{No Disturbance Without Information (NDWI)} iff  every measurement $\st m$ that does not extract information  about the source $\rho$  has a  realization  $(\boldsymbol {\map T},  u)$ that is non-disturbing for this source. 
  \end{df}

\subsection{From Causality and NDWI to the joint distinguishability of orthogonal states}

Here we show that Causality and NDWI  imply that orthogonal states can be perfectly distinguished.  Although this fact may sound obvious (it is trivially true in Quantum Theory), its validity is far from obvious in a general physical theory.


\begin{theorem}[Orthogonal states are perfectly distinguishable]\label{theo:distinguishable}
In a convex theory \footnote{A ``convex theory" is defined as a theory where all the set of states, effects, and transformations are convex. So far, we never assumed convexity, and indeed such assumption is not part of the basic framework of OPTs.   Even in the present theorem, we will use convexity only in a minor way,  just to guarantee that we can mix with non-zero probabilities the states in a given set.  }  satisfying Causality and NDWI  orthogonal states  are  perfectly distinguishable.
\end{theorem}

\begin{proof}
Let $\{\rho_y\}_{y\in\set Y}$ be a set of orthogonal states.  By definition \ref{def:biort}, there exists  a set of  effects  $ \{m_y\}_{y\in \set Y} $ such that $(m_y|\rho_{y'})  = \delta_{y,y'}$ for every $y,y'\in\set Y$.  As a consequence, the measurement $ \st  m^{(y)} =  \{m_y, m_{ y}^\perp\}$, $m_{y}^\perp  : =  u-m_y$ does not extract information about  the source
\begin{align*}
\rho_y^\perp  :  =  \frac 1{ |\set Y|   -1}  \sum_{y'\not = y}  \rho_{y'}  \, .
\end{align*}
Indeed, the relations $(m_y|  \rho^\perp_{y})=0$  and $(m_y^\perp|  \rho_y^\perp)=1$ imply analog  relations
\begin{align}
\nonumber (  m_y  |  \tau)  &= 0 \\
\label{bastissima} (m_y^\perp|\tau) & = 1
\end{align}
for every state $\tau$ in the source $\rho_y^\perp$.

By the NDWI axiom, $\st m^{(y)}$ has a non-disturbing realization, given by two transformations  $\{\map T_{y},  \map T_y^\perp\}$ such that
\begin{align}
\nonumber (u|  \map T_y &  =  m_y \\
  (u|  \map T_y^\perp &  =  m_y^\perp \label{bastaa}
\end{align}
and $\left( \map T_y  +  \map T_y^\perp   \right)    |\tau)  =  |\tau)$ for every state $\tau$ in the source $\rho_y^\perp$.  In particular, we have
\begin{align}\label{basta}
\left( \map T_y  +  \map T_y^\perp   \right)  |\rho_{y'})   =   |\rho_{y'})  \qquad \forall y' \not =   y \, .
\end{align}
Applying the unit effect on both sides of Eq. (\ref{basta}) and using Eqs. (\ref{bastaa}) and (\ref{bastissima}) one obtains  $(u|    \map T_y |  \rho_{y'})  =  0$. By Causality, this relation implies $\map T_y  |\rho_{y'})   =    0$ \footnote{Causality implies that for every effect $a\in\Eff  (S)$, the two effects $\{a,u-a\}$ form a legitimate measurement.   Hence, the condition $(  u|   \rho)=  0$ implies $(a|\rho)  =  0$ for every effect,~i.~e.~$\rho=0$.
}.    Hence, Eq. (\ref{basta}) becomes
\begin{align}\label{aaaa}
\map T_y  |\rho_{y'})   &=     0  \\ 
\label{bbb}\map T_y^\perp  |\rho_{y'})& =     |\rho_{y'}  )      \qquad  \forall y'\in \set Y,y'\not = y
\end{align}
Note also that by construction we have $ (u|   \map T_y^\perp |  \rho_y)   =  (m_y^\perp|\rho_y)   =  0$, which implies
\begin{align}\label{ccc}
\map T_y^\perp |\rho_y  )  =   0 \, .
\end{align}
In addition,  we can also assume without loss of generality
\begin{align}\label{ddd}
\map T_{y}  |\rho_{y})& =     |\rho_{y}  )    \, ,
\end{align}
because we can always replace $\map T_y$ with $\map T_y'  :  =   |\rho_y)(u|  \map T_y $  \footnote{  The existence of the transformation $\map T_y'$ is guaranteed by Causality along with Assumption \ref{ass:display2}. Indeed,  one can perform the test $\{\map T_y,\map T_y^\perp\}$ and, conditionally on outcome  $y$, re-prepare the state $\rho_y$.}.

Summarizing Eqs. (\ref{aaaa}), (\ref{bbb}), (\ref{ccc}) and (\ref{ddd})   we have
\begin{align}
\nonumber
\map T_{y}   |\rho_{y'}  )  &=    \delta_{y,y'}    |\rho_y)  \\
\label{eee}
\map T_{y}^\perp  |  \rho_{y'})  &  =  (1  - \delta_{y,y'}  )    |\rho_{y'})  \qquad \forall y,y'\in\set Y \, .
\end{align}

Now, the test  $\{\map T_{y},  \map T_{y^\perp}\}$ allows one to discriminate between the state $\rho_y$ and all the other states   $\{\rho_{y'}\}_{y'\in\set S , y'\not = y}$  without introducing any disturbance.    Hence, one way to distinguish perfectly the states $\{\rho_y\}_{y\in\set Y}$ is to enumerate the elements of $\set Y$, say $\set Y  =  (y_1,\dots, y_N)$ and to apply the tests $\{ \map T_{y_n},  \map T_{y_n}^\perp\}$ one after the other.
The resulting test, denoted by $\{\map S_y\}_{y\in\set Y}$ will consist of the transformations
\begin{align*}
\nonumber \map S_{y_1}   & :=  \map T_{y_1} \\
\nonumber \map S_{y_2}   &:=  \map T_{y_2}  \map T_{y_1}^\perp   \\
\nonumber \map S_{y_3}       &: =  \map T_{y_3} \map T_{y_2}^\perp  \map T_{y_1}^\perp  \\
 \nonumber &\, ~~\vdots  \\
\nonumber \map S_{y_{N-1}}    &:=  \map T_{y_{N-1}} \map T_{y_{N-2}}^\perp\dots    \map T_{y_1}^\perp\\
\map S_{y_N}   &:=   \map T_{y_{N-1}}^\perp \map T_{y_{N-2}}^\perp\dots    \map T_{y_1}^\perp
\end{align*}
Clearly, Eq. (\ref{eee})  implies
$   \map S_{y_m}    |\rho_{y_n})    =    \delta_{m,n}    \,  |\rho_{y_n}) $.
Hence, the  states $\{\rho_y\}_{y\in\set Y}$ can be perfectly distinguished using the measurement $ \st m$ defined by $  ( m_y|  :  =   ( u|   \map S_y $, $\forall y\in\set Y$.

\end{proof}

Note the difference between the statement theorem \ref{theo:distinguishable} and the  statement that a set of pairwise distinguishable states are jointly distinguishable.    As we already observed, orthogonality [cf. definition (\ref{def:biort})] and pairwise distinguishability are different notions.   In general, it is not clear whether the joint distinguishability of pairwise distinguishable states follows from NDWI.

\subsection{Strong No Disturbance Without Information}

We now present a strengthened version of the NDWI axiom,  stating that, in addition to not disturbing the states in a given source, a non-informative measurement does not disturb the pure effects that occur with unit probability on those states:

\begin{df}[Strongly non-disturbing test]
 The test  $\boldsymbol{\map T}  =\{\map T_y\}_{y\in\set Y}$ is \emph{strongly non-disturbing} for the  source $\rho\in\St_1(S)$ iff
\begin{align*}
  \sum_{y \in\set Y}    \map T_y   ~ |\tau  )  =      |\tau)
  \end{align*}
for every  state   $\tau$ in the source  and
\begin{align*}
  \sum_{y \in\set Y}   (a|   \map T_y     =      (a|
  \end{align*}
  for every pure effect $  a$ such that $(a|\tau) =1$.
\end{df}
A strongly non-disturbing realization of a measurement is defined in the obvious way, as a realization in terms of a strongly non-disturbing test.     Using this definition, we can now state the strong version of the NDWI principle:

\addtocounter{axiom}{-1}
\renewcommand{\theaxiom}{\arabic{axiom}$'$}
\begin{axiom}[StrongNDWI]
Every measurement  that does not extract information about the  source $\rho$ has a strongly non-disturbing realization for this source.
  \end{axiom}
\renewcommand{\theaxiom}{\arabic{axiom}}

\subsection{Derivation of SO}
It is easy to prove that  Causality and StrongNDWI  imply Sufficient Orthogonality:
\begin{theorem}\label{theo:SO}
Every convex theory satisfying Causality and StrongNDWI must satisfy SO.
\end{theorem}

\begin{proof}
Let $  \{a_y\}_{y\in\set Y}$ be a set of  pure orthogonal effects and let $ \{\rho_y\}_{y\in\set Y} $ a  set of states biorthogonal to it.   For these two sets, we follow  the construction of theorem \ref{theo:distinguishable}:
we consider the measurement $ \st  m^{y} =  \{a_y, a_y^\perp\}$, $a_y^\perp  : =  u-a_y$ and note that it is non-informative for the state
$$ \rho_y^\perp     =  \frac 1 {|\set Y|  -1}   \sum_{y'\not  =  y }  \rho_{y'} \, .$$  By the StrongNDWI principle, $\st m^{y}$ will have a strongly non-disturbing realization, given by a binary test  $\{\map T_{y},  \map T_y^\perp\}$ such that
\begin{align}
\label{prima}  (u|  \map T_y &  =  a_y \\
 \nonumber (u|  \map T_y^\perp &  =  a_y^\perp
\end{align}
and
\begin{align*}
(  a_{y'}  |  \map T_{y} +  (  a_{y'}  |  \map T_y^\perp  =  (a_{y'}|  \qquad \forall y'\not = y,
\end{align*}
the last equation coming from the StrongNDWI condition. Now, since $a_{y'}$ is a pure effect, the above equation implies
$(  a_{y'}  |  \map T_{y}  =   p \, (  a_{y'}  |$ and $(  a_{y'}  |  \map T_y^\perp  =    (1-  p) (  a_{y'}  |$ for some probability $p\in [0,1]$.
Now, it is easy to show that $p =0$:  indeed, we have
\begin{align*}
p    &  = p \,  (  a_{y'}  |\rho_{y'})  \\
  &  =   (  a_{y'}  |  \map T_{y}     |\rho_{y'})    \\
  &  \le  ( u |     \map T_{y}     |\rho_{y'})\\
  &  =  (a_y  | \rho_{y'})  \\
  &  = 0  \, .
\end{align*}
Hence, we conclude that
\begin{align}\label{seconda}
( a_{y'}  |   \map T_y^\perp  =   ( a_{y'}|   \qquad \forall y'\not = y \, .
\end{align}

Like in the proof of Theorem \ref{theo:distinguishable}, we now  enumerate the elements of $\set Y$, say $\set Y  =  (y_1,\dots, y_N)$, and  apply the tests $\{  \map T_{y_n},  \map T_{y_n}^\perp\}$ one after the other.
As a result, we obtain a test  $\{\map S_y\}_{y\in\widetilde{\set Y}}$ with outcomes in the set $\widetilde  {\set Y}  :  =  \set Y \cup  \{rest\}$,
defined by
\begin{align*}
\nonumber \map S_{y_1}   & :=  \map T_{y_1} \\
\nonumber \map S_{y_2}   &:=  \map T_{y_2}  \map T_{y_1}^\perp   \\
\nonumber \map S_{y_3}       &: =  \map T_{y_3} \map T_{y_2}^\perp  \map T_{y_1}^\perp  \\
 \nonumber &\, ~~\vdots  \\
\nonumber 
\map S_{y_N}   &:= \map T_{ y_{N} }  \map T_{y_{N-1}^\perp}  \dots    \map T_{y_1^\perp}  \\
 \map S_{rest}   &:= \map T_{ y^\perp_{N} }  \map T_{y_{N-1}^\perp}  \dots    \map T_{y_1^\perp}
\end{align*}
To conclude the proof, we consider the measurement $ \st m$ defined by $  ( m_y|  :  =   ( u|   \map S_y $, $\forall y\in \widetilde{\set Y}$:  for this measurement we have
\begin{align*}
( m_{y_1}|   & =   (u|  \map T_{y_1}  =  (a_{y_1}|    \\
(m_{y_2}  |  & =    (u|   \map T_{y_2}  \map T_{y_1}^\perp  =  (a_{y_2}|  \map T_{y_1}^\perp  =  (  a_{y_2}|    \\
(m_{y_3} |  &=   (u|   \map T_{y_3} \map T_{y_2}^\perp  \map T_{y_1}^\perp  =  (  a_{y_3}  |    \map T_{y_2}^\perp  \map T_{y_1}^\perp=  (a_{y_3}|  \\
 \nonumber &\, ~~\vdots  \\
(m_{y_N}|  &=  (u|  \map T_{ y_{N} }  \map T_{y_{N-1}}^\perp  \dots    \map T_{y_1}^\perp  =  (a_{y_N} |  \map T_{y_{N-1}}^\perp  \dots    \map T_{y_1}^\perp  =  (a_{y_N}|  \\
(m_{rest} |  &=  (u| \map T_{ y_{N} }^\perp  \map T_{y_{N-1}}^\perp  \equiv  (u|  -  \sum_{y\in\set Y}  (a_y|  \, ,
\end{align*}
 where we used Eqs. (\ref{prima}) and (\ref{seconda}). Hence, we have proven that an arbitrary set of  pure orthogonal effects $\{ a_y\}_{y\in\set Y}$ can coexist in a single measurement, as required by SO.
 \end{proof}

\section{The device-independent framework for contextuality}\label{sec:CE}

In this section we review the device-independent framework for studying contextuality.   We also review  the principle known as \emph{Consistent Exclusivity (CE)} \cite{kochen,cabello2013,henson2012,acin2012,cabello2012specker}.


\subsection{Contextual games}  
Consider a game featuring   a referee, who  asks  a question  $x\in\set X$, and a player, who responds with an answer $y\in \set Y_x$.   In general, two different questions may have overlapping sets of answers,~i.~e.~one can have $\set Y_{x}  \cap  \set Y_{x'}  \not  =  \emptyset$ for some $x\not = x'$.     At each round of the game, the referee chooses a question  $x$ at random with probability $q(x)$ and assigns a payoff $\omega(x,y)$  to the answer $y$.  The goal of the player is to maximize her expected payoff, given by
\begin{align}\label{payoff1}
\omega  =   \sum_{x}   q(x)   \left[ \sum_{ y}  \omega ( x,  y)   ~ p( y| x)       \right]  \, ,
\end{align}
where $p( y| x)$ is the conditional probability of producing the output  $ y $ upon receiving the input $ x$.  We call a game of the above form a \emph{contextual game}, by analogy with the non-local games discussed before.  

Without further restrictions, the maximization of the payoff is trivial:  for every given  question  $x$, the player only needs  to respond deterministically with  the answer  $y(x)$ that maximizes $\omega(x,y)$.   The problem becomes non-trivial
if the player is forced to assign to each answer $y$ a probability that does not depend on which question---among the questions that have $y$ as possible answer---is asked by the referee.  Mathematically, this amounts to the \emph{response non-contextuality   condition} \footnote{``Non-contextuality" here refers to the fact that the context $x$ that gives rise to an answer  does not influence the probability of its occurrence.     The reader should not confuse the response non-contextuality condition of Eq. (\ref{noncont}) with the statement that quantum mechanics is ``contextual".   The latter is just a shorthand for the statement  that some of the conditional probability distributions $p(y|x)   =  \Tr\left [  P^x_y\,   \rho\right]$, generated by  measurements on quantum states, cannot be reproduced by a \emph{non-contextual ontological  model} \cite{spekkens2005contextuality},~i.~e.~cannot be written in the form  
\begin{align}
p(y|x)   =   \sum_{\lambda  \in \Lambda}     p(\lambda)  \,    r_\lambda(y|x) 
\end{align}
where $\lambda$ is a random variable  with  probability distribution $p(\lambda)$, and, for every $\lambda \in \Lambda$,   $r_\lambda(y|x)$ is a deterministic, non-contextual response function, that is, a conditional probability distribution of the  form  
\begin{align}
r_\lambda  (y|x)    =   \left \{   
\begin{array}{ll}  
1  \qquad &    y  =  f_\lambda (x)\\
0     &  y\not  =   f_\lambda(x)    
\end{array}
\right.
\end{align} 
for some suitable function $f_\lambda:  \set X  \to \set Y$ satisfying the conditions 
\begin{enumerate}
\item \emph{normalization:}  for every question  $x\in\set X$, there exists only one answer $y\in\set Y$ such that $y  = f_\lambda(x)$
\item \emph{non-contextuality:} for every pair of questions $x,x'\in\set X$ and every answer  $y\in  \set Y_x\cap  \set Y_{x'}$, one has $f_\lambda  (x)  =  f_\lambda (x')$.  
\end{enumerate}}  
\begin{align}\label{noncont}
p(y|x)   \equiv   p(y|x')   \qquad \forall  x,x'\in\set X \, ,  \quad \forall y\in\set Y_x\cup \set Y_x' \, .
\end{align}
A strategy satisfying Eq. (\ref{noncont}) is a strategy where  the player partly disregards the question $x$. {\em Partly}, because she will still make use of her knowledge of $x$, by restricting the range of her answers  to the set $\set Y_x$.  
Note that  response non-contextuality   is different from  no-signalling, in that  the choice of question $x$ can affect the  conditional probability distribution  of the answer $y$, albeit in a highly constrained way.

We call a strategy \emph{response-non-contextual} iff it satisfies  the response non-contextuality constraint of Eq. (\ref{noncont}).  
 One way  to enforce  response non-contextuality  is  by imagining that  the game is played a large number of times, allowing the referee to estimate the conditional  probability distribution  $p(y|x)$ and to penalize  deviations from Eq. (\ref{noncont}).    

\subsection{The graph-theoretic framework}
Contextual games can be conveniently cast in a graph-theoretic framework \cite{cabello2010non,acin2012}, which has its roots in the framework   \emph{test spaces}\footnote{Test spaces (modulo minor variations) have been also called \emph{manuals}, {\em spaces}, {\em hypergraphs}, {\em cover spaces}, and {\em generalized sample spaces}, see  \cite{foulis1993logicoalgebraic} for references to the original articles.}  \cite{randall1970approach,foulis1989coupled,foulis1993logicoalgebraic} (see also \cite{wilce2000test,barnum2012post}).

To a given game, one can associate  a hypergraph \footnote{We recall that a hypergraph $\set H  = (\set Y,\set X)$ consists of a collection of vertices $\set Y$, along with a collection $\set X$ of subsets of $\set Y$, called hyperedges \cite{berge1984hypergraphs}.}      $\set H  =(\set Y,\set X)$ as follows:
  \begin{enumerate}
  \item the vertices  are the answers in $\set Y :  =   \bigcup_{x\in \set X}  \set Y_x$
  \item the hyperedges are the questions in $\set X$, with the question $x\in\set X$ being identified with the subset of its possible answers $\set Y_x$ 
  \newline[accordingly, we will  write $y\in x$ in place of $y\in\set Y_x$]. 
   \end{enumerate}

A conditional probability distribution $p(y|x)$ obeying the response non-contextuality condition  (\ref{noncont}) can be completely described by the function $w:  \set Y  \to [0,1]$ defined by 
\begin{align}
w( y):   =     p(y|x)  \qquad \forall x  \in \set X :  \,   y\in x \, .
\end{align}  
The function $w$ is a  \emph{probability weight} on the hypergraph $\set H$, in the following sense:  
\begin{df}\label{def:weight}
A function $w:   \set Y  \to [0,1]$ is a \emph{probability weight} (or a \emph{state}  \cite{gudder1986states,gudder1986logical}) on the hypergraph  $\set H  = (\set Y,\set X)$ iff it satisfies the condition  
\begin{align}\label{dddd}
\sum_{y\in x}    w(y)   = 1  \qquad \forall x\in\set X  \, .
\end{align}  
\end{df}

In terms of the probability weight $w$, the payoff (\ref{payoff1}) can be re-written as 
\begin{align}\label{payoff2}
\omega  =  \sum_{y\in\set Y}   c(y)\,  w(y)   
 \, ,
\end{align}
with $c(y):  =  \sum_{x\in\set X} \,  q(x)  \, \omega  (x,y)$.  
 
The maximization of the payoff over all possible response-non-contextual strategies is then equivalent to the maximization of the payoff over all possible probability weights.    We denote such maximum by  $\omega_{\rm RNC}$. 

\subsection{Principles about probability weights: the example of Consistent Exclusivity}  

The physical limitations affecting the player's strategy will result into constraints on the probability weight $w(y)$. 
Consistent Exclusivity  is one such  requirement: it states that the sum of the probabilities associated to a set of mutually exclusive vertices of the hypergraph $\set H$ should be smaller than 1.  
The precise definitions are given as follows: 
\begin{df}\label{def:exset}
Two vertices $\{y, y'\} \subseteq\set Y$ are \emph{exclusive} iff there exists an hyperedge $x$ such that $\{y,y'\}\subseteq x$.      A subset of outcomes $\set E  \subset \set Y$ is called \emph{mutually exclusive} iff  every pair of distinct outcomes $\{y,y'\} \subseteq \set E$ are exclusive. 
\end{df}

  \begin{df}[Consistent Exclusivity  (\cite{kochen,cabello2013,henson2012,acin2012,cabello2012specker})]\label{def:CE}
  A probability weight $w:  \set Y \to [0,1]$  satisfies Consistent Exclusivity (CE) iff 
  one has 
    \begin{align}\label{ceee}
  \sum_{y\in \set E}   w(y)  \le 1 
  \end{align}  
   for every mutually exclusive set $\set E  \subseteq  \set Y$. 
  \end{df}   

The  above formulation of  CE  is ``device-independent",  in that it makes  reference  to the probability weight $w$, but not to the specific set of measurements $\{\st m^x\}_{x\in\set X}$ that generate $w$. Nevertheless, the request that a probability weight satisfies  CE   is  hard to motivate on physical grounds \footnote{Ref. \cite{cabello2012specker} reports the following comment by Simon Kochen on the background role of CE in the Kochen-Specker paper: ``Ernst and I spent many hours discussing the principle. [...] The difficulty lays in trying to justify it on general physical grounds, without already assuming the Hilbert space formalism of quantum mechanics.".  }. 
To find such motivation, we argue that one should look outside the device-independent framework. 
  This point will be discussed in section \ref{sec:derivingCE}.

\subsection{The CE hierarchy}

Like  LO,   CE can be generalized to an infinite hierarchy of constraints \cite{acin2012}.  The $k$-th level of the hierarchy is defined  by considering $k$ identical copies of a black box, the $i$-th copy   generating an answer $y_i$ with probability weight $w(y_i)$.  In this setting,  the string of answers $\st y  =  (y_1,y_2,\dots, y_k)  \in  \set Y^{\times k}$  has probability weight
\begin{align*}
w^{\otimes k}(\st y) := w(y_1) \,  w(y_2)  \, \dots  \, w(y_k)
\end{align*}
and the $k$-th level of the hierarchy is defined as follows:
\begin{df}[Consistent Exclusivity at the $k$-th level  \cite{acin2012,cabello2013}.]
A  probability weight  $w(y)$ \emph{satisfies CE} at the $k$-th level iff one has
\begin{align}\label{cek}
\sum_{\st y \in  \set E^k}  w^{\otimes k}  ( \st  y)   \le 1 
\end{align}
 for every  mutually exclusive set $\set E_k  \subseteq  \set Y^{\times k}$. 
\end{df}


\subsection{Characterizing the  degree of contextuality of projective quantum measurements}

Since the pioneering work of Kochen and Specker, projective measurements  have played a privileged role in the study of contextuality in quantum mechanics  (see Spekkens \cite{spekkens2005contextuality} for a discussion).  Following this tradition, a number of recent works   (\cite{cabello2010non,acin2012,cabello2013})  have attempted a device-independent 
characterization of the input-output  probability distributions  of the form  
  \begin{align}\label{sharpquantum}
  p(y|x)   =  \Tr  \left[   P^{x}_y\,  \rho \right ] \, ,
  \end{align}
where 
\begin{enumerate}
\item $\rho$ is a quantum state, and
\item  for every $x\in\set X$,  $\st P^x  :  = \{  P^x_y\}_{y\in\set Y_x}$ is a projective quantum measurement satisfying the non-contextuality condition 
\begin{align}
P_y^x  =   P_y^{x'}  \qquad \forall x,x'  \in \set X , \quad \forall y\in\set Y_x\cap  \set Y_{x'} \, .
\end{align} 
\end{enumerate}

We call an input-output probability  distribution of the form  (\ref{sharpquantum})  \emph{projective quantum  (PQ)}.

It is not hard to see that, for given input/output alphabets, the  set of PQ input-output  distributions is convex. Hence,  characterizing it is equivalent to characterizing the maximum payoffs achievable in all possible contextual games.  
Since  the maximum payoff is an indicator of the degree of contextuality, we refer to the problem as ``characterizing the degree of contextuality of projective quantum measurements".   Just like in the case of nonlocality,  finding a device-independent characterization is  a spectacularly hard problem.  CE provides   remarkable results in this direction, but provenly \cite{almost} not a complete characterization.

\section{Physical implementation of contextual games}\label{sec:bridge} 

Like in the case of non-locality, the framework of operational-probabilistic theories can be applied to the study of contextual games.  
  For the physical implementation of a given contextual game, we propose the following model:
\begin{df}
A \emph{physical implementation of a contextual game} is a protocol where
\begin{enumerate}
\item the referee sends to the player a physical system $S$, prepared in a state $\rho$, and a classical  input $x\in \set X$, chosen at random with probability $q(x)$
\item The player performs a measurement $\st m^x =   \{m^x_y\}_{y\in\set Y_x}$  on system $S$ and communicates to the referee the measurement outcome  $y$
\item The referee assigns the payoff $\omega(x,y)$ to the answer.     
\end{enumerate} 
In the implementation of the protocol, the player's measurements are subject to the following constraints 
\begin{enumerate}
\item they should satisfy the  \emph{effect non-contextuality condition} 
\begin{align}\label{noncont1}
m^x_y   =   m^{x'}_y   \qquad  \forall x,x'\in\set X  \, , \,  \forall y\in\set Y_x\cap \set Y_{x'} \, .
\end{align}
\item they have to be performed on the input  system provided by the referee, 
(not on some other system of the same type prepared in the player's laboratory). In other words, the conditional probability distribution of the answer $y$ must satisfy 
\begin{align}\label{probcont}
p(y|x)   =   \left(  \left.  m^x_y  \right  |   \rho  \right) \, . 
\end{align}
\end{enumerate}
\end{df} 
   
   The  above physical implementation departs radically from the device-independent scenario.  This can be observed in the following points:
\begin{enumerate}
\item  While the original game had only a classical input $x$ and a classical output $y$, its physical implementation involves also the communication of a specific physical system  $S$,
known to the player. 
\item  The effect non-contextuality condition (\ref{noncont1}) is  device-dependent.  In order to check its validity, one needs to make a full tomography of the measurement devices $\{\st m^x\}_{x\in\set X}$.   Indeed,  effect non-contextuality  is a stronger condition than response non-contextuality  (\ref{noncont}):   it is equivalent to  response non-contextuality \emph{for every possible state $\rho\in\St_1(S)$}.   
\item  Imagining that the game is played a large number of times, the effect non-contextuality condition   can be enforced by the referee by randomly switching from the ``game-playing mode" to a ``constraint-checking mode", which  consists in sending, instead of the state $\rho$, a state chosen at random from a tomographically complete set of states. By collecting enough statistics, the referee will be able to identify (up to statistical errors) the measurements $\{\st m^x\}_{x\in\set X}$ and to check (up to statistical errors) whether they satisfy the effect non-contextuality condition.    
\item The constraint that the player's measurement are performed on the state $\rho$ can also be checked once a tomographic estimate of the measurements $\{\st m^x\}_{x\in\set X}$ is available. For this purpose, the referee only needs to compare the empirical distribution of the player's answers with the desired distribution $p(y|x)  =    \left(  \left.  m^x_y  \right  |   \rho  \right)$.  
\end{enumerate}

The physical implementation of a contextual game can also be phrased in graph theoretic terms. Given the hypergraph $\set H  =  (\set X,\set Y)$ associated to the original game, the player's strategy is completely specified by the function $\hat w:   \set Y   \to  \Eff(S)$ defined by  
\begin{align}
\hat w(y):   =      m^x_y  \qquad \forall x\in\set X   :    \quad y\in \set x 
\end{align}
[recall that the label $x$ is identified with the subset $\set Y_x \subseteq \set Y$]. 
We refer to the function $\hat w (y)$ as an \emph{effect-valued weight} on the hypergraph $\set H$:  
\begin{df}
A function $\hat w:  \set Y \to \Eff(S)$ is an \emph{effect-valued weight} on the hypergraph $\set H  = (\set Y,\set X)$ iff the collection of effects $\{  \hat w (y)\}_{y\in\set x}$ is a measurement for every $x\in\set X$. 
\end{df}


Given an effect-value weight $\hat w(y)$ and a state  $\rho$, one obtains a probability weight  $w(y)$, defined as  
\begin{align}\label{physicalweight}
w(y):  =   \left( \left.  \hat w (y)   \right|  \rho  \right)   \qquad \forall y\in\set Y  \,. 
\end{align} 
Once a physical theory has been specified, the goal of the player is to find the best measurements that maximize her expected payoff $\omega$.   Among all possible physical implementations with a system $S$ and a state $\rho\in\St_1(S)$, it is interesting to consider the ones that lead to the highest payoff.     For a given theory $\mathbb T$, we denote by $\omega_{\mathbb T}$ the maximum payoff that can be obtained by optimizing over all  systems, states, and measurements.

\section{Reformulating Consistent Exclusivity as a  (device-dependent) physical principle}\label{sec:derivingCE}

In its original formulation, CE is  a principle about probability weights.  To interpret it as a physical principle, one needs to specify what  physical situations  give rise to probability weights satisfying Eq. (\ref{ceee}). 
 This specification, however, is far from straightforward.  The naive  formulation   
\emph{``All the probability weights arising in Nature   satisfy CE"} is ultimately wrong, since  one can easily construct examples of quantum measurements giving rise to probability weights violating the CE property: in fact, for every contextual game, the maximum payoff achievable in quantum mechanics is \emph{equal} to the maximum payoff achievable with arbitrary response-non-contextual strategies\footnote{Trivially, every probability weight $w(y)$ defines a set of  quantum measurements, the $x$-th measurement described by the  POVM  $\st P^x:  = \{ P^x_y \}_{y\in x}$ defined  by  $  P^x_y  :  =    w(y)  \,   I_S$, where $I_S$ is the identity operator on the system's Hilbert space.   For every system  and for every density matrix $\rho$  one then has    $p(y|x)   =  \Tr  \left[   P^x_y\,  \rho\right] \equiv w(y)$.  No matter what the system's dimensionality is,  the maximum of the payoff over all quantum measurements coincides with the maximum over all response-non-contextual strategies.     } , namely  
\begin{align}
\omega_{\rm QUANTUM}   =  \omega_{\rm RNC}     \, .
\end{align}   The fix for this problem is to restrict the validity of CE to probability weights generated by \emph{projective} measurements: the correct condition satisfied by quantum mechanics is   \emph{``All the probability weights arising from projective measurements satisfy the CE property"}. 
Note that this is by no means a device-independent statement, as  it refers explicitly to a  property of the devices used to generate the probability weight.

 In order to formulate  CE  as physical principle, one has first to define the analogue of the ``projective measurements".   This can be done in different ways,  depending on which aspect of projective quantum measurements is chosen as distinctive.         Every definition   will lead to a different ``CE principle", potentially encompassing a different picture of the physical world. For example,   in Ref. \cite{Chiribella14} we proposed a formulation of the CE principle in terms of sharp measurements: 
  \begin{df}[SharpCE]
A theory satisfies \emph{SharpCE  (at the $k$-th level of the hierarchy)} iff every probability weight generated by sharp measurements according to Eq. (\ref{physicalweight}) satisfies CE  (at the $k$-th level of the hierarchy).  
\end{df}    
This formulation of the CE principle has been adopted by Cabello  in Refs. \cite{cabello2014exclusivity,cabello2014simple}, as capturing the intuition at the basis of the formulation of CE in the graph-theoretic framework.   We now explore an alternative formulation, in terms of spiky measurements: 
 \begin{df}[SpikyCE]
A theory satisfies \emph{SpikyCE  (at the $k$-th level of the hierarchy)} iff every probability weight generated by spiky measurements according to Eq. (\ref{physicalweight}) satisfies CE (at the $k$-th level of the hierarchy).  
\end{df}    
Interesting, the two formulations turn out to be equivalent if we restrict our attention to \emph{pure} measurements, because in this case ``sharp" and ``spiky" are equivalent notions.   The equivalence is discussed in Section \ref{sec:relations}, which also lists other alternative generalizations of the notion of projective measurement in quantum theory.   

\section{Deriving Consistent Exclusivity for Spiky Measurements}\label{sec:spikyCE} 

Here  we provide a derivation of SpikyCE from the following three principles:  
\begin{enumerate}
\item Causality 
\item Strong No Disturbance Without Information
\item Pure State  Identification,
\end{enumerate}  
the last  of which will be defined precisely later in this section. 

\subsection{Reduction to Coexistence of Mutually Exclusive Spiky Effects}
Our derivation of SpikyCE proceeds through a sequence of reductions.  The first reduction is based on the following notions:
\begin{df}[Mutually exclusive effects]\label{def:muex}
Two effects $m$ and $m'$ are \emph{exclusive}  iff there exists a measurement $\st m$ such that $ \{m,m'\}\subseteq  \st m$.  A set of effects $\{m_{y}\}_{y\in\set E}$ are \emph{mutually exclusive} iff  every pair of effects $\{m,m'\} \subseteq \{m_{y}\}_{y\in\set E}$ are exclusive. 
\end{df}
\begin{df}[SpikyCMEE]
A theory satisfies \emph{Coexistence of  Mutually Exclusive Spiky  Effects (SpikyCMEE)} iff every set of mutually  exclusive  spiky  effects  can coexist in a measurement. 
\end{df}
SpikyCMEE coincides with  the formulation  of the CE principle used by Barnum, M\"uller, and Ududec in Ref. \cite{barnum2014higher}.   Their choice of name was motivated by the following observation  
\begin{prop}\label{prop:spikyCMEE}
SpikyCMEE implies  SpikyCE. 
\end{prop}  

The proof is elementary and can be found in \ref{app:spikyCMEE}.

\subsection{Deriving SpikyCMEE from Sufficient Orthogonality and Pure State Identification}
We now  reduce SpikyCMEE to Sufficient Orthogonality combined with  a principle of  Pure State Identification. In order to phrase the latter, we need the following definitions: 
\begin{df}
An effect $m  \in\Eff (S) $ is \emph{normalized} iff there exists a state   $\rho\in\St (S)$  such that  $(m|\rho)  = 1$.  
\end{df}
\begin{df}
 Let $m$  and $\varphi$ be an effect and a pure state of system $S$, respectively.  We say that  $m$ \emph{identifies} $\varphi$ iff
 \begin{enumerate}
 \item $(  m|\varphi  )=1$
 \item $  (m|\rho)  <  1$ for every  state $\rho \not =  \varphi$.
 \end{enumerate}
 \end{df}
For example, in Quantum Theory every rank-one projector, considered as a measurement effect, identifies a pure state. 
In a general theory, the fact that every normalized pure effect identifies a state is a nontrivial property   \footnote{Think,~e.~g.~of the square bit, where every pure effect happens with probability 1 on all the states on  one of the four sides of the square  (cf. example \ref{squit}).}.   We refer to it  as  \emph{Pure State Identification}\footnote{A very similar axiom appeared in Hardy's  2011 axiomatization \cite{hardy11,hardy2013reconstructing}, under the name \emph{Logical Sharpness}.  Hardy's axiom is slightly stronger than Pure State Identification,  in that it requires that every pure state is identified by some pure effect.   Another, closely related axiom was put forward by Wilce \cite{wilce2012conjugates}, who  considered a privileged set of measurements  with the property that each outcome identifies a pure state.  The privileged measurements are not assumed to be pure from the start, but  turn out to be so   as a consequence of the axioms.   }:
\begin{axiom}[Pure State Identification]\label{ax:id}
A theory satisfies  {\em Pure State Identification (PSI)}   iff  every normalized  pure effect  identifies a pure state.  \end{axiom}

In a general theory, one has the following
\begin{prop}\label{theo:SpikyCE}
 Pure State Identification  and Sufficient Orthogonality imply  SpikyCMEE. 
\end{prop}
\begin{proof}
We first prove SpikyCMEE  for pure effects.  Let $\{  a_i\}_{i=1}^N$ be a set of mutually exclusive pure effects.   By PSI, each  pure effect $a_i$ identifies a  pure state  $\varphi_i$.   Since the effects are mutually exclusive, for every pair $\{a_i,a_j\} $ there exists a measurement $\st m^{ij}$ such that   $\{a_i,a_j\} \subseteq  \st m^{ij}$.  Hence, the condition $ (a_j|\varphi_j)=1$ implies $(a_i|\varphi_j)=  \delta_{ij}$.  Since $i$ and $j$ are arbitrary, this means that the effects $\{a_i\}_{i=1}^N$ are orthogonal.  SO  then implies  that the effects $\{a_i\}_{i=1}^N$ coexist in a  measurement. This argument proves the validity SpikyCMEE for pure effects.   The extension to arbitrary spiky measurements is immediate, since  spiky measurements are coarse-graining of pure orthogonal measurements. 
\end{proof}

\subsection{Derivation of SpikyCE}
Combining proposition \ref{theo:SpikyCE}   with theorem \ref{theo:SO}, we get the desired result: 
\begin{theorem}  
If a theory satisfies Causality, Strong No Disturbance Without Information, and Pure State Identification, then it also satisfies SpikyCE. 
\end{theorem}
A derivation of SpikyCE from completely different axioms is provided in  Ref. \cite{barnum2014higher},   where SpikyCE is obtained from the requirement  that \emph{i)} every state can be represented as a mixture of perfectly distinguishable pure states and \emph{ii)} all sets of perfectly distinguishable pure states of a given  cardinality can be transformed  into one another by reversible transformations.  It is also interesting to compare Proposition \ref{theo:SpikyCE} with the results of Ref. \cite{Chiribella14}, where we  formulated  SharpCE and derived it from a  \emph{single} axiom about sharp measurements.  Yet another way of deriving CE was found by Wilce \cite{wilcecomm}, who interestingly obtained it from a requirement about \emph{bipartite} systems \footnote{More specifically, Wilce requires the existence of a \emph{conjugate system},   in the sense of \cite{wilce2012conjugates}.   Roughly speaking, a system $S$ is said to have a conjugate $\overline S$ if there exists a suitable state of $S\otimes \overline S$ that exhibits perfect correlations for all measurements in a suitable class of privileged measurements, which we can identify~e.~g.~with the spiky measurements of this article, or with the sharp measurements of Ref. \cite{Chiribella14}.}   and from a requirement about coarse-graining of tests, very similar in spirit to the axiom used in \cite{Chiribella14}.     
The existence of different,  alternative ways to obtain the CE principle   provides a good illustration of the fact that a device-independent feature  can arise from  different  features of the underlying physical theory.

\subsection{Derivation of SpikyCE at all levels of the hierarchy}
Deriving SpikyCE at higher levels of the hierarchy is easy if we assume the Locality of Pure Orthogonal Measurements.   This principle guarantees that, for every pure orthogonal measurement $\st a^x  :  =\{a^x_y\}_{y\in\set Y_x}$, the product measurement  with effects 
\[  a^{\st x}_{\st y}  :  =  a^{x_1}_{y_1}  \otimes a^{x_2}_{y_2}  \otimes \cdots \otimes a^{x_k}_{y_k}\]
is also a pure orthogonal measurement.    The $k$-th level of the hierarchy just follows from the application of the SpikyCMEE. 
 In summary, we have proven the following  
 \begin{cor}  
If a theory satisfies Causality, Locality of Pure Orthogonal Measurements, Strong No Disturbance Without Information, and Pure State Identification, then it also satisfies SpikyCE at all levels of the hierarchy. 
\end{cor} 
\section{Different generalizations of the notion of projective quantum measurement}\label{sec:relations}
We have seen that CE, as a physical principle,  can be formulated in different was, depending on how the notion of ``projective quantum measurement" is generalized  to arbitrary  physical theories.    In this section we discuss four different generalizations and establish a number of relations between them.  

\subsection{Sharp measurements}\label{subsect:sharp}  

\begin{df}[Sharp measurement \cite{Chiribella14}]\label{def:sharp}
A measurement $\st m  =  \{m_x\}_{x\in\set X}$ is \emph{sharp} iff it can be implemented by a repeatable and minimally disturbing test $\boldsymbol  {\map T}   =  \{\map T_x\}_{x\in\set X}$,~i.~e.~a test satisfying the repeatability condition 
\begin{align}
(m_x  |  \map T_x    =   (m_x|   \qquad \forall x\in\set X   
\end{align}
and the minimal disturbance condition
\begin{align}
(n_y|   \left(   \sum_{x\in\set X}   \map T_x \right)   =    (n_y|   \qquad \forall y\in\set Y  \, , 
\end{align}
for every measurement $\st n  =  \{  n_y\}_{y\in\set Y}$ that is compatible\footnote{We recall that two measurements   $\st m=  \{  m_x\}_{x\in\set X}$ and $\st n=  \{n_y\}_{y\in\set Y}$ are said to be \emph{compatible} iff there exists a third measurement $\st o  =\{   o_{z}\}_{z\in\set Z}$  and two disjoint partitions of $\set Z$, denoted by $\left\{\set Z^{\st m}_x\right\}_{x\in\set X}$ and $\left \{  \set Z^{\st n}_y\right\}_{y\in\set Y}$, respectively, such that 
\begin{align*}
m_x    =  \sum_{z\in\set Z^{\st m}_x}  o_{z} \qquad   &\forall x\in\set X  \\
 n_y  =  \sum_{z\in\set Z^{\st n}_y}  o_{z}    \qquad &\forall y\in \set Y \, . 
 \end{align*}}.
 \end{df}
An equivalent characterization of sharp measurements is provided by the following
\begin{prop}\label{prop:sharp}
A measurement   $\st m$ is \emph{sharp} iff there exists a test $\boldsymbol {\map T}$ such that
\begin{align}\label{sharp}
 ( n_{xy} |  \map T_x  =   (n_{xy}|  \qquad\forall x,\in  \set X ,  \,  \forall y\in\set Y
\end{align}
for every measurement  $\{n_{xy}\}_{x\in \set X,  y\in\set Y}$  that refines $\st m$, i.e.   $\sum_y   n_{xy}  =  m_x$.
\end{prop}
The proof can be found in  section III of \cite{Chiribella14}.  Eq. (\ref{sharp}) is closely related to the notion of \emph{coherent L\"uders rule} introduced  by Kleinmann in Ref. \cite{kleinmann2014sequences}.    Roughly speaking, the  sharp measurements of definition \ref{def:sharp}  are the  measurements  that can be implemented by tests in which each transformation is a  coherent L\"uders rule for the corresponding effect \footnote{Some care is required with such an identification,  which sometimes turns out to be incorrect. The main differences   between Refs. \cite{Chiribella14} and \cite{kleinmann2014sequences} can be summarized as follows:  
\begin{enumerate}
\item {\em Framework.}  Ref. \cite{kleinmann2014sequences} associates to a physical system  $S$ an order unit vector space $V_S$, making the following
 \begin{assumption}\label{kleineffect}
  $\Eff (S)  =   \{   m  \in  V_S  ~|~   0\le m  \le u_S  \}$.
  \end{assumption}
  Not every OPT satisfies such an assumption, which is strictly stronger than convexity of the space of effects and is partly related to the so-called  No-Restriction Hypothesis \cite{puri,Janotta13}  (see \ref{app:luders} for  details).   
  \item {\em Positive maps vs physical transformations.}   A \emph{coherent L\"uders rule (CLR)} is defined as a positive linear map $\phi:  V_S \to V_S$ satisfying the conditions   
  \begin{align}
\label{compatible}
\phi  (u_S) &  =    m    \qquad        &  {\rm (} m{\rm-compatibility)}\\
 \label{coherent} 
\phi  ( n)   &=   n  \qquad \forall  n  \in  V_S  \, :    \quad   0\le n\le m   \qquad  & {\rm (coherence)} \, .
\end{align} 
Each positive map is regarded as a potential candidate for a physical transformation, leaving the actual choice of physical transformations open.  
     We argue that the most sensible way to make such a choice is to start from a full OPT, where the composition of transformations in parallel and sequence is built in the operational structure, thanks to the adoption of the categorical framework \cite{abramsky2004,abramsky2008,coecke2010universe}. 
 This allows one to bypass problems like  the difference between positivity and complete positivity, and the problem that 
 the correspondence between positive maps and physical transformations  may not be uniquely defined if the axiom of  Local Tomography is not satisfied  \cite{puri,chiribella14dilation}.  
 \item {\em Sharpness vs coherence.}   The sharpness condition (\ref{sharp})  is generally not equivalent to the coherence condition (\ref{coherent}).  The two conditions become equivalent  under the validity of the following 
 \begin{assumption}\label{cocomp}  
Every two effects $m,n \in  V_S$ satisfying $  n\le m$ are compatible, that is, the three effects $ n,  m-n$ and $u_S  - m$ can coexist in a measurement allowed by the theory.   
\end{assumption} 
Assumption \ref{cocomp} holds for theories satisfying  the Purification axiom (see Corollary 36 of Ref. \cite{puri}) and for theories with a Jordan-algebraic structure \cite{wilce2012conjugates,barnum2014higher}.  Sharpness and coherence are potentially different notions for OPTs that do not satisfy assumption \ref{cocomp}.   
   
\item {\em Effects vs measurements.}   
While Ref. \cite{Chiribella14} focusses on measurements, Ref. \cite{kleinmann2014sequences} focusses on individual effects. As a result,  an effect  with a CLR  may not be a sharp effect (i.~e.~an effect belonging to a sharp measurement):  indeed, an effect $m$ can have a CLR \emph{even if there exists no CLR for the complementary effect} $u_S- m$ \cite{kleincomm}.   
\end{enumerate} 
We have seen that  defining a privileged set of measurements  is important for the study of contextuality.  Hence, one may want define measurements starting from CLRs.     There is a tricky issue here:  the most obvious definition ``CLR measurement $: =$ measurement  $\st m=   \{  m_x\}_{x\in\set X}$ where each effect has a CLR rule  $\phi_x$"  does not have a clear operational meaning, because the collection of maps $\{\phi_x\}_{x\in\set X}$ may not correspond to any test allowed by the theory.  
For this reason,  we suggest  to define ``CR measurement  $:  =$  measurement that can be implemented by a test $\boldsymbol{\map T} =  \{\map T_x\}_{x\in\set X}$ wherein each transformation induces a CLR for the corresponding effect."  
Adopting this definition, we have the following
\begin{prop}\label{prop:CLR}
 Sharp measurements coincide with CLR measurements in causal  OPTs   satisfying Assumptions \ref{kleineffect} and  \ref{cocomp}.
\end{prop}
The proof follows from the discussion presented in \ref{app:luders}, which also contains more details on the relations between sharp measurements and CLR measurements.}.

\subsection{Maximally discriminating measurements}
A third generalization of projective quantum measurements   appeared often  in the literature on the reconstructions of quantum theory  \cite{hardy01,deri,masanes11, hardy11, masanes12, dakic11}. In this context it is often convenient to consider  measurements  that distinguish perfectly among a maximal set of states,~i.~e.~sets of states  $\{\rho_n\}_{n=1}^N$  with the property that  there is no state $\rho_{N+1}$ such that the states $\{\rho_{n}\}_{n=1}^{N+1}$ are jointly distinguishable.
Here we call a measurement that distinguish among a maximal set of states a \emph{maximally discriminating measurement}.
In the case of quantum theory it is easy to see that the maximally discriminating  measurements coincide with the  projective measurements.

\subsection{Measurements consisting of extremal effects}

Yet another possible generalization of projective quantum measurement  is the one in terms of measurements that consist of extremal effects \cite{cabello2010non}, i.e. effects that are extreme points of the set of effects associated to a given system \footnote{This definition presupposes  the mild assumption that such  effects  form a convex set.}.        Note the difference between extremal effects and the notion of pure effects used in this paper:  in quantum theory, the effect $p|  0\rangle\langle 0|$ is pure in our sense, but is not extremal in the convex set of effects, because it is a mixture of the effect $|0 \rangle\langle0|$ with the zero effect.  On the other hand, the projector $ |0\rangle \langle 0  |  +  |1\rangle \langle 1|$ is an extremal effect, but is not pure, because it can be obtained by coarse-graining the pure effects $|0\rangle \langle 0|$ and $|1\rangle \langle 1|$.
 In general, it is easy to see that the extremal effects in quantum theory are the projectors, while the pure effects are the rank-one positive operators upper bounded by the identity matrix.

\subsection{Relations among the different definitions}

A summary of the possible generalizations of projective quantum measurement is as follows:
\begin{enumerate}
\item Maximally discriminating  measurements. Definition based on the  notion of distinguishability of a maximal set of states.
\item Spiky measurements. Definition based on purity and orthogonality. 
\item Sharp measurements. Definition  based on the dynamical features  of the measurement process, which is required to be repeatable and minimally disturbing
\item Measurements consisting of extremal effects. Definition based on the convex structure of the set of effects.
\end{enumerate}

Not much is known about the relations among these four definitions, except for  a few observations that one can readily  make. First,  if a pure measurement is maximal, one has the following implications:
\begin{prop}\label{pureequivalence}
Let $\st m$ be a pure measurement in a causal theory.  Then,
\begin{enumerate}
\item if $\st m$ is maximal, then it is also spiky
\item  $\st m$ is  spiky iff it is sharp
\item if $\st m$ is spiky, then it consists of extremal effects.
\end{enumerate}
\end{prop}
The proof is presented in  \ref{app:a}.    One interesting question here is under which conditions the implication 1 can be reversed, i.e. under which conditions a pure spiky measurement is maximal.  Here is a possible answer: Suppose that
\begin{enumerate}
\item the theory satisfies Pure State Identification, and
\item the set of effects that give probability 1 on a given state has a non-trivial lower bound, that is, for every system $S$ and for every state $\rho  \in\St (S)$, there exists an effect $a_\rho  \in \Eff (S),   \,  a_\rho  \not  = 0$ such that $a_\rho  \le m$ for every effect  $m$ such that
$  (  m|\rho)  =1  \, .$
\end{enumerate}
These two conditions are sufficient to guarantee that all spiky pure measurements are maximal:
\begin{prop}\label{spikymax}
In a causal theory satisfying Conditions 1-2 every spiky pure  measurement is maximal.
\end{prop}
The proof is provided in the  \ref{app:b}.  It is worth stressing  that the simple equivalences presented here hold for \emph{pure} measurements, while the situation is much more involved for generic measurements. This fact prevents a direct comparison of the results of  this paper with those of Ref.  \cite{Chiribella14}, where the main arguments were based  on  the properties of \emph{non-pure} sharp measurements.

\section{Conclusions}\label{sec:conclu}

In this paper we reviewed the  device-independent framework for nonlocality/contextuality and the framework of general probabilistic theories, with the aim of  bridging  the gap between the two approaches.  
We see a high payoff in the  transfer of results from one paradigm to the other.  From the point of view of quantum axiomatizations,        being able to reconstruct a device-independent principle from the axioms provides a direct access to many fundamental features of quantum nonlocality and contextuality.    From the point of view of quantum nonlocality/contextuality,  the approach of general probabilistic theories offers the possibility to find a  deeper understanding of the device-independent features, which may help  overcoming the current difficulties in finding a complete device-independent characterization.  

In this paper we explored both directions.  Following Ref. \cite{Chiribella14}, we focussed on the principles of Local Orthogonality and Consitent Exclusivity and derived them from principles  about the structure of the measurement process.  The derivation presented here differs  significantly from the one presented in Ref. \cite{Chiribella14}, both in the requirements and in the notions used to formulate them. Essentially, the two papers  investigate two different notions of ``ideal measurement", providing two different  and potentially inequivalent generalizations of the notion of projective measurement  in quantum theory.  The two generalizations, called  sharp and spiky measurements, respectively, refer to different operational properties of measurements: repeatability and minimal disturbance for the former, purity and orthogonality in the latter.

How should we interpret the fact that the same device-independent features---LO and CE in this case---can be reduced to two different physical pictures?
Several answers are possible:  On the one hand, one could argue that correlations are only a partial aspect of a physical theory and that, in fact,  it is even possible that two different  physical theories  lead to the same set of  correlations.    In this sense, it is no surprise that inequivalent sets of physical principles entail the same device-independent bounds.         On the other hand, one could argue  that the framework of general probabilistic theories is \emph{too} general, in that  it allows for more  theories than those that are actually worth studying.  In general, the inequivalence of two sets of axioms could be due to some  artificial  and uninteresting counterexample.  Inequivalent axioms could turn out to be equivalent under some reasonable  assumption---the only problem being that the right assumptions have yet to be pinpointed.
As a matter of fact, we believe that both answers contain elements of truth.    

In the present work, a partial   simplification was achieved at the level of pure measurements, where the difference between  sharp and spiky measurements disappears.     It is remarkable that, once more \cite{puri,deri,chiribella14dilation}, bringing  the  analysis to the level of pure processes  simplifies proofs and unites different notions.   This fact could be taken as a clue  that the core of Quantum Theory is encoded in the peculiar interaction between the operational level and an underlying world of pure processes.
From this point of view, the most natural continuation of the  research initiated in this paper  is to combine the Measurement Purification axiom  with the State Purification axiom of Ref. \cite{puri}, seeking for a new axiomatization of Quantum Theory only in terms of Purification features.

\subparagraph*{Acknowledgements}     This work is the expanded version of an earlier paper, entitled ``From Quantum Pictures to Quantum Correlations", presented at QPL 2013, ICFO Barcelona, July 17-19 2013 and published in the conference pre-proceedings. The original paper contained already the definition of spiky measurements---called ``orthogonal'' therein---and the derivation of LO from Measurement Purification, Locality of Pure Orthogonal Measurements, and Sufficient Orthogonality.   
We are indebted to the two anonymous referees of this new expanded version for their extremely careful reading and for providing a number of suggestions that greatly improved the paper. GC is grateful to M Kleinmann for an email discussion on the relation between sharp measurement and coherent L\"uders rules, which has been essential for clarifying the presentation of this subject.

The work is supported   by  the Foundational Questions Institute through the large grant ``The
fundamental principles of information dynamics'' (FQXi-RFP3-1325),   by the National Natural Science Foundation of China through Grants 11450110096 and 11350110207, and by the 1000 Youth Fellowship Program of China.
GC acknowledges R Spekkens and L Masanes for organizing the workshop ``Causal Structure in Quantum Theory" at the Centro de Ciencias de Benasque ``Pedro Pasqual" (2-6/08/2013), that provided an ideal environment for the early elaboration of the ideas in this paper. He also acknowledges the hospitality of the Simons Center for the Theory of Computation and of Perimeter Institute, where part of this work was done.  Research at Perimeter Institute for Theoretical Physics is supported in part by the Government of Canada through NSERC and by the Province of Ontario through MRI.

\appendix

\section{Equivalence between causality and no-signalling}\label{app:causnosig}

\begin{proof}  The implication $1\Longrightarrow 2$ is an immediate consequence  of the uniqueness of the deterministic effect. To prove the implication $2 \Longrightarrow 1$, let us assume that $u_0$ and $u_1$ are two deterministic effects for some system $A$. Consider a the following two-party scenario:  \begin{itemize}
 \item Party 1 holds system $A$ and party $2$ holds system $B$. Systems $A$ and $B$ are prepared  in some joint deterministic state $\sigma  \in \St_1  (A\otimes B)  $
 \item  Party 1 applies either the effect $u_0$ or the effect $u_1$ depending on the value of her input $x_1  \in  \{0,1\}$.  Since both measurements have a single outcome, in both cases the output $y_1$ can take a single value, say $y_1  =  0$
 \item  Party 2 has a single measurement setting, say $x_2  = 0$, which consists in   performing a measurement $\{b_{y_2}\}_{y_2\in \set Y}$ on her system, getting the outcome  $y_2$.
 \end{itemize}
Defining the probability distributions  $$ p( y_1 , y_2  | x_1, x_2 )  :  =  (u_{x_1} \otimes b_{y_2}| \sigma )  \, ,$$
we have that the no-signalling condition becomes
$$    (u_{0} \otimes b_{y_2}| \sigma )   =  (u_{1} \otimes b_{y_2}| \sigma ) \qquad \forall y_2\in\set Y \, ,$$
or, equivalently
$$  (u_{0}| \rho_{y_2} )   =  (u_{1}| \rho_{y_2}  ) \qquad  \forall    y_2\in\set Y  \, ,   $$
where $\rho_{y_2}$ is the state   $\rho_{y_2}  :  =   (  \map I_A  \otimes  b_{y_2})  \sigma$.
Now, Assumption \ref{ass:display} guarantees that every state of system $A$ is of the form $\rho_{y_2}    =   (  \map I_A  \otimes  b_{y_2})  \sigma$ for some suitable state $\sigma$ and some suitable measurement $\{b_{y_2}\}$. Hence, $u_0$ and $u_1$ give the same probability on every input state.  By the very definition of effect (cf. paragraph II.F of Ref. \cite{puri}), this means $u_0  =  u_1$.   \end{proof}

\medskip

\section{Proof of proposition \ref{prop:orthotwo}}\label{app:orthotwo}  

\begin{proof}
Let $m_0$ and $m_1$ be two effects such that $(m_i|  \rho_j)=  \delta_{i,j}$ for arbitrary $i,j,\in\{0,1\}$.    Since every effect belongs to a measurement, there must exist a measurement $\{ \widetilde m_y\}_{y\in\set Y}$ such that $m_0  \equiv  \widetilde m_{y_0}$ for some outcome $y_0\in\set Y$.  By coarse-graining, the measurement  $\{ \widetilde m_y\}_{y\in\set Y}$ can be turned into a binary measurement    $\{ m_0, m_{\lnot 0}\}$, defined by  
\[ m_{\lnot 0}  :=  \sum_{y\in\set Y, \, y\not  =  y_0}    \widetilde m_y  \, .\]
By construction one has $  (m_0|\rho_0)  =  (m_{\lnot 0}|\rho_1)  =  1$ and $(m_0|\rho_1)  =  (m_{\lnot 0}|\rho_0)  = 0$.   In other words, the states $\rho_1$ and $\rho_0$ can be perfectly distinguished using the measurement     $\{ m_0, m_{\lnot 0} \}$.  
 \end{proof}
\medskip

\section{Proof of lemma \ref{piccolo}}\label{app:mpuri}

\begin{proof}
For  every  setting $x$, use the measurement purification axiom to  represent the measurement $\st m^{x}$  as  $   (m^{x}_{y}  |       =    \left (  M^{x}_{y} \right|     [   \map I_{S}  \otimes |\sigma_{x})  ] $ for some  orthogonal measurement $  \st  M^{x}  $ on $S \otimes E_{x}$ and for some state $\sigma_{x}\in \St (E_x)$, where $E_{x}$ is a suitable environment.
Since there is a finite number of settings,  one can always define $E : =   \bigotimes_{x\in\set  X } E_{x}$, $\sigma  : = \bigotimes_{x\in\set X}  \sigma_{x}$ and replace the measurement $\st M^{x}$   with a new spiky  measurement $\st N^{x}$ given by
\begin{align*}
\left( N^{x}_{y} \right| :   =  \left( M^{x }_{y}\right |   \otimes \left [ \bigotimes_{x'\in\set X, x'  \not = x  }  (  u_{x}|   \right] \, ,
\end{align*}
where $u_x$ denotes a unit effect on system $E_x$  (note that, since Causality is not assumed in the hypothesis, the unit effect may not be unique).
In this way, the probability distribution can be expressed as
\begin{align*}
p(y|x)  &  =  \left(  m^{x}_{y}  |  \rho\right)  \\
  &  =  \left(  M^{x}_{y}  |  \rho\otimes \sigma_x\right)  \\
  & =  \left(  N^{x}_{y}  |  \rho\otimes \sigma \right).
\end{align*}
\end{proof}

\section{Violation of SO for polygons with odd number of vertices}\label{app:trigo}
 To prove that the effects $\{  a_y,  a_{y\oplus \frac{n+1} 2}\}$ cannot coexist in a measurement, we show that the sum of their probabilities on the state $|\varphi_{y  \ominus 1})$ exceed one.  Indeed,  define
  \begin{align*}
s &:= (a_y|\varphi_{y\ominus 1}) + (a_{y\oplus \frac{ n+1} 2}|\varphi_{y\ominus 1})  \, .
\end{align*}
Then, Eq. (\ref{forodd}) yields
\begin{align*}
s= \frac{1}{r_n^2+1}\left[r_n^2\cos\left(\frac{2\pi}{n}\right)-r_n^2\cos\left(\frac{3\pi}{n}\right)+2\right] \, .
 \end{align*}
Now, the condition $s  > 1$ is equivalent to
  \begin{align*}
r_n^2\cos\left(\frac{2\pi}{n}\right)-r_n^2\cos\left(\frac{3\pi}{n}\right)- r_n^2+1 >0 \, .
 \end{align*}
Inserting the definition  $r_n:=\sqrt{{1}/{\cos(\pi/n)}}$ into this inequality, one obtains
  \begin{align*}
\cos\left(\frac{2\pi}{n}\right) +\cos\left(\frac{\pi}{n}\right)  >\cos\left(\frac{3\pi}{n}\right) + 1,
 \end{align*}
 which is equivalent to
 \begin{align}\label{ultima}
2\cos\left(\frac{3\pi}{2n}\right)\cos\left(\frac{\pi}{2n}\right)>2\left[\cos\left(\frac{3\pi}{2n}\right)\right]^2 \, ,
 \end{align}
 having used the relation
 $\cos \alpha  +  \cos \beta  =  2  \cos  \frac{ \alpha+  \beta}2  \cos  \frac{ \alpha-  \beta}2 $.   Clearly, the inequality (\ref{ultima}) is satisfied for every $n  \ge 5$.  Hence, all the  polygons  with $n>3$ odd vertices violate SO.
\medskip

\section{Proof of proposition \ref{prop:spikyCMEE}}\label{app:spikyCMEE}
\begin{proof}
 Let  $\set H  =  (\set Y,\set X)$ be the hypergraph associated to a given  contextual game and let $\hat w:  \set Y  \to \Eff (S)$ be the effect-valued weight describing the player's strategy in a  physical implementation of the game. 
Clearly, for every exclusive set of vertices  $\set E\subset \set Y$  (in the sense of definition \ref{def:exset}),  the effects $\{ \hat w (y)\}_{y\in\set E}$ are mutually exclusive (in the sense of  definition \ref{def:muex}). It is also clear that, if the mutually exclusive effects $\{ \hat w (y)\}_{y\in\set E}$ coexist in a single measurement, then one has  the inequality 
\begin{align}
\sum_{y\in\set E}  (\hat w(y)|\rho)    \le   1   \, ,
\end{align}
for every state $\rho \in \St_1 (\rho)$. 
This means that  every probability weight  $w$ generated by the effect-valued weight $\hat w$ satisfies CE. 
\end{proof}

\section{Sharp measurements and coherent L\"uders rules}\label{app:luders}  

In the following we provide a more detailed discussion of the difference between sharp measurements and CLRs. 

\subsection{Framework} 
The framework  of Ref. \cite{kleinmann2014sequences} differs from the OPT framework in a number of significant ways.    Ref. \cite{kleinmann2014sequences} associates 
  a physical system  $S$ with an order unit vector space (OUVS), which we denote  by $V_S$.  The unit in $V_S$ corresponds to the deterministic effect  $u_S$, and every positive element $m\in V_S$  satisfying the condition $ m  \le u_S$ is assumed to be an effect, physically realizable in some test.   In the language of OPTs, this amounts to the assumption 
     \begin{assumption}\label{kleineffect1}
  $\Eff (S)  =   \{   m  \in  V_S  ~|~   0\le m  \le u_S  \}$.
  \end{assumption} 
Such a condition means that 
\begin{enumerate}
\item system $S$ has a unique deterministic effect
\item the set of effects  $\Eff (S)$ is convex
\item effects can be ``scaled up": for every effect $m\in\Eff (S)$   and for every scaling coefficient $\lambda \ge 1$ satisfying $  \lambda \,  m  \le u_S$, one has  that $\lambda\,   m$ belongs to $\Eff  (S)$. 
\end{enumerate} 
 For theories that are not deterministic, the first two conditions can be operationally motivated as  part of the ``Causality package"---in particular, see  Corollary 5 of Ref. \cite{puri} for the convexity of $\Eff(S)$.    Although Causality is a very natural  assumption (one that we also wanted to make in this paper), it is  worth noting that the operational definition of sharp measurement (definition \ref{def:sharp}) in terms of repeatability and minimal disturbance can be applied even in exotic non-causal scenarios, like those arising in Refs.\cite{hardy2005probability,hardy2007towards,puri,chiribella2013quantum,oreshkov2012quantum,chiribella2011perfect}. 

The condition that effects can be scaled up is more specific.   It would follow if one assumed the No Restriction Hypothesis    \cite{puri,Janotta13}, which guarantees both the validity of Assumption (\ref{kleineffect}) and the fact that $\Eff (S)$ is the full dual cone associated to the set of states  $\St  (S)$.  One way to motivate the No-Restriction Hypothesis on operational grounds is provided by Barnum, M\"uller, and Ududec  \cite{barnum2014higher}, who showed that the No-Restriction Hypothesis holds if  \emph{i)} all states of system $S$ can be decomposed into convex combinations of perfectly distinguishable pure states and \emph{ii)}   every two sets of of perfectly distinguishable pure states can be connected by a reversible transformation.  Interesting, one way to bypass the No-Restriction hypothesis is to assume the Purification axiom \cite{puri}, which also guaranteed that effects can be scaled up (for the proof, see Corollary 36 of Ref. \cite{puri}).

\subsection{Positive maps vs physical transformations} A direct comparison between Refs. \cite{Chiribella14} and \cite{kleinmann2014sequences} can be made only for the restricted set of OPTs  that satisfy the conditions 1-3. From now on  we restrict our attention to such theories.

  Given an effect $m\in V_S$, Kleinmann defines a \emph{coherent L\"uders rule  (CLR)}  for $m$  as   a positive linear map $\phi: V_S\to V_S$ satisfying the conditions  
\begin{align}
\label{compatible1}
\phi  (u_S) &  =    m    \qquad        &  {\rm (} m{\rm-compatibility)}\\
 \label{coherent1} 
\phi  ( n)   &=   n  \qquad \forall  n  \in  V_S  \, :    \quad   0\le n\le m   \qquad  & {\rm (coherence)} \, .
\end{align} 
The map $\phi$ is interpreted as a \emph{potential candidate} for a physical transformation. 
  Whether or not a given map $\phi$ really represents a physical transformation, however, is another issue: the definition of CLR only refers to  positive maps satisfying conditions (\ref{compatible1}) and (\ref{coherent1}).    

In order to compare Refs. \cite{Chiribella14} and \cite{kleinmann2014sequences} we need to restrict our attention to those effects $m$ that admit a CLR with the extra property that  the map $\phi$ represents a physical transformation.  
While Ref. \cite{kleinmann2014sequences} does not specify how  physical transformations are defined,  for the sake of comparison we now assume that \emph{some} choice has been made.  In our opinion, the most sensible way to make such a choice is to start from a full OPT, where the composition of transformations in parallel and sequence is built in the operational structure (thanks to the adoption of the categorical framework by Abramsky and Coecke   \cite{abramsky2004,abramsky2008,coecke2010universe}).  The advantage of using the categorical approach with respect to the single-system approach is that in this way one can  bypass problems like 
\begin{enumerate}
\item the difference between positivity and complete positivity (maps that send effects into effect for  individual systems may not do so when applied locally on composite systems), and
\item  the fact that the action of a transformation on a composite system may not be uniquely determined by the linear map $\phi$ (this is often the case when the axiom of Local Tomography is not satisfied \cite{puri,chiribella14dilation}).
\end{enumerate}  

 Let us denote by $V_{S\to S}^{\rm phys}$ the set of positive maps induced by  physical transformations, in the following sense: 
\begin{df}[cf. Eq. (22) of \cite{puri}]
The map $\phi:  V_S \to V_S$ is induced by the physical  transformation $\map T\in\Transf  (S\to S)$ iff for every effect $m  \in  V_S$ one has   $\phi(m)  =  m'$ where $m'$ is the effect defined by $(m'|    :=  ( m|   \map T$.   
\end{df}
 A  \emph{physical} CLR is one where the map $\phi$ belongs to $V_{S\to S}^{\rm phys}$.  
 
 \subsection{Sharpness vs coherence} Equipped with the  definition of physical CLR,  we can now compare the coherence condition  of Eq. (\ref{coherent1})  with the sharpness condition of Eq. (\ref{sharp}).   In general, these two condition express different operational requirements:   Kleinmann's coherence condition requires  that  $\phi$ do not disturb al the  effects $n$ that are ``less likely to be triggered"  than $m$, in the following sense  
 \begin{df}
The effect $n$ is \emph{less likely to be triggered than $m$} iff  $n\le m$.
\end{df} 
 Our sharpness condition requires fact that $\phi$ do not disturb all the effects   that are ``compatible with"   $m$.  This means that the three effects 
 \begin{align*}
 m_1   =   n  \, ,   \qquad  m_2   =    m-n   \, ,  \qquad m_3   =   u_S  -  m 
 \end{align*}
 coexist in a measurement. 
Such a condition is stronger than just $n\le m$.  Due to this fact, an effect $m$ may satisfy the sharpness condition   (\ref{sharp}), and still fail to satisfy the coherence condition  (\ref{coherent}).  The two conditions become equivalent under the following
\begin{assumption}\label{chomp}  
Every two effects $m,n \in  V_S$ satisfying $  n\le m$ are compatible. 
\end{assumption}
 At present, we (the authors)  only know that the assumption holds for theories satisfying for theories satisfying the Purification axiom (for the latter, see again Corollary 36 of Ref. \cite{puri}) and for theories with a Jordan-algebraic structure \cite{wilce2012conjugates,barnum2014higher}.

\subsection{Effects vs measurements.}  While we have so far we discussed about individual effects  $m$, it is eventually  interesting to bring the comparison to the level of  measurements.      In this there is some ambiguity, since Ref. \cite{kleinmann2014sequences} does not give an explicit definition of  \emph{CLR measurement}.  One might be tempted to define it  as a measurement  $\st m  =  \{  m_x\}_{x\in\set X}$ for which each  every effect $m_x$ admits a physical CLR  $\phi_x$.  This definition, however, does not have a clear operational meaning, because it is not a priori clear if the collection of maps $\{\phi_x\}_{x\in\set X}$ is a measurement allowed by the theory.    
As per our present knowledge,  such a condition is met by theories satisfying Local Tomography and  Purification,  and possibly for some theories with Jordan algebraic structure. In general, the most reasonable approach is to define a CLR measurement   $\st m= \{ m_x\}_{x\in\set X}$ as a measurement that can be implemented by test $\boldsymbol{\map T}  =  \{\map T_x  \}_{x\in\set X}$ for which each transformation induces  a CLR. 
  When this definition is adopted, Proposition \ref{prop:CLR} follows from the discussion presented in the previous points.

\section{Proof of Proposition \ref{pureequivalence}}\label{app:a}

\begin{proof}  The implication 1 is an immediate  consequence of  the definitions: by definition,  a maximal measurement is orthogonal. Hence, a pure maximal measurement is a pure orthogonal measurement.  By definition, pure orthogonal measurement are a special case of spiky measurements (with ``only one spike per outcome").    The equivalence between spiky and sharp measurements  at point 2 is proven as follows:  Since the refinements of a pure measurement  are trivial, Proposition \ref{prop:sharp} implies that  a pure measurement $\st m  =  \{m_y\}_{y\in\set Y}$ is sharp iff there exists a test $\{\map M_y\}_{y\in\set Y}$ such that
\begin{align}\label{appp}
(m_y|  \map M_y  =   (m_y|  \qquad \forall y\in\set Y \, . 
\end{align}
In other words, a pure measurement is sharp iff it is repeatable.

We first  prove  the implication ``$\st m$ is spiky"  $\Longrightarrow$ ``$\st m$ is sharp".     If $\st m$ is spiky,      then there exists a set of states $\{\rho_y\}$ such that $(m_y|\rho_{y'}) =  \delta_{y,y'}$.  Since the theory satisfies Assumption \ref{ass:display2}, we can define the measure-and-prepare test $\{\map M_y\}$ with $  \map M_y  :  =   |\rho_y) (m_y|$, which, by construction satisfies the condition (\ref{appp}).    Hence, we proved that $\st m$ is sharp. 
Let us prove the converse implication  ``$\st m$ is sharp"  $\Longrightarrow$ ``$\st m$ is spiky". If $\st m$ is sharp, one can take a state $\rho$ such that $(m_y|\rho)  >  0$ for every $y$.   In a convect theory, one can always find one such state by mixing sufficiently many states in the state space of the system  \footnote{In the whole proof, this is the only point invoking convexity.}.  Note that the possibility of mixing states is guaranteed by Assumption \ref{ass:display2}.
Since $\st m$ is a sharp measurement, there exists a test $\boldsymbol{\map M}  = \{\map M_y\}_{y\in\set Y}$ such that $  (m_y|    =   (u|  \map M_y$ for every outcome $y$.   Thus, one can define the state
$$   \rho_y  :  =  \frac{\map M_y  |\rho)}{(m_y|  \rho)}  \, .$$
By construction, one has $(m_y|  \rho_y)  = 1$ for every $y$, which, by the normalization of probabilities implies the orthogonality condition $(m_y|\rho_{y'})  =  \delta_{y,y'}$ for every $y,y'$.  This proves that $\st m$ is a (pure) orthogonal measurement.  Hence, $\st m$ is be spiky.  Finally, we prove the implication at point 3:  ``$\st m$ is spiky" $\Longrightarrow$  ``$\st m$ is extremal".    For a given outcome $y\in\set Y$, suppose that 
\begin{align}\label{bbbb}m_y  =   p  \, n_y   +   (1-p) \, o_y \, ,
\end{align} 
where  $p\in(0,1)$ is a probability and  $n_y$ and $o_y$ are two effects.  Since $\st m$ is orthogonal, there exists a state $\rho_y$ such that $(m_y|\rho_y)  = 1$.  
Hence, Eq. (\ref{bbbb}) implies the relation  $1  = p \, (n_y| \rho_y) + (1-p)  \,  (o_y|\rho_y)$ and, therefore \begin{align}\label{uno}  
(n_y|\rho_y)  =  (o_y|\rho_y)  =  1  \, .  
\end{align}
 Since $m_y$ is pure, Eq. (\ref{bbbb}) implies  $n_y=  \alpha_y   \,  m_y $ and $o_y  =  \beta_y\,  m_y$ for two suitable constants $\alpha_y\ge 0$ and $\beta_y \ge 0$.   Using Eq. (\ref{uno}) one finally obtains  
 \[  1   =   (n_y| \rho_y)   =   \alpha_y   \,  (m_y|\rho_y)   =  \alpha_y  \] 
 and 
 \[  1   =   (o_y| \rho_y)   =   \beta_y   \,  (m_y|\rho_x)   =  \beta_y  \, . \] 
 In conclusion, one has  $n_y  \equiv   o_y  \equiv  m_y$.  This means that the effect $m_y$ is extremal. Since the outcome $y$ is generic, we obtained that the whole measurement $\st m$ is extremal.      
  \end{proof}

\section{Proof of Proposition \ref{spikymax}}\label{app:b}

\begin{proof}
We have to prove that a pure spiky measurement is maximal.  To this purpose, observe  that for pure measurements ``spiky" is synonymous of ``orthogonal".   Now, suppose that $\st a  = \{a_y\}_{y\in\set Y}$ is a pure orthogonal measurement and let $\{\rho_y\}_{y\in\set Y}$ be the set of states such that $(a_y|  \rho_y')  =  \delta_{y,y'}$.   By the Pure State Identification Property, each state $\rho_y$ must be pure---let us denote it as  $\rho_y =:  \varphi_y$.  Now, we prove that the set $\{\varphi_y\}$ is maximal.  Indeed, suppose by absurd  that the states  $\{\varphi_y\}  \cup  \rho$ are perfectly distinguishable for some $\rho$, and let $ \{  m_y\}  \cup  m_\rho $ the measurement that distinguishes among them.
 Since $(m_y|\varphi_y)  =1$, we must have   $m_y  \ge   l_y$, where $l_y  \not =  0$ is the lower bound to the set of effects that have probability 1 on $\varphi_y$ (the existence of the lower bound is guaranteed by Condition 2).
But since   $(a_y|\varphi_y)  =1$, we must also have $a_y  \ge l_y$.  By the purity of $a_y$, this implies $l_y    =     p_y   a_y$, for some probability $p_y> 0$.
This condition  implies the relation
\begin{align*}    (a_y|  \rho)    &=  p_y^{-1}   (l_y|\rho)  \\
&  \le  p_y^{-1}       (m_y |\rho)   \\
   &  =  0   \qquad \forall y\in\set Y \,.
   \end{align*}
 Moreover, since the effects  $\{a_y\}_{y\in\set Y}$ form a measurement, we have $(u|\rho)  =  \sum_y  (a_y|\rho)  =  0$, which implies that $\rho$ is the zero state, $\rho=0$. Hence, the set of pure states $\{\varphi_y\}$ is maximal.
\end{proof}

\section*{References}

\bibliography{apssamp}

\begin{thebibliography}{100}
\expandafter\ifx\csname url\endcsname\relax
  \def\url#1{\texttt{#1}}\fi
\expandafter\ifx\csname urlprefix\endcsname\relax\def\urlprefix{URL }\fi
\expandafter\ifx\csname href\endcsname\relax
  \def\href#1#2{#2} \def\path#1{#1}\fi

\bibitem{Chiribella14}
G.~{Chiribella}, X.~{Yuan}, {Measurement sharpness cuts nonlocality and
  contextuality in every physical theory. }\href
  {http://arxiv.org/abs/1404.3348} {\path{arXiv:1404.3348}}.

\bibitem{EPR}
A.~Einstein, B.~Podolsky, N.~Rosen, Can quantum-mechanical description of
  physical reality be considered complete?, Phys. Rev. 47 (1935) 777--780.
\newblock \href {http://dx.doi.org/10.1103/PhysRev.47.777}
  {\path{doi:10.1103/PhysRev.47.777}}.

\bibitem{bell}
J.~S. Bell, On the Einstein-Podolsky-Rosen Paradox. Physics 1, 195--200 (1964),
  Speakable and Unspeakable in Quantum Mechanics, Cambridge University Press,
  1987.

\bibitem{rastall1985locality}
P.~Rastall, Locality, bell's theorem, and quantum mechanics, Foundations of
  physics 15~(9) (1985) 963--972.

\bibitem{pr95}
S.~Popescu, D.~Rohrlich, Quantum nonlocality as an axiom, Foundations of
  Physics 24~(3) (1994) 379--385.
\newblock \href {http://dx.doi.org/10.1007/BF02058098}
  {\path{doi:10.1007/BF02058098}}.

\bibitem{Dam05}
W.~{van Dam}, {Implausible Consequences of Superstrong Nonlocality}, Natural
  Computing 12~(1) (2013) 9--12.
\newblock \href {http://arxiv.org/abs/quant-ph/0501159}
  {\path{arXiv:quant-ph/0501159}}.

\bibitem{brassard2006}
G.~Brassard, H.~Buhrman, N.~Linden, A.~A. M{\'e}thot, A.~Tapp, F.~Unger, Limit
  on nonlocality in any world in which communication complexity is not trivial,
  Phys. Rev. Lett. 96 (2006) 250401.
\newblock \href {http://dx.doi.org/10.1103/PhysRevLett.96.250401}
  {\path{doi:10.1103/PhysRevLett.96.250401}}.

\bibitem{linden2007}
N.~Linden, S.~Popescu, A.~J. Short, A.~Winter, Quantum nonlocality and beyond:
  Limits from nonlocal computation, Phys. Rev. Lett. 99 (2007) 180502.
\newblock \href {http://dx.doi.org/10.1103/PhysRevLett.99.180502}
  {\path{doi:10.1103/PhysRevLett.99.180502}}.

\bibitem{Paw09}
M.~{Paw{\l}owski}, T.~{Paterek}, D.~{Kaszlikowski}, V.~{Scarani}, A.~{Winter},
  M.~{{\.Z}ukowski}, {Information causality as a physical principle}, Nature
  461 (2009) 1101--1104.
\newblock \href {http://dx.doi.org/10.1038/nature08400}
  {\path{doi:10.1038/nature08400}}.

\bibitem{Navascues09}
M.~{Navascues}, H.~{Wunderlich}, {A glance beyond the quantum model},
  Proceedings of the Royal Society A: Mathematical, Physical and Engineering
  Science 466 (2009) 881--890.
\newblock \href {http://dx.doi.org/10.1098/rspa.2009.0453}
  {\path{doi:10.1098/rspa.2009.0453}}.

\bibitem{Fritz12}
T.~Fritz, A.~B. Sainz, R.~Augusiak, J.~B. Brask, R.~Chaves, A.~Leverrier,
  A.~Ac{\'\i}n, Local orthogonality as a multipartite principle for quantum
  correlations, Nature Communications 4 (2013) 2263.
\newblock \href {http://dx.doi.org/10.1038/ncomms3263}
  {\path{doi:10.1038/ncomms3263}}.

\bibitem{almost}
M.~Navascu{\'e}s, Y.~Guryanova, M.~J. Hoban, A.~Ac{\'\i}n, Almost quantum
  correlations, Nature Communication 6 (2015) 6288.
\newblock \href {http://dx.doi.org/10.1038/ncomms7288}
  {\path{doi:10.1038/ncomms7288}}.

\bibitem{kochen}
S.~Kochen, E.~P. Specker, The problem of hidden variables in quantum mechanics,
  Journal of Mathematics and Mechanics 17~(1) (1967) 59--87.

\bibitem{spekkens2005contextuality}
R.~W. Spekkens, Contextuality for preparations, transformations, and unsharp
  measurements, Phys. Rev. A 71 (2005) 052108.
\newblock \href {http://dx.doi.org/10.1103/PhysRevA.71.052108}
  {\path{doi:10.1103/PhysRevA.71.052108}}.

\bibitem{cabello2013}
A.~Cabello, Simple explanation of the quantum violation of a fundamental
  inequality, Phys. Rev. Lett. 110 (2013) 060402.
\newblock \href {http://dx.doi.org/10.1103/PhysRevLett.110.060402}
  {\path{doi:10.1103/PhysRevLett.110.060402}}.

\bibitem{henson2012}
J.~{Henson}, {Quantum contextuality from a simple principle? }\href
  {http://arxiv.org/abs/1210.5978} {\path{arXiv:1210.5978}}.

\bibitem{acin2012}
A.~Ac\'i­n, T.~Fritz, A.~Leverrier, A.~Sainz, A combinatorial approach to
  nonlocality and contextuality, Communications in Mathematical Physics 334~(2)
  (2015) 533--628.
\newblock \href {http://dx.doi.org/10.1007/s00220-014-2260-1}
  {\path{doi:10.1007/s00220-014-2260-1}}.

\bibitem{cabello2012specker}
A.~{Cabello}, {Specker's fundamental principle of quantum mechanics. }\href
  {http://arxiv.org/abs/1212.1756} {\path{arXiv:1212.1756}}.

\bibitem{henson2015macroscopic}
J.~Henson, A.~B. Sainz, Macroscopic non-contextuality as a principle for almost
  quantum correlations, arXiv preprint arXiv:1501.06052.

\bibitem{deri}
G.~Chiribella, G.~M. D'Ariano, P.~Perinotti, Informational derivation of
  quantum theory, Phys. Rev. A 84 (2011) 012311.
\newblock \href {http://dx.doi.org/10.1103/PhysRevA.84.012311}
  {\path{doi:10.1103/PhysRevA.84.012311}}.

\bibitem{hardy11}
L.~{Hardy}, {Reformulating and Reconstructing Quantum Theory. }\href
  {http://arxiv.org/abs/1104.2066} {\path{arXiv:1104.2066}}.

\bibitem{masanes11}
L.~{Masanes}, M.~P. {M{\"u}ller}, {A derivation of quantum theory from physical
  requirements}, New Journal of Physics 13~(6) (2011) 063001.
\newblock \href {http://dx.doi.org/10.1088/1367-2630/13/6/063001}
  {\path{doi:10.1088/1367-2630/13/6/063001}}.

\bibitem{dakic11}
B.~{Dakic}, C.~{Brukner}, {Quantum Theory and Beyond: Is Entanglement
  Special?}, in: H.~Halvorson (Ed.), Deep Beauty: Understanding the Quantum
  World through Mathematical Innovation, Cambridge University Press, 2011, pp.
  365--392.

\bibitem{hardy01}
L.~{Hardy}, {Quantum Theory From Five Reasonable Axioms. }\href
  {http://arxiv.org/abs/quant-ph/0101012} {\path{arXiv:quant-ph/0101012}}.

\bibitem{dariano}
G.~D'Ariano, {What is Special about Quantum Mechanics?}, in: A.~Bokulich,
  G.~Jaeger (Eds.), Philosophy of Quantum Information and Entanglement,
  Cambridge University Press, 2010, pp. 85--126.

\bibitem{goyal10}
P.~Goyal, K.~H. Knuth, J.~Skilling, Origin of complex quantum amplitudes and
  feynman's rules, Phys. Rev. A 81 (2010) 022109.
\newblock \href {http://dx.doi.org/10.1103/PhysRevA.81.022109}
  {\path{doi:10.1103/PhysRevA.81.022109}}.

\bibitem{masanes12}
L.~{Masanes}, M.~P. {M\"uller}, R.~{Augusiak}, D.~{Perez-Garcia}, {Existence of
  an information unit as a postulate of quantum theory}, Proceedings of the
  National Academy of Science 110 (2013) 16373--16377.
\newblock \href {http://dx.doi.org/10.1073/pnas.1304884110}
  {\path{doi:10.1073/pnas.1304884110}}.

\bibitem{wilce2012conjugates}
A.~{Wilce}, {Conjugates, Filters and Quantum Mechanics. }\href
  {http://arxiv.org/abs/1206.2897} {\path{arXiv:1206.2897}}.

\bibitem{khalfin1985quantum}
L.~A. Khalfin, B.~S. Tsirelson, Quantum and quasi-classical analogs of bell
  inequalities, in: Symposium on the foundations of modern physics, Singapore:
  World Scientific, 1985, pp. 441--460.

\bibitem{scarani}
V.~Scarani, \href{http://www.physics.sk/aps/pub.php?y=2012&pub=aps-12-04}{The
  device-independent outlook on quantum physics}, Acta Physica Slovaca 62
  (2012) pp. 347--409.
\newline\urlprefix\url{http://www.physics.sk/aps/pub.php?y=2012&pub=aps-12-04}

\bibitem{bancal2013device}
J.-D. Bancal, On the device-independent approach to quantum physics: advances
  in quantum nonlocality and multipartite entanglement detection, Ph.D. thesis,
  University of Geneva (2013).
\newblock \href {http://dx.doi.org/10.1007/978-3-319-01183-7}
  {\path{doi:10.1007/978-3-319-01183-7}}.

\bibitem{coecke2010universe}
B.~Coecke, \href{http://dx.doi.org/10.1017/CBO9780511976971.004}{A universe of
  processes and some of its guises}, in: H.~Halvorson (Ed.), Deep Beauty,
  Cambridge University Press, 2011, pp. 129--186, cambridge Books Online.
\newline\urlprefix\url{http://dx.doi.org/10.1017/CBO9780511976971.004}

\bibitem{mauro2}
G.~M. D'Ariano, How to derive the hilbert space formulation of quantum
  mechanics from purely operational axioms, in: Quantum Mechanics: Are There
  Quantum Jumps? - and On the Present Status of Quantum Mechanics, Vol. 844 of
  AIP Conference Proceedings, American Institute of Physics Melville, NY, 2006,
  pp. 101--128.
\newblock \href {http://dx.doi.org/10.1063/1.2219356}
  {\path{doi:10.1063/1.2219356}}.

\bibitem{barrett07}
J.~Barrett, Information processing in generalized probabilistic theories, Phys.
  Rev. A 75 (2007) 032304.
\newblock \href {http://dx.doi.org/10.1103/PhysRevA.75.032304}
  {\path{doi:10.1103/PhysRevA.75.032304}}.

\bibitem{nobroad}
H.~Barnum, J.~Barrett, M.~Leifer, A.~Wilce, Generalized no-broadcasting
  theorem, Phys. Rev. Lett. 99 (2007) 240501.
\newblock \href {http://dx.doi.org/10.1103/PhysRevLett.99.240501}
  {\path{doi:10.1103/PhysRevLett.99.240501}}.

\bibitem{teleportation}
H.~Barnum, J.~Barrett, M.~Leifer, A.~Wilce, Teleportation in general
  probabilistic theories, in: Mathematical Foundations of Information Flow,
  Vol.~71 of Proceedings of Symposia in Applied Mathematics, 2012, pp. 25--48.
\newblock \href {http://dx.doi.org/http://dx.doi.org/10.1090/psapm/071}
  {\path{doi:http://dx.doi.org/10.1090/psapm/071}}.

\bibitem{puri}
G.~Chiribella, G.~M. D'Ariano, P.~Perinotti, Probabilistic theories with
  purification, Phys. Rev. A 81 (2010) 062348.
\newblock \href {http://dx.doi.org/10.1103/PhysRevA.81.062348}
  {\path{doi:10.1103/PhysRevA.81.062348}}.

\bibitem{barnum11}
H.~Barnum, A.~Wilce, Information processing in convex operational theories,
  Electronic Notes in Theoretical Computer Science 270~(1) (2011) 3 -- 15,
  proceedings of the Joint 5th International Workshop on Quantum Physics and
  Logic and 4th Workshop on Developments in Computational Models (QPL/DCM
  2008).
\newblock \href
  {http://dx.doi.org/http://dx.doi.org/10.1016/j.entcs.2011.01.002}
  {\path{doi:http://dx.doi.org/10.1016/j.entcs.2011.01.002}}.

\bibitem{hardy2013}
L.~Hardy, A formalism-local framework for general probabilistic theories,
  including quantum theory, Mathematical Structures in Computer Science 23~(02)
  (2013) 399--440.
\newblock \href {http://dx.doi.org/10.1017/S0960129512000163}
  {\path{doi:10.1017/S0960129512000163}}.

\bibitem{chiribella14dilation}
G.~Chiribella, Dilation of states and processes in operational-probabilistic
  theories, in: B.~Coecke, I.~Hasuo, P.~Panangaden (Eds.), {\rm Proceedings
  11th workshop on} Quantum Physics and Logic, {\rm Kyoto, Japan, 4-6th June
  2014}, Vol. 172 of Electronic Proceedings in Theoretical Computer Science,
  Open Publishing Association, 2014, pp. 1--14.
\newblock \href {http://dx.doi.org/10.4204/EPTCS.172.1}
  {\path{doi:10.4204/EPTCS.172.1}}.

\bibitem{Pfister}
C.~{Pfister}, S.~{Wehner}, {An information-theoretic principle implies that any
  discrete physical theory is classical}, Nature Communications 4 (2013) 1851.
\newblock \href {http://dx.doi.org/10.1038/ncomms2821}
  {\path{doi:10.1038/ncomms2821}}.

\bibitem{barrett2005nosignaling}
J.~Barrett, L.~Hardy, A.~Kent,
  \href{http://link.aps.org/doi/10.1103/PhysRevLett.95.010503}{No signaling and
  quantum key distribution}, Phys. Rev. Lett. 95 (2005) 010503.
\newblock \href {http://dx.doi.org/10.1103/PhysRevLett.95.010503}
  {\path{doi:10.1103/PhysRevLett.95.010503}}.
\newline\urlprefix\url{http://link.aps.org/doi/10.1103/PhysRevLett.95.010503}

\bibitem{acin06}
A.~Ac\'in, N.~Gisin, L.~Masanes, From bell's theorem to secure quantum key
  distribution, Phys. Rev. Lett. 97 (2006) 120405.
\newblock \href {http://dx.doi.org/10.1103/PhysRevLett.97.120405}
  {\path{doi:10.1103/PhysRevLett.97.120405}}.

\bibitem{ekert1991quantum}
A.~K. Ekert, Quantum cryptography based on bell's theorem, Phys. Rev. Lett. 67
  (1991) 661--663.
\newblock \href {http://dx.doi.org/10.1103/PhysRevLett.67.661}
  {\path{doi:10.1103/PhysRevLett.67.661}}.

\bibitem{mayers1998quantum}
D.~Mayers, A.~Yao, Quantum cryptography with imperfect apparatus, in:
  Foundations of Computer Science, 1998. Proceedings. 39th Annual Symposium on,
  IEEE, 1998, pp. 503--509.

\bibitem{acin2007device}
A.~Ac\'in, N.~Brunner, N.~Gisin, S.~Massar, S.~Pironio, V.~Scarani,
  Device-independent security of quantum cryptography against collective
  attacks, Phys. Rev. Lett. 98 (2007) 230501.
\newblock \href {http://dx.doi.org/10.1103/PhysRevLett.98.230501}
  {\path{doi:10.1103/PhysRevLett.98.230501}}.

\bibitem{masanes2011secure}
L.~{Masanes}, S.~{Pironio}, A.~{Ac{\'{\i}}n}, {Secure device-independent
  quantum key distribution with causally independent measurement devices},
  Nature Communications 2 (2011) 238.
\newblock \href {http://dx.doi.org/10.1038/ncomms1244}
  {\path{doi:10.1038/ncomms1244}}.

\bibitem{vazirani2014fully}
U.~Vazirani, T.~Vidick, Fully device-independent quantum key distribution,
  Physical review letters 113~(14) (2014) 140501.

\bibitem{goldwasser1989knowledge}
S.~Goldwasser, S.~Micali, C.~Rackoff, The knowledge complexity of interactive
  proof systems, SIAM J. Comput. 18~(1) (1989) 186--208.
\newblock \href {http://dx.doi.org/10.1137/0218012}
  {\path{doi:10.1137/0218012}}.

\bibitem{brassard2008classical}
G.~Brassard, A.~Broadbent, A.~Methot, E.~Hanggi, S.~Wolf, Classical, quantum
  and non-signalling resources in bipartite games, in: Quantum, Nano and Micro
  Technologies, 2008 Second International Conference on, 2008, pp. 80--89.
\newblock \href {http://dx.doi.org/10.1109/ICQNM.2008.18}
  {\path{doi:10.1109/ICQNM.2008.18}}.

\bibitem{brassard2013classical}
G.~Brassard, A.~Broadbent, A.~A. M{\'e}thot, E.~Hanggi, S.~Wolf, Classical,
  quantum and nonsignalling resources in bipartite games, Theoretical Computer
  Science 486~(0) (2013) 61 -- 72, theory of Quantum Communication Complexity
  and Non-locality.
\newblock \href {http://dx.doi.org/http://dx.doi.org/10.1016/j.tcs.2012.12.017}
  {\path{doi:http://dx.doi.org/10.1016/j.tcs.2012.12.017}}.

\bibitem{raz1998parallel}
R.~Raz, A parallel repetition theorem, SIAM Journal on Computing 27~(3) (1998)
  763--803.
\newblock \href {http://dx.doi.org/10.1137/S0097539795280895}
  {\path{doi:10.1137/S0097539795280895}}.

\bibitem{holenstein2007parallel}
T.~Holenstein, Parallel repetition: Simplifications and the no-signaling case,
  in: Proceedings of the Thirty-ninth Annual ACM Symposium on Theory of
  Computing, STOC '07, ACM, New York, NY, USA, 2007, pp. 411--419.
\newblock \href {http://dx.doi.org/10.1145/1250790.1250852}
  {\path{doi:10.1145/1250790.1250852}}.

\bibitem{ito2009oracularization}
T.~Ito, H.~Kobayashi, K.~Matsumoto, Oracularization and two-prover one-round
  interactive proofs against nonlocal strategies, in: Computational Complexity,
  2009. CCC '09. 24th Annual IEEE Conference on, 2009, pp. 217--228.
\newblock \href {http://dx.doi.org/10.1109/CCC.2009.22}
  {\path{doi:10.1109/CCC.2009.22}}.

\bibitem{ito2010polynomial}
T.~Ito, Polynomial-space approximation of no-signaling provers, in:
  S.~Abramsky, C.~Gavoille, C.~Kirchner, F.~Meyer auf~der Heide, P.~Spirakis
  (Eds.), Automata, Languages and Programming, Vol. 6198 of Lecture Notes in
  Computer Science, Springer Berlin Heidelberg, 2010, pp. 140--151.
\newblock \href {http://dx.doi.org/10.1007/978-3-642-14165-2_13}
  {\path{doi:10.1007/978-3-642-14165-2_13}}.

\bibitem{kalai2014delegate}
Y.~T. Kalai, R.~Raz, R.~D. Rothblum, How to delegate computations: The power of
  no-signaling proofs, in: Proceedings of the 46th Annual ACM Symposium on
  Theory of Computing, STOC '14, ACM, New York, NY, USA, 2014, pp. 485--494.
\newblock \href {http://dx.doi.org/10.1145/2591796.2591809}
  {\path{doi:10.1145/2591796.2591809}}.

\bibitem{chiribella2016quantum}
G.~Chiribella, R.~W. Spekkens (Eds.), Quantum Theory: Informational Foundations
  and Foils, Springer, 2016.
\newblock \href {http://dx.doi.org/10.1007/978-94-017-7303-4}
  {\path{doi:10.1007/978-94-017-7303-4}}.

\bibitem{quantumfromprinciples}
G.~Chiribella, G.~M. D'Ariano, P.~Perinotti, Quantum from principles, in:
  G.~Chiribella, R.~W. Spekkens (Eds.), Quantum Theory: Informational
  Foundations and Foils, Springer Netherlands, Dordrecht, 2016, pp. 171--222.
\newblock \href {http://dx.doi.org/10.1007/978-94-017-7303-4}
  {\path{doi:10.1007/978-94-017-7303-4}}.

\bibitem{nielsen2010quantum}
M.~A. Nielsen, I.~L. Chuang, Quantum computation and quantum information,
  Cambridge University Press, 2010.

\bibitem{mermin93}
N.~D. Mermin, Hidden variables and the two theorems of john bell, Rev. Mod.
  Phys. 65 (1993) 803--815.
\newblock \href {http://dx.doi.org/10.1103/RevModPhys.65.803}
  {\path{doi:10.1103/RevModPhys.65.803}}.

\bibitem{abramsky2004}
S.~Abramsky, B.~Coecke, A categorical semantics of quantum protocols, in: Logic
  in Computer Science, 2004. Proceedings of the 19th Annual IEEE Symposium on,
  IEEE, 2004, pp. 415--425.
\newblock \href {http://dx.doi.org/10.1109/LICS.2004.1319636}
  {\path{doi:10.1109/LICS.2004.1319636}}.

\bibitem{abramsky2008}
S.~Abramsky, B.~Coecke, Categorical quantum mechanics, in: K.~Engesser, D.~M.
  Gabbay, D.~Lehmann (Eds.), Handbook of quantum logic and quantum structures:
  quantum logic, Elsevier, 2008, pp. 261--324.
\newblock \href {http://dx.doi.org/10.1016/B978-0-444-52869-8.50014\-1}
  {\path{doi:10.1016/B978-0-444-52869-8.50014\-1}}.

\bibitem{hardybook}
L.~Hardy, Reconstructing quantum theory, in: G.~Chiribella, R.~W. Spekkens
  (Eds.), Quantum Theory: Informational Foundations and Foils, Springer
  Netherlands, Dordrecht, 2016, pp. 223--248.
\newblock \href {http://dx.doi.org/10.1007/978-94-017-7303-4}
  {\path{doi:10.1007/978-94-017-7303-4}}.

\bibitem{ludwig1968attempt}
G.~Ludwig, \href{http://projecteuclid.org/euclid.cmp/1103839941}{Attempt of an
  axiomatic foundation of quantum mechanics and more general theories. ii},
  Comm. Math. Phys. 4~(5) (1967) 331--348.
\newline\urlprefix\url{http://projecteuclid.org/euclid.cmp/1103839941}

\bibitem{coecke2010}
B.~{Coecke}, {Quantum picturalism}, Contemporary Physics 51 (2010) 59--83.
\newblock \href {http://dx.doi.org/10.1080/00107510903257624}
  {\path{doi:10.1080/00107510903257624}}.

\bibitem{selinger2011survey}
P.~Selinger, A survey of graphical languages for monoidal categories, in:
  B.~Coecke (Ed.), New Structures for Physics, Vol. 813 of Lecture Notes in
  Physics, Springer Berlin Heidelberg, 2011, pp. 289--355.
\newblock \href {http://dx.doi.org/10.1007/978-3-642-12821-9_4}
  {\path{doi:10.1007/978-3-642-12821-9_4}}.

\bibitem{SpekkensToy}
R.~W. Spekkens,
  \href{http://link.aps.org/doi/10.1103/PhysRevA.75.032110}{Evidence for the
  epistemic view of quantum states: A toy theory}, Phys. Rev. A 75 (2007)
  032110.
\newblock \href {http://dx.doi.org/10.1103/PhysRevA.75.032110}
  {\path{doi:10.1103/PhysRevA.75.032110}}.
\newline\urlprefix\url{http://link.aps.org/doi/10.1103/PhysRevA.75.032110}

\bibitem{davies1970operational}
E.~B. Davies, J.~T. Lewis,
  \href{http://projecteuclid.org/euclid.cmp/1103842336}{An operational approach
  to quantum probability}, Comm. Math. Phys. 17~(3) (1970) 239--260.
\newline\urlprefix\url{http://projecteuclid.org/euclid.cmp/1103842336}

\bibitem{coecke2013causal}
B.~Coecke, R.~Lal, Causal categories: Relativistically interacting processes,
  Foundations of Physics 43~(4) (2013) 458--501.
\newblock \href {http://dx.doi.org/10.1007/s10701-012-9646-8}
  {\path{doi:10.1007/s10701-012-9646-8}}.

\bibitem{coecke2014terminality}
B.~Coecke, Terminality implies no-signalling, in: B.~Coecke, I.~Hasuo,
  P.~Panangaden (Eds.), {\rm Proceedings 11th workshop on} Quantum Physics and
  Logic, {\rm Kyoto, Japan, 4-6th June 2014}, Vol. 172 of Electronic
  Proceedings in Theoretical Computer Science, Open Publishing Association,
  2014, pp. 1--14.
\newblock \href {http://dx.doi.org/10.4204/EPTCS.172.1}
  {\path{doi:10.4204/EPTCS.172.1}}.

\bibitem{fritz2014beyond}
T.~{Fritz}, {Beyond Bell's theorem: correlation scenarios}, New Journal of
  Physics 14~(10) (2012) 103001.
\newblock \href {http://dx.doi.org/10.1088/1367-2630/14/10/103001}
  {\path{doi:10.1088/1367-2630/14/10/103001}}.

\bibitem{henson2014theory}
J.~{Henson}, R.~{Lal}, M.~F. {Pusey}, {Theory-independent limits on
  correlations from generalized Bayesian networks}, New Journal of Physics
  16~(11) (2014) 113043.
\newblock \href {http://dx.doi.org/10.1088/1367-2630/16/11/113043}
  {\path{doi:10.1088/1367-2630/16/11/113043}}.

\bibitem{chaves2015information}
R.~{Chaves}, C.~{Majenz}, D.~{Gross}, Information--theoretic implications of
  quantum causal structures, Nature Communications 6 (2015) 5766.
\newblock \href {http://dx.doi.org/10.1038/ncomms6766}
  {\path{doi:10.1038/ncomms6766}}.

\bibitem{barnum2010entropy}
H.~{Barnum}, J.~{Barrett}, L.~{Orloff Clark}, M.~{Leifer}, R.~{Spekkens},
  N.~{Stepanik}, A.~{Wilce}, R.~{Wilke}, {Entropy and information causality in
  general probabilistic theories}, New Journal of Physics 12~(3) (2010) 033024.
\newblock \href {http://dx.doi.org/10.1088/1367-2630/12/3/033024}
  {\path{doi:10.1088/1367-2630/12/3/033024}}.

\bibitem{janotta11}
P.~{Janotta}, C.~{Gogolin}, J.~{Barrett}, N.~{Brunner}, {Limits on nonlocal
  correlations from the structure of the local state space}, New Journal of
  Physics 13~(6) (2011) 063024.
\newblock \href {http://dx.doi.org/10.1088/1367-2630/13/6/063024}
  {\path{doi:10.1088/1367-2630/13/6/063024}}.

\bibitem{brunner2014dimension}
N.~{Brunner}, M.~{Kaplan}, A.~{Leverrier}, P.~{Skrzypczyk}, {Dimension of
  physical systems, information processing, and thermodynamics}, New Journal of
  Physics 16~(12) (2014) 123050.
\newblock \href {http://dx.doi.org/10.1088/1367-2630/16/12/123050}
  {\path{doi:10.1088/1367-2630/16/12/123050}}.

\bibitem{PironBook}
C.~Piron, \href{http://books.google.com.hk/books?id=zhW2AAAAIAAJ}{Foundations
  of quantum physics}, Mathematical Physics Monograph Series, Benjamin-Cummings
  Publishing Company, 1976.
\newline\urlprefix\url{http://books.google.com.hk/books?id=zhW2AAAAIAAJ}

\bibitem{barnum2014higher}
H.~{Barnum}, M.~P. {M{\"u}ller}, C.~{Ududec}, {Higher-order interference and
  single-system postulates characterizing quantum theory}, New Journal of
  Physics 16~(12) (2014) 123029.
\newblock \href {http://dx.doi.org/10.1088/1367-2630/16/12/123029}
  {\path{doi:10.1088/1367-2630/16/12/123029}}.

\bibitem{Luders50}
G.~L{\"u}ders, {\"U}ber die zustands{\"a}nderung durch den messprozess, Annalen
  der Physik 443~(5-8) (1950) 322--328.
\newblock \href {http://dx.doi.org/10.1002/andp.19504430510}
  {\path{doi:10.1002/andp.19504430510}}.

\bibitem{Busch96}
P.~Busch, P.~Lahti, P.~Mittelstaedt,
  \href{http://books.google.com.hk/books?id=1YO9tQ4mFY8C}{The Quantum Theory of
  Measurement}, Springer, 1996.
\newline\urlprefix\url{http://books.google.com.hk/books?id=1YO9tQ4mFY8C}

\bibitem{Holevo11}
A.~S. Holevo, Probabilistic and statistical aspects of quantum theory, Vol.~1,
  Springer, 2011.

\bibitem{heinosaari2011mathematical}
T.~Heinosaari, M.~Ziman, The mathematical language of quantum theory: from
  uncertainty to entanglement, Cambridge University Press, 2011.

\bibitem{chiribella2012quantum}
G.~{Chiribella}, G.~M. {D'Ariano}, P.~{Perinotti}, {Quantum Theory, Namely the
  Pure and Reversible Theory of Information}, Entropy 14 (2012) 1877--1893.
\newblock \href {http://arxiv.org/abs/1209.5533} {\path{arXiv:1209.5533}},
  \href {http://dx.doi.org/10.3390/e14101877} {\path{doi:10.3390/e14101877}}.

\bibitem{ChiribellaXiao13}
G.~Chiribella, X.~Yuan, Quantum theory from quantum information: the
  purification route, Canadian Journal of Physics 91~(6) (2013) 475--478.
\newblock \href {http://dx.doi.org/10.1139/cjp-2012-0472}
  {\path{doi:10.1139/cjp-2012-0472}}.

\bibitem{coecke2008axiomatic}
B.~Coecke, Axiomatic description of mixed states from selinger's
  cpm-construction, Electronic Notes in Theoretical Computer Science 210~(0)
  (2008) 3 -- 13, proceedings of the 4th International Workshop on Quantum
  Programming Languages (QPL 2006).
\newblock \href
  {http://dx.doi.org/http://dx.doi.org/10.1016/j.entcs.2008.04.014}
  {\path{doi:http://dx.doi.org/10.1016/j.entcs.2008.04.014}}.

\bibitem{coecke2010environment}
B.~Coecke, S.~Perdrix, Environment and classical channels in categorical
  quantum mechanics, in: Computer Science Logic, Springer, 2010, pp. 230--244.

\bibitem{selinger2007dagger}
P.~Selinger, Dagger compact closed categories and completely positive maps:
  (extended abstract), Electronic Notes in Theoretical Computer Science 170~(0)
  (2007) 139 -- 163, proceedings of the 3rd International Workshop on Quantum
  Programming Languages (QPL 2005).
\newblock \href
  {http://dx.doi.org/http://dx.doi.org/10.1016/j.entcs.2006.12.018}
  {\path{doi:http://dx.doi.org/10.1016/j.entcs.2006.12.018}}.

\bibitem{stueckelberg1960quantum}
E.~C. Stueckelberg, Quantum theory in real hilbert space, Helv. Phys. Acta
  33~(727) (1960) 458.
\newblock \href {http://dx.doi.org/10.5169/seals-113093}
  {\path{doi:10.5169/seals-113093}}.

\bibitem{hardy2012limited}
L.~Hardy, W.~Wootters, Limited holism and real-vector-space quantum theory,
  Foundations of Physics 42~(3) (2012) 454--473.
\newblock \href {http://dx.doi.org/10.1007/s10701-011-9616-6}
  {\path{doi:10.1007/s10701-011-9616-6}}.

\bibitem{wootters2013optimal}
W.~K. {Wootters}, {Optimal Information Transfer and Real-Vector-Space Quantum
  Theory. }\href {http://arxiv.org/abs/1301.2018} {\path{arXiv:1301.2018}}.

\bibitem{Janotta13}
P.~Janotta, R.~Lal, Generalized probabilistic theories without the
  no-restriction hypothesis, Phys. Rev. A 87 (2013) 052131.
\newblock \href {http://dx.doi.org/10.1103/PhysRevA.87.052131}
  {\path{doi:10.1103/PhysRevA.87.052131}}.

\bibitem{Sainz14exploring}
A.~B. Sainz, T.~Fritz, R.~Augusiak, J.~B. Brask, R.~Chaves, A.~Leverrier,
  A.~Ac\'in, Exploring the local orthogonality principle, Phys. Rev. A 89
  (2014) 032117.
\newblock \href {http://dx.doi.org/10.1103/PhysRevA.89.032117}
  {\path{doi:10.1103/PhysRevA.89.032117}}.

\bibitem{hardy2005probability}
L.~{Hardy}, {Probability Theories with Dynamic Causal Structure: A New
  Framework for Quantum Gravity. }\href {http://arxiv.org/abs/gr-qc/0509120}
  {\path{arXiv:gr-qc/0509120}}.

\bibitem{hardy2007towards}
L.~Hardy, Towards quantum gravity: a framework for probabilistic theories with
  non-fixed causal structure, Journal of Physics A: Mathematical and
  Theoretical 40~(12) (2007) 3081--3099.
\newblock \href {http://dx.doi.org/10.1088/1751-8113/40/12/S12}
  {\path{doi:10.1088/1751-8113/40/12/S12}}.

\bibitem{chiribella2013quantum}
G.~Chiribella, G.~M. D'Ariano, P.~Perinotti, B.~Valiron, Quantum computations
  without definite causal structure, Phys. Rev. A 88 (2013) 022318.
\newblock \href {http://dx.doi.org/10.1103/PhysRevA.88.022318}
  {\path{doi:10.1103/PhysRevA.88.022318}}.

\bibitem{oreshkov2012quantum}
O.~{Oreshkov}, F.~{Costa}, {\v C}.~{Brukner}, {Quantum correlations with no
  causal order}, Nature Communications 3 (2012) 1092.
\newblock \href {http://dx.doi.org/10.1038/ncomms2076}
  {\path{doi:10.1038/ncomms2076}}.

\bibitem{chiribella2011perfect}
G.~Chiribella, Perfect discrimination of no-signalling channels via quantum
  superposition of causal structures, Phys. Rev. A 86 (2012) 040301.
\newblock \href {http://dx.doi.org/10.1103/PhysRevA.86.040301}
  {\path{doi:10.1103/PhysRevA.86.040301}}.

\bibitem{cabello2010non}
A.~{Cabello}, S.~{Severini}, A.~{Winter}, {(Non-)Contextuality of Physical
  Theories as an Axiom. }\href {http://arxiv.org/abs/1010.2163}
  {\path{arXiv:1010.2163}}.

\bibitem{foulis1993logicoalgebraic}
D.~Foulis, R.~Greechie, G.~Rüttimann, Logicoalgebraic structures ii. supports
  in test spaces, International Journal of Theoretical Physics 32~(10) (1993)
  1675--1690.
\newblock \href {http://dx.doi.org/10.1007/BF00979494}
  {\path{doi:10.1007/BF00979494}}.

\bibitem{randall1970approach}
C.~H. Randall, D.~J. Foulis, \href{http://www.jstor.org/stable/2316143}{An
  approach to empirical logic}, The American Mathematical Monthly 77~(4) (1970)
  pp. 363--374.
\newline\urlprefix\url{http://www.jstor.org/stable/2316143}

\bibitem{foulis1989coupled}
D.~Foulis, Coupled physical systems, Foundations of Physics 19~(7) (1989)
  905--922.
\newblock \href {http://dx.doi.org/10.1007/BF01889305}
  {\path{doi:10.1007/BF01889305}}.

\bibitem{wilce2000test}
A.~Wilce, Test spaces and orthoalgebras, in: B.~Coecke, D.~Moore, A.~Wilce
  (Eds.), Current Research in Operational Quantum Logic, Vol. 111 of
  Fundamental Theories of Physics, Springer Netherlands, 2000, pp. 81--114.
\newblock \href {http://dx.doi.org/10.1007/978-94-017-1201-9_4}
  {\path{doi:10.1007/978-94-017-1201-9_4}}.

\bibitem{barnum2012post}
H.~{Barnum}, A.~{Wilce}, {Post-Classical Probability Theory. }\href
  {http://arxiv.org/abs/1205.3833} {\path{arXiv:1205.3833}}.

\bibitem{berge1984hypergraphs}
C.~Berge, Hypergraphs: combinatorics of finite sets, Vol.~45, Elsevier, 1984.

\bibitem{gudder1986states}
S.~Gudder, M.~Kl{\"a}y, G.~R{\"u}ttimann, States on hypergraphs, Demonstratio
  Math 19 (1986) 503--526.

\bibitem{gudder1986logical}
S.~Gudder, \href{http://eudml.org/doc/76342}{Logical cover spaces}, in: Annales
  de l'IHP Physique th{\'e}orique, Vol.~45, Gauthier-villars, 1986, pp.
  327--337.
\newline\urlprefix\url{http://eudml.org/doc/76342}

\bibitem{cabello2014exclusivity}
A.~Cabello, Exclusivity principle and the quantum bound of the bell inequality,
  Phys. Rev. A 90 (2014) 062125.
\newblock \href {http://dx.doi.org/10.1103/PhysRevA.90.062125}
  {\path{doi:10.1103/PhysRevA.90.062125}}.

\bibitem{cabello2014simple}
A.~{Cabello}, A simple explanation of the quantum limits of genuine $ n $-body
  nonlocality.\href {http://arxiv.org/abs/1411.0153} {\path{arXiv:1411.0153}}.

\bibitem{hardy2013reconstructing}
L.~{Hardy}, {Reconstructing quantum theory. }\href
  {http://arxiv.org/abs/1303.1538} {\path{arXiv:1303.1538}}.

\bibitem{wilcecomm}
A.~Wilce, private communication.

\bibitem{kleinmann2014sequences}
M.~{Kleinmann}, {Sequences of projective measurements in generalized
  probabilistic models}, Journal of Physics A Mathematical General 47 (2014)
  5304.
\newblock \href {http://dx.doi.org/10.1088/1751-8113/47/45/455304}
  {\path{doi:10.1088/1751-8113/47/45/455304}}.

\bibitem{kleincomm}
M.~Kleinmann, private communication.

\end{thebibliography}

\end{document}